\documentclass[conference]{IEEEtran}
\IEEEoverridecommandlockouts
\newcommand{\vlong}[1]{#1}
\newcommand{\vshort}[1]{}

\usepackage[sort,compress]{cite}

\usepackage{amsmath,amssymb,amsfonts}
\usepackage{algorithmic}
\usepackage{graphicx}
\usepackage{textcomp}
\usepackage{xcolor}

\usepackage{soul}
\usepackage{tikz-cd}
\usepackage{xspace}
\usepackage[normalem]{ulem}
\usepackage{xcolor}
\usepackage{perfectcut}
\let\cut\perfectcut
\usepackage[final,inline,nomargin]{fixme}
\fxusetheme{color}
\fxuseenvlayout{color}
\FXRegisterAuthor{lc}{alc}{\color{purple}[L]}
\FXRegisterAuthor{ag}{aag}{\color{teal}[A]}
\FXRegisterAuthor{dk}{adk}{\color{magenta}[D]}
\FXRegisterAuthor{em}{aem}{\color{blue}[E]}
\FXRegisterAuthor{rt}{art}{\color{orange}[R]}
\FXRegisterAuthor{rev}{arev}{\color{red}[Reviewer]}

 \usepackage[switch]{lineno}
\RequirePackage{xparse}

\usepackage[colorinlistoftodos,prependcaption,textsize=scriptsize,textwidth=2cm]{todonotes}
\usepackage[inline]{enumitem}
\usepackage{multicol}
\usepackage{mathpartir}
\usepackage{marvosym}
\usepackage{float}
\usepackage{mathbbol}
\usepackage{mathtools}
\usepackage{mathrsfs}
\usepackage{amsmath}
\usepackage{amsthm}
\usepackage{thmtools}
\usepackage{thm-restate}
\usepackage[notext,nomath]{stix}
\usepackage{stmaryrd}
\usepackage{arydshln} % horizontal lines
\usepackage{proof}
\usepackage{hyperref}
\usepackage{cleveref}
\usepackage{wrapfig}
\usepackage{stix}
\usepackage{array}
\usepackage{subcaption}
\usepackage{booktabs}
\usepackage{upgreek}
\usepackage{empheq}
\usepackage{perfectcut}
\let\cut\perfectcut
\usepackage{flafter} % ensures floats never drawn before their reference
\usetikzlibrary{matrix, positioning, fit, arrows, chains,shapes}

\usepackage{mathtools}
\usepackage[bb=boondox]{mathalfa}
\usepackage{proof}
\hypersetup{colorlinks=true}

\newtheorem{theorem}{Theorem}
\newtheorem{proposition}[theorem]{Proposition}
\newtheorem{lemma}[theorem]{Lemma}
\newtheorem{corollary}[theorem]{Corollary}

\newtheorem{property}[theorem]{Property}
\theoremstyle{definition}
\newtheorem{definition}[theorem]{Definition}
\newtheorem{remark}[theorem]{Remark}
\newtheorem{example}[theorem]{Example}

\Crefname{equation}{Eq.}{Eqs.}
\Crefname{figure}{Fig.}{Figs.}
\Crefname{tabular}{Tab.}{Tabs.}
\Crefname{section}{Sec.}{Secs.}
\Crefname{definition}{Def.}{Defs.}
\Crefname{defi}{Def.}{Defs.}
\Crefname{lemma}{Lem.}{Lems.}
\Crefname{lem}{Lem.}{Lems.}
\Crefname{theorem}{Thm.}{Thms.}
\Crefname{thm}{Thm.}{Thms.}
\Crefname{paragraph}{Sec.}{Secs.}
\Crefname{appendix}{Appx.}{Appxs.}
\Crefname{corollary}{Cor.}{Cors.}
\Crefname{example}{Ex.}{Exs.}
\Crefname{proposition}{Prop.}{Props.}
\Crefname{remark}{Rem.}{Rems.}

\let\varmathbb\mathbb

\def\BaseUrl{https://dominik-kirst.github.io/mca/effhol/}
\newcommand{\coqdoc}[1]{\href{\BaseUrl#1}{\raisebox{-.9mm}{\includegraphics[height=1em]{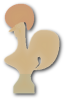}}}}

\DeclareMathAlphabet{\sys}{OT1}{cmss}{sbc}{n}

\newcommand{\hopl}[0]{\sys{EffHOL}}
\newcommand{\hoplinst}{\hopl^{\!\mathbf{-}}}
\newcommand{\hoplk}[0]{\hopl^{\!\neg\neg}}
\newcommand{\effhol}[0]{\hopl}
\newcommand{\hol}[0]{\sys{HOL}}
\newcommand{\ihol}[0]{\hol}
\newcommand{\fomega}{F_\omega}
\newcommand{\feffhol}[0]{Effectful Higher-Order Logic}

\newcommand{\trkind }[1]{\left\llbracket #1 \right\rrbracket^{\mathrm{K}}}

\newcommand{\trtype }[2]{\left\llbracket #2 \right\rrbracket_{#1}^{\mathrm{T}}}

\newcommand{\trspec }[2]{\left\llbracket #2 \right\rrbracket_{#1}^{\mathrm{S}}}

\newcommand{\trind  }[2]{\left\llbracket #2 \right\rrbracket_{#1}^{\mathrm{I}}}

\newcommand{\trform }[1]{\left\llbracket #1 \right\rrbracket}

\newcommand{\ttrkind }[1]{\left\llbracket #1 \right\rrbracket^{\mathrm{K}}}

\newcommand{\ttrtype }[2]{\left\llbracket #2 \right\rrbracket^{\mathrm{T}}}

\newcommand{\ttretype}[3]{\left\llbracket #2 \right\rrbracket^{\mathrm{t}}}
\newcommand{\ttrtrm  }[3]{\left\llbracket #2 \right\rrbracket^{\mathrm{e}}}
\newcommand{\ttrspec }[3]{\left\llbracket #2 \right\rrbracket^{\mathrm{S}}_{#3}}
\newcommand{\ttrind  }[2]{\left\llbracket #2 \right\rrbracket^{\mathrm{I}}_{#1}}
\newcommand{\pred}[1]{\mathsf{P}\left( #1 \right)}

\newcommand{\refpred}[2]{\mathsf{R}_{#1}\left( #2 \right)}
\newcommand{\refpredN}[1]{\mathsf{R}_{#1}}

\newcommand{\eltype}[2]{#1  \vdash  #2}

\newcommand{\eltrm}[3]{#1  \vdash  #2 : #3}

\newcommand{\elprop}[2]{#1  \vdash  #2 }

\newcommand{\sequent}[3]{#1  \vdash  #2 \, \Rightarrow \, #3 }
\newcommand{\bisequent}[3]{#1  \mid  #2 \, \dashv \vdash \, #3 }
\newcommand{\elsequent}[3]{#1  \vdash \! #2  \Rightarrow  #3 }
\newcommand{\selsequent}[3]{#1  \vdash \! #2  \Rrightarrow  #3 }

\newcommand{\etriple}[3]{  \left\{ #1 \right\} #2 \left\{ #3 \right\} }
\newcommand{\ehtriple}[4]{ \etriple{#1}{#2 \longleftarrow #3}{#4} }

\newcommand{\triple}[4]{ #1 \vdash \etriple{#2}{#3}{#4} }
\newcommand{\htriple}[5]{ \triple{#1}{#2}{#3 \longleftarrow #4}{#5} }

\newcommand{\after}[3]{\left\langle #2 \leftarrow #1 \right\rangle \, #3 }

\newcommand{\subst}[1]{\left[ #1 \right]}

\newcommand{\comprehend}[3]{ \left\{ #1 : #2 \; \middle| \; #3 \right\} }

\newcommand{\comp}[2]{ \left\{ #1 \; \middle| \; #2 \right\} }

\newcommand{\tmem}[3]{ #1 ; {#3} \inplus {#2} }
\newcommand{\tmembase}[2]{ #1 \inplus #2}
\newcommand{\smem}[2]{ #1 \!\in\! #2 }
\newcommand{\smembase}[1]{\overline{#1}}
\newcommand{\compbase}[1]{ \left\{ #1 \right\} }

\newcommand{\jprop}[2]{#1 \mid #2 \;\; \mathbf{Prop}}
\newcommand{\jtrm}[3]{#1 \mid #2 : #3}

\newcommand{\myskip}{1.2em}
\newcommand{\sskip}{0.8em}

\newcommand\STAR{\raisebox{-.1em}{\tikz{\node[scale=0.25,draw,star,star point height=.7em,minimum size=0.1em]{};} }}

\newcommand{\eqdef}{\triangleq}

\newcommand{\typecon}[1]{#1 \rightarrowtail \star}
\newcommand{\predcon}[1]{#1 \rightarrowtail \STAR}
\newcommand{\typefun}[2]{#1 \! \rightarrow \! #2}
\newcommand{\typeapp}[2]{#1 \,#2}
\newcommand{\typevar}{X}
\newcommand{\typecomp}[1]{M\left(#1\right)}
\newcommand{\typeabs}[3]{\bar{\Lambda}_{ #1 : #2 }. #3}

\newcommand{\indice}{\sigma}
\newcommand{\spred}{u}
\newcommand{\tpred}{y}
\newcommand{\epred}{y}
\newcommand{\sprop}{\psi}
\newcommand{\tprop}{\varphi}
\newcommand{\sprops}{\Psi}
\newcommand{\tprops}{\Phi}
\newcommand{\sort}{s}
\newcommand{\kind}{\kappa}
\newcommand{\term}{p}
\newcommand{\exprs}{e}
\newcommand{\exprsvar}{y}
\newcommand{\eforall}[2]{\Uplambda_{#1}.#2}
\newcommand{\eapp}[2]{{#1}\,#2}

\newcommand{\sterm}{t}
\newcommand{\termvar}{x}
\newcommand{\termtypeabs}[3]{\Lambda #1 \!:\!#2 . #3}
\newcommand{\termabs}[3]{\lambda #1\!:\!#2 . #3}
\newcommand{\termtypeapp}[2]{#1 \, #2 }
\newcommand{\termapp}[2]{#1 \, #2}
\newcommand{\termret}[1]{\left[#1\right]}
\newcommand{\termbind}[3]{\mathsf{let} \: #1 \leftarrow #2 \: \mathsf{in} \: #3}

\newcommand{\timplies}[2]{#1 \supset #2}

\newcommand{\simplies}[2]{#1 \sqsupset #2}
\newcommand{\sconj}[2]{#1 \sqcap #2}
\newcommand{\tindprod}[2]{\bigwedge_{#1}.#2}
\newcommand{\ttypprod}[2]{\prod_{#1}.#2}
\newcommand{\sforall}[2]{\forall_{#1}.{#2}}
\newcommand{\sexists}[2]{\exists_{#1}.{#2}}
\newcommand{\tspecindprod}[2]{\forall_{#1}.{#2}}
\newcommand{\tspectypprod}[2]{\sqcap_{#1}.{#2}}
\newcommand{\tspeckinprod}[2]{\cap_{#1}.{#2}}
\newcommand{\type}{\tau}
\newcommand{\scontext}{\mathbb{S}}
\newcommand{\kcontext}{\mathbb{K}}
\newcommand{\tcontext}{\mathbb{T}}
\newcommand{\icontext}{\mathbb{I}}

\newcommand{\context}{\kcontext \!\mid \!\icontext \!\mid\! \tcontext}

\global\long\def\tcomp#1#2#3#4#5{{\left\{  \right.\negthickspace\negthickspace\negthickspace\circ}\,
#1 : #2 \; ; #3 : #4 \mid \; #5
\,{\circ\negthickspace\negthickspace\negthickspace\left.\right\}  }}

\global\long\def\tcompbase#1#2#3{{\left\{  \right.\negthickspace\negthickspace\negthickspace\circ}\,
#1 : #2 \mid \; #3
\,{\circ\negthickspace\negthickspace\negthickspace\left.\right\}  }}

\newcommand{\ecomp}[5]{{\left\{  \right.\negthickspace\negthickspace\negthickspace\circ}\,
#1 : #2 \; ; #3 : #4 \! \mid  #5
\,{\circ\negthickspace\negthickspace\negthickspace\left.\right\}  }}

\newcommand{\ecompbase}[3]{{\left\{  
\right.\negthickspace\negthickspace\negthickspace\circ}\,
#1 : #2 \!\mid  #3
\,{\circ\negthickspace\negthickspace\negthickspace\left.\right\}  }}

\newcommand{\scompbase}[3]{{\left\{  #1 \in #2 \mid \; #3\right\}  }}

\newcommand{\textrule}[1]{\texttt{#1}}
\newcommand{\elimrule }[1]{{#1}\textrule{-E}}
\newcommand{\introrule}[1]{{#1}\textrule{-I}}
\newcommand{\autorule}[1]{\relax\ifmmode{\scriptstyle({#1})}\else$(#1)$\fi}

\newcommand{\modelimrule }{\autorule{\elimrule{\textrule{Mod}}}}
\newcommand{\modintrorule}{\autorule{\introrule{\textrule{Mod}}}}
\newcommand{\impelimrule }{\autorule{\elimrule{\textrule{Imp}}}}
\newcommand{\impintrorule}{\autorule{\introrule{\textrule{Imp}}}}
\newcommand{\progunielimrule }{\autorule{\elimrule{\textrule{UniProg.}}}}
\newcommand{\proguniintrorule}{\autorule{\introrule{\textrule{UniProg.}}}}
\newcommand{\expunielimrule }{\autorule{\elimrule{\textrule{UniExp.}}}}
\newcommand{\expuniintrorule}{\autorule{\introrule{\textrule{UniExp.}}}}
\newcommand{\expshortunielimrule }{\autorule{\elimrule{\textrule{UniExp.}}}}

\newcommand{\typeunielimrule }{\autorule{\elimrule{\textrule{UniType}}}}
\newcommand{\typeuniintrorule}{\autorule{\introrule{\textrule{UniType}}}}
\newcommand{\memelimrule }{\autorule{\elimrule{\textrule{Mem}}}}
\newcommand{\memintrorule}{\autorule{\introrule{\textrule{Mem}}}}
\newcommand{\memunelimrule }{\autorule{\elimrule{\textrule{Mem}_0}}}
\newcommand{\memunintrorule}{\autorule{\introrule{\textrule{Mem}_0}}}
\newcommand{\idrule}{\autorule{\textrule{Id}}}
\newcommand{\monrule}{\autorule{\textrule{Mon}}}
\newcommand{\monshortrule}{\autorule{\textrule{Mon.}}}
\newcommand{\antiredtermrule}{\autorule{\textrule{$\rightsquigarrow$}}}
\newcommand{\convrule}{\autorule{\textrule{$\conv$}}}

\newcommand{\holunielimrule }{\autorule{\elimrule{\textrule{Uni}}}}
\newcommand{\holuniintrorule}{\autorule{\introrule{\textrule{Uni}}}}

\newcommand{\reduction}{\rightsquigarrow}
\newcommand{\betared}{\rightsquigarrow_\beta}

\newcommand{\conv}{\equiv}
\newcommand{\progcontext}{\mathcal{C}}

\newcommand{\typecontext}{\mathcal{T}}
\newcommand{\exprcontext}{\mathcal{E}}
\newcommand{\speccontext}{\mathcal{S}}
\newcommand{\indcontext}{\mathcal{I}}

\newcommand{\exprecontext}{\exprcontext_\mathrm{e}}
\newcommand{\spececontext}{\speccontext_\mathrm{e}}

\newcommand{\exprtcontext}{\exprcontext_\mathrm{T}}
\newcommand{\spectcontext}{\speccontext_\mathrm{T}}

\newcommand{\indtcontext }{ \indcontext_\mathrm{T}}

\newcommand{\exprscontext}{\exprcontext_\mathrm{s}}
\newcommand{\specscontext}{\speccontext_\mathrm{s}}

\newcommand{\expricontext}{\exprcontext_\mathrm{i}}
\newcommand{\specicontext}{\speccontext_\mathrm{i}}

\newcommand{\indicontext }{ \indcontext_\mathrm{i}}

\newcommand{\specpcontext}{\speccontext_\mathrm{p}}

\newcommand{\hole}{\square}
\newcommand{\termvalue}{\mathcal{V}}

\newcommand{\callcc}{\ensuremath{\mathtt{call/cc}}}
\newcommand{\throw}[2]{\ensuremath{\mathtt{throw}_{#2}^{#1}}}

\let\rq\relax
\let\lq\relax
\DeclareMathSymbol{\lq}{\mathord}{operators}{``}
\DeclareMathSymbol{\rq}{\mathord}{operators}{`'}

\newcommand{\pole}{\bot\!\!\!\bot}

\newcommand{\holhyp}{\vec{\psi}}
\newcommand{\trtypecontext}{\sprops^{\mathrm{T}}}
\newcommand{\trspeccontext}{\sprops}
\newcommand{\trcontext}{\Gamma}

\DeclareMathSymbol{\mlq}{\mathord}{operators}{``}
\DeclareMathSymbol{\mrq}{\mathord}{operators}{`'}
\newcommand{\qu}[1]{\mlq #1 \mrq}

\newcommand{\Total}{\mathsf{Tot}}
\newcommand{\Det}{\mathsf{Det}}

\newcommand{\N}{\mathbb{N}}
\newcommand{\isnat}{\mathbb{n}_{n}}

\newcommand{\cc}{\mathsf{CC}}

\newcommand{\lookup}{\mathsf{lookup}}
\newcommand{\update}{\mathsf{update}}
\newcommand{\alloc}{\mathsf{alloc}}

\newcommand{\emem}{e_{\cc}}
\newcommand{\Phimem}{\Phi_{\mathsf{mem}}}

\newcommand{\ie}{i.e.\xspace}
\newcommand{\eg}{e.g.\xspace}

\begin{document}

\title{\scalebox{0.95}{Syntactic Effectful Realizability in Higher-Order Logic}
\thanks{Cohen, Grunfeld and Kirst were partially supported by Grant No. 2020145 from the United States-Israel Binational Science Foundation (BSF). Kirst also received funding from the European Union’s Horizon research and innovation programme under the Marie Skłodowska-Curie grant agreement No.101152583 and a Minerva
Fellowship of the Minerva Stiftung Gesellschaft für die Forschung mbH.}
}

\author{
\IEEEauthorblockN{Liron Cohen}
\IEEEauthorblockA{
\textit{Ben-Gurion University}\\
Be'er Sheva, Israel \\
cliron@bgu.ac.il}
\and
\IEEEauthorblockN{Ariel Grunfeld}
\IEEEauthorblockA{
\textit{Ben-Gurion University}\\
Be'er Sheva, Israel \\
arielgru@post.bgu.ac.il}
\and
\IEEEauthorblockN{Dominik Kirst}
\IEEEauthorblockA{
\textit{Ben-Gurion University, Israel}\\
\textit{Inria Paris, France} \\
dominik.kirst@inria.fr}
\and
\IEEEauthorblockN{\'{E}tienne Miquey}
\IEEEauthorblockA{
\textit{Aix Marseille Univ., CNRS, I2M}\\
Marseille, France \\
etienne.miquey@univ-amu.fr	}
}
% \and
% \IEEEauthorblockN{Ross Tate}
% \IEEEauthorblockA{
% \textit{Independent Researcher}\\
% NY, USA \\
% ross.tate@gmail.com}

% }

\maketitle
\thispagestyle{plain}
\pagestyle{plain}

\begin{abstract}
%Realizability theory concretizes the principle of constructivity by interpreting propositions as specifications for computational entities within a programming language.
Realizability interprets propositions as specifications for computational entities in programming languages. 
Specifically, syntactic realizability is a powerful machinery that handles realizability as a syntactic translation of propositions into new propositions that describe what it means to realize the input proposition.
%
%However, while realizability has been adapted to handle effectful programs, syntactic realizabiliy remains confined to pure (up to non-termination) programs.
%\dknote{Doesn't that conflict with the Nantes crew's work?}
%However, while traditional realizability has been extended in recent years to support effectful programs and thus modern programming languages,   
%current works on syntactic realizabiliy are still limited to pure (up to non-termination) programs.
%
This paper introduces $\hopl$ (Effectful Higher-Order Logic), a novel framework that expands syntactic realizability to uniformly support modern programming paradigms with side effects. 
$\hopl$ combines higher-kinded polymorphism, enabling typing of realizers for higher-order propositions, with a computational term language that uses monads to represent and reason about effectful computations. 
We craft a syntactic realizability translation from  (intuitionistic) higher-order logic ($\hol$) to $\hopl$, ensuring the extraction of computable realizers through a constructive soundness proof.
%
%The constructive soundness proof of the translation provides machinery for synthesizing a computable realizer from a proof in $\hol$.
$\hopl$’s parameterization by monads allows for the synthesis of effectful realizers for propositions unprovable in pure HOL, bridging the gap between traditional and effectful computational paradigms. 
Examples, including continuations and memoization, showcase $\hopl$'s capability to unify diverse computational models, with traditional ones as special cases.
For a semantic connection, we show that any instance of $\hopl$ induces an evidenced frame, which, in turn, yields a tripos and a realizability topos.

\end{abstract}
%
% \begin{IEEEkeywords}
% component, formatting, style, styling, insert.
% \end{IEEEkeywords}

\section{Introduction}

\lcnote{maybe, given our late reviewer, somewhere mention the importance of HOL/why we chose it? so we don't get the "I only care about dependent TT" review}

%\paragraph{General realizability}
Realizability, rooted in the works of Kleene~\cite{kleene1945interpretation},  aims to concretize the principle of constructivity by interpreting propositions as specifications for computational entities within a programming language.
A key feature of realizability is its assurance that the evidence for every proposition is computable. 
Originally, realizability used natural numbers to interpret formulas, but to get a more general notion of realizability, its modern notion generally uses some complete code language for realizing formulas. In the constructive case, it is standardly based on the notion of Partial Combinatory Algebras (PCAs)~\cite{hofstra2004partial}, while  Krivine classical realizability uses an abstract notion of stack machines~\cite{Krivine09}.
%
%The formalism of PCAS offers an abstract interface ensuring computability, maintaining Turing completeness without tying to the specifics of any programming language. 
%\lcnote{TODO:Realizability theory has been highly successful in various applications in programing languages, such as ....}

% By interpreting logical propositions as computational processes, realizability enables validating  both  syntactic and semantic aspects of programming languages.
% \lcnote{examples! } \emnote{not sure what you meant here}
%

%\paragraph{Syntactic realizability}
However, both the constructive and classical approaches are based on semantic constructions to yield models defined as a tripos or a topos~\cite{hyland1980tripos}. %\emnote{I rewrote this for the syntax vs semantics clash}
% \revnote{Krivine realizability is literally the most syntactic notion I can think of, as they sometimes use primitives that break the natural equational theory of the $\lambda\mu$-calculus. Despite the name, there is nothing "abstract" in the "Krivine abstract machine". Where did you get this idea?}
%
On the other hand, the syntactic approach to realizability,  pioneered by Gödel~\cite{godel1958bisher} and further developed in Kreisel's modified realizability~\cite{kreisel1959interpretation}, can be seen as abstracting away many of the complex semantic machinery otherwise required for constructing realizability models for rich languages. 
%Furthermore, the syntactic approach does not require the usually heavy machinery required for constructing meaningful realizability models for rich languages. 
By restricting to the syntax and abstracting away all the nifty details of any particular semantic structure, this also allows for a broader spectrum of possible interpretations, each yielding its own realizability interpretation by virtue of being a model of the target language, without having to tailor the realizability construction to some particular structure.
Roughly speaking, it is based on  handling  realizability as a syntactic translation of formulas into new formulas describing what it means to realize the input formula.
This can be seen as an internalization of the notion of realizability of the source language into the target language.
Recent works adopting the syntactic approach include, e.g.,~\cite{van1997modified,jaber2016definitional,Pedrot+Tabareau:esop:2018,letouzey2002new,forster2024verified,bernardy2011realizability}.

%\paragraph{the gap.}
Yet, works on \emph{syntactic} realizability tend to focus on the traditional notion of realizability or single computational effects.
% However, in its syntactic form it is limited to first order logic and purely functional programs, which do not allow for computational effects, and reverts to semantic constructions such as the modified realizability topos\cite{van1997modified} to interpret higher order logic.
%
%But this notion is limited in that it can only reason about pure programs, that is, programs whose outputs are functionally related to their inputs with the only allowed side-effect being non-termination. 
However, features of modern programming languages include multiple effects like non-determinism, stateful computation, probabilistic computation, etc. 

%\paragraph{Effectful realizability is even better}
%To support reasoning about current programming languages that contain various side effects, 
To this end, much work has been devoted to the extension of the notion of realizability to allow for side effects, e.g.,~\cite{Lepigre16,miquey18,geoffroy18,ahman2017dijkstra,maillard2019dijkstra,CohMiqTat21,cohen2019effects,Pedrot+Tabareau:lics:2017,Boulier+Pedrot+Tabareau:cpp:2017,Pedrot+Tabareau:esop:2018,Pedrot+Tabareau+Fehrmann+Tanter:icfp:2019,Pedrot+al:lics:2020,Pedrot+Tabareau:popl:2020}. 
Other than providing support for richer programming languages, it was also shown that effectful realizability provides new models, which, in turn, can implement new features. 
For instance, with the development of classical realizability continuing work of Griffin~~\cite{griffin90}, Krivine evidences the fact that
extending the $\lambda$-calculus with new programming instructions may result in  new
reasoning principles: \callcc~to get classical logic~\cite{Krivine09}, \texttt{quote} for dependent choice~\cite{krivine03}, etc. 
% Indeed, much work in recent years has been devoted to the development of effectful realizability. 
% \lcnote{cite!}

%\paragraph{Extend syntactic approach with effects}
Accordingly, this work provides a framework for syntactic effectful realizability.
That is, we extend the syntactic realizability approach by considering a target language that supports effectful programs as realizers for  a  higher-order source language.
Concretely,  we present a framework, dubbed $\effhol$ for Effectful Higher-Order Logic, that combines two key features:  higher-kinded polymorphism and a computational term language.
The higher-kinded polymorphic type system, inspired by Girard's System $\fomega$~\cite{Fomega}, allows for typed realizers of higher-order propositions. %, inspired by the view of System $\fomega$ as a higher-order propositional logic\agnote{CITE}.  
The computational term language, which can be seen as a simplification of Pitts' Evaluation Logic~\cite{pitts1991evaluation}, 
 enables reasoning about effectful programs in the sense of a general program logic~(cf.~\cite{VSJ25carte}).
%
%It can be seen as a simplification of Pitts' Evaluation Logic~\cite{pitts1991evaluation} that only has a single monotonic modality with the values and composition laws holding only in the left-to-right direction.
%
In particular, %the effectful aspect of the language is captured through monads.
to support effectful realizability, we adhere to the  standard approach for reasoning about effectful programs via \emph{monads}~\cite{moggi1991notions}.
%
%Moreover, our framework supports effectful realizability by incorporating effects through the notion of a monad. 
\lcnote{fix}This provides a uniform language parameterized intuitively by a monad that carries the computational behavior of the language. 
%The use of monads provides a structured way to handle computational effects 
%The use of monads allows for a structured and modular way to handle computational effects within the realizability framework in that monads  carry a lot of the effectful structure which facilitates the handling of effectful elements.
By providing internal support for standard program language features, for example by having typed realizers and effectful programs, $\hopl$ is applicable to a broad range of languages, and reasoning about realizers is done in a natural manner, similar to the way one would reason about programs.
Traditional systems like System $F_\omega$ and computational $\lambda$-calculus are subsystems of $\effhol$, illustrating the versatility and robustness of the system. 

%
% Nonetheless, current realizability frameworks present a realizability interpretation over some languages \emph{externally}.
% To make full use of the capabilities of realizability \lcnote{in what way?}, this paper aims to \emph{internalize} it. 
% That is, we present a framework, dubbed $\effhol$ for Higher-Order Program Logic, that internalizes realizability within a concrete language, not merely in some vague meta theory.

%\paragraph{what we show}
Concretely, we show how to model (intuitionistic) higher-order logic~\cite{Jacobs99CLTT}  ($\ihol$) using $\effhol$ by providing a syntactic realizability translation from $\ihol$ into $\effhol$. 
The realizability translation assigns to an $\hol$ proposition the type of its realizers in $\hopl$, along with a specification describing which programs of the corresponding type are realizers of said proposition. 
% $\effhol$ uses indexes to model the sorts in $\ihol$ and specifications to model the propositions in $\ihol$, on an abstract notion of programs with minimal structure.
%%
 %Many programming languages do not have types, let alone kinds, but we illustrate how $\effhol$ can be fit onto even untyped languages, keeping in line with our goal to admit diverse languages.
% As such, traditional models like System $\fomega$ and computational $\lambda$-calculus are special cases at opposite extremes.
% In particular, these languages can have side effects, which brings in new models with interesting properties.
%
A key feature of our syntactic realizability translation is that when translating a provable $\ihol$ sequent, it provides an algorithm for constructing an $\effhol$ proof of its translation, which, in turn, contains a realizer. Thus, our translation can be seen as synthesizing $\hol$ realizers within $\effhol$.
Furthermore, as $\hopl$ is parameterized by a monad, the synthesized realizer obtained from the soundness proof is agnostic to the specifics of the monad. 
That is, the $\hol$ realizers stemming from the soundness proof do not make actual use of the specific behavior of the monad. However, as we show via a few illustrative examples, we can take advantage of the monad to get effectful realizers for propositions that are not provable in pure $\hol$.
%
% \lcnote{we should say something about the monad as parameter: the realizer works for every monad, so it shows that the effect does not interfere with the realizing of the proposition, so our translated realizers do not really  \emph{use} the monad, but, as we show in an example, we can extend with  realizers that use the monad}

To obtain a semantic connection, we link our syntactic realizability translation to the framework of evidenced frames~\cite{CohMiqTat21}, which are abstract algebraic structures from which one can construct effectful realizability models (\ie, triposes) via a uniform construction. 
Here we show that for each instance of $\hopl$ with a concrete choice of monad, the realizability translation induces an evidenced frame, thereby associating our realizability translation with the tripos models.

\vspace{0.1cm}
\paragraph*{Outline and Main Contributions}
%The rest of the paper is organized as follows. 
\begin{itemize}[leftmargin=*]
\item \Cref{sec:overview} provides a high-level, intuitive overview of the construction of $\hopl$ and the realizability translation, emphasizing the reasoning behind each component and the interplay between them.
\item \Cref{sec:HOL} fixes the formalism of higher-order logic used in this paper.
\item \Cref{sec:HOPL} presents $\hopl$, a higher-order %program 
logic that combines higher-kinded polymorphism and computational types, thus  allowing higher-order reasoning about effectful programs. 
%the $\hopl$ system, emphasizing some key features and design choices. 
\item \Cref{sec:translations} provides a realizability translation of $\ihol$ to $\hopl$, yielding a realizability model of $\ihol$ from any instance of $\hopl$, and gives a constructive soundness proof for the translation, extracting $\effhol$ programs from $\hol$ proofs.
%provides the realizability translation from $\hol$ to $\hopl$, and proves its soundness. 
%We also establish the consistency of $\effhol$ through a trivialising translation in the converse direction.
\item \Cref{sec:app} demonstrates the utility and generality of $\hopl$ %can be used for capturing effectful %program 
%logics 
via illustrative examples, including continuations and memoization.
\item \Cref{sec:EF} relates our syntactic approach to realizability with the semantic ones by showing how the realizability translation induces a structure of an evidenced frame.
%establishes the connection between our syntactic realizability translation and the structure of Evidenced Frames.
\item \Cref{sec:related}  discusses related works and~\Cref{sec:conc} concludes.
\end{itemize}
We supplement our development with a Rocq mechanization (\url{https://github.com/dominik-kirst/effhol}), %, focusing on the formalization of the systems and the realizability translations. 
and mechanized results are indicated with clickable \coqdoc{EffHOL.html} icons\vlong{ and full proofs and details can be found in the appendix}.
\section{Overview of the Realizability Framework}\label{sec:overview}

The realizability approach to semantics relates logic with computation by interpreting propositions as denoting specifications of computer programs, or equivalently, using computer programs to demonstrate the validity of propositions.
While the approach is very general, with many variations, at its core it usually relates a specific logic with a specific programming language by translating each proposition in the logic to a specification for programs of this programming language.
 The wide literature on realizability interpretation introduces numerous very different such interpretations with different purposes. 
A recent work~\cite{CohMiqTat21} aiming to pinpoint the common structure of these interpretations identified a structure, dubbed evidenced frame, mathematically defined as a triple $(E,\Phi,\Vdash)$ precisely capturing these three components : $E$ being the set of so-called evidences (the programs that serve as realizers), $\Phi$ the set of formulas (\ie the logic that is being interpreted), while $\Vdash$ is the realizability relation connecting the former components. 
Many realizability interpretations usually define this realizability relation $\cdot\Vdash\cdot $ \emph{externally}, in the meta-theory. 

Following the idea that this relation can be understood as defining specifications for programs to be adequate realizers of the corresponding formulas, 
we make this slogan more concrete by relying on a program logic to define such specifications.
In line with Kreisel's modified realizability~\cite{kreisel1959interpretation}, we formally consider a realizability interpretation as a translation: 
it then translates propositions, \ie statements about mathematical structures, as  specifications in our target logic, \ie statements about computer programs.
\[\begin{tikzcd}
	{\textit{\small{\color{gray}Logical} Proposition}} &&& {\textit{\small{\color{gray}Program} Specification}}
	\arrow["{\text{\scriptsize {realizability translation}}}", from=1-1, to=1-4]
\end{tikzcd}\]

By carefully defining a general enough target language, which we call \emph{Effectful Higher-Order Logic ($\hopl$)} (introduced in \Cref{sec:HOPL}), it turns out that not only does our approach encompass usual realizability interpretations, but it also provides us with an even more general framework, compatible with typed realizers, effectful programs, etc.
To better grasp intuitions on how these features fit into the picture, let us first recall through a simple example how realizability works. Consider for instance a simple tautology $\Phi$ expressing that the conjunction is, from a logical point of view, commutative:
$$ \Phi\eqdef \left(\Phi_{1} \wedge \Phi_{2}\right) \supset \left(\Phi_{2} \wedge \Phi_{1}\right) $$
Following the core intuition of realizability as depicted by the BHK interpretation~\cite{bates1985proofs}, a realizer of a conjunction is intended to be a program providing a pair of realizers, one for $\Phi_1$ and one for $\Phi_2$, while a realizer of an implication is meant to compute out of a realizer of the premise a realizer of the conclusion. 
Formally, this takes the form of a translation of any proposition $\Phi$ into a specification $\trspec{}{\Phi}$ mapping programs to propositions, $\trspec{}{\Phi}{}(r)$ expressing that $r$ is a realizer of $\Phi$.
A realizer of this proposition is expected, out of a program computing a pair of realizers for $\Phi_1$ and $\Phi_2$, to provide a computation that produces a pair of realizers for $\Phi_1$ and $\Phi_2$:
\begin{align*}  
\trspec{}{\varphi\supset \psi}(t)&\eqdef \forall c.\left(\trspec{}{\varphi}(c) \supset \trspec{}{\varphi}(t\,c)\right)
\\
\trspec{}{\varphi_1 \land \varphi_2}{}(c)&\eqdef \trspec{}{\varphi_1}{}(\pi_1\, c) \land \trspec{}{\varphi_2}{}(\pi_2\,c)
\end{align*}
For instance, using a $\lambda$-calculus with primitive pairs and projections as the language of realizers, the term $r\eqdef \lambda x.\langle \pi_2\, x,\pi_1\,x\rangle$
would do the job: formally, with the definitions above, $\trspec{}{\Phi}{}(r)$ holds. 
In fact, since such a term does not rely on a specific choice of $\Phi_1$ and $\Phi_2$ (or of their interpretations), this term would be a realizer for the same proposition if $\Phi_1$ and $\Phi_2$ were universally quantified.
Moreover, depending on the choice of the set of programs that may serve as realizers, many other programs may realize the same proposition. In particular, realizers need not be purely functional (they may use, \eg, printing instructions, increase some references, etc) and there are no rules as to how that ordered pair has to be computed (right-to-left, left-to-right, in parallel?), what representation it should have (are pairs primitive, encoded?), or even whether an actual value of an ordered pair is ever given (or only a program that may compute such a pair).
Before introducing in detail our target logic $\hopl$, we first give an overview of some of its features and how they allow our framework to account for a large spectrum of realizability interpretations.

\subsection{Typed realizers}
Realizability is commonly based on an untyped notion of computer programs, where realizers of a given proposition $\Phi$ are characterized as being the computationally well-behaved programs (w.r.t. the specification $\Phi$ provides). Notably, when the languages for propositions and types coincide, typed programs are shown to be realizers of their types (but not the only ones) in what is usually called the \emph{adequacy lemma} of realizability interpretation. 
However, certain settings necessitate a clear distinction between the language of propositions that are interpreted and the language of types. This is for instance the case in Blot's interpretation of second-order arithmetic using programs in an extension of PCF
 % with an update recursion operators\footnote{The well-foundedness of his system (and in particular the fact that the update recursor defines a valid realizer of the corresponding proposition) crucially relies on programs to be proven continuous, which in turns is based on a denotational semantics defined inductively on PCF typing judgments.}
~\cite{blot22}, or in Paulin-Mohring's realizability model used to prove the soundness of Rocq's original extraction mechanism, where formulas in the Calculus of Constructions are realized by programs in $\fomega$~\cite{paulin_mohring89}. 
More generally, to enable our framework to provide specifications for actual programming languages, we need to (i) distinguish between the language used for logical propositions and that used for types, and (ii) allow the language of realizers to be typed.  The latter is not a restrictive requirement since considering a unique type assigned to all programs reduces to an untyped setting.
To further refine the separation, we consider an \emph{a priori} distinct language for the propositions in our logic, which we call \emph{specifications}.
To address this challenge, we thus enhance the translation to further include an assignment of a type for each proposition:
\[\begin{tikzcd}
	{\text{Proposition}} &&&& {\text{Type} \times \text{Specification}}
	\arrow["{\text{realizability translation}}", from=1-1, to=1-5]
\end{tikzcd}\]
Technically, if $\Phi$ is a proposition of the source logic, we define both its interpretations as a type $\trtype{}{\Phi}$ and as a specification $\trspec{}{\Phi}$. Going back to our example, assuming that $\Phi_1$ and $\Phi_2$ are respectively interpreted by some types $\tau_1$ and $\tau_2$, a realizer of $\Phi=\left(\Phi_{1} \wedge \Phi_{2}\right) \supset \left(\Phi_{2} \wedge \Phi_{1}\right)$ is required to be a program of type $\trtype{}{\Phi}=\tau_1 \times \tau_2 \to \tau_2\times \tau_1$, while $\trspec{}{\Phi}$  now defines a propositions on programs of that type. 
% \emnote{Not sure how this is formally expressed in the translation, in fact it looks that only the membership predicate allows to ensure something about the types of programs}

% \agnote{Now, go back to the examples and show what changes when you add types}
% That is, we now need two translation functions: One for the types, and one for the refinements.

% % However, given that many programming languages are typed, we may wish to extend realizability to a typed setting. This presents our first challenge: To support typed realizers in a typed programming language. \agnote{\st{REWRITE: given that most programming languages are typed, the most common realizability method is insufficient [TODO: check typed realizability and compare to this work]. For example.... this presents our first challenge: To support typed realizers}} However, to do so, we have to make sure the realizability translation respects the type structure.
% % \agnote{GIVE CONCRETE EXAMPLE}

% If a realizer of $\varphi \supset \psi$ is a program that converts realizers of $\varphi$ to realizers of $\psi$, then in a typed language it will have an input type $S$, and an output type $T$. Since the program can only take inputs of type $S$ and return outputs of type $T$, all realizers of $\varphi$ will have to have the type $S$ and all realizers of $\psi$ will have to have the type $T$.

% \agnote{\st{REWRITE: To address this challenge, we enhance the translation to further include an assignment of a type for each proposition}}

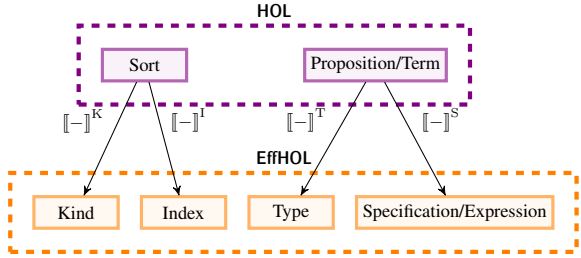
\begin{figure}[t]
    \centering
 \begin{tikzpicture}[scale=0.7, transform shape,
      0node/.style={rectangle, draw=pink!60, fill=pink!5, very thick, minimum size=6mm, minimum width=16mm}, 
      1node/.style={rectangle, draw=violet!60, fill=violet!5, very thick, minimum size=6mm, minimum width=16mm}, 
      2node/.style={rectangle, draw=orange!60, fill=orange!5, very thick, minimum size=6mm, minimum width=16mm}, 
      3node/.style={rectangle, draw=teal!60, fill=teal!5, very thick, minimum size=6mm, minimum width=16mm},
      4node/.style={rectangle, draw=pink!60, fill=pink!5, very thick, minimum size=6mm, minimum width=16mm},
      ,>=stealth']
      
    \node[1node]   at ( 2.3, 2.8) (1A)      {Sort}; 
    \node[1node, right=2.2  cm of 1A]  (1B)    {Proposition/Term}; 
 
 \node (1) [draw=violet, fit= (1A) (1B) , inner sep=0.3cm, 
                dashed, ultra thick,label={[align=center]above:{$\hol$} }] {};
                
    \node[2node]   at ( 1, 0) (2A)      {Kind}; 
    \node[2node, right=0.4cm of 2A]  (2B)    {Index}; 
       \node[2node, right=0.4cm of 2B]  (2C)    {Type}; 
          \node[2node, right=0.4cm of 2C]  (2D)    {Specification/Expression}; 
    
 \node (2) [draw=orange, fit= (2A) (2B) (2C) (2D) , inner sep=0.3cm, 
                dashed, ultra thick,label={[align=center]above:{$\hopl$\quad}}] {};
    
    \draw ( 1A) edge[->,"$\trkind{-}$" '] ( 2A);
    \draw ( 1A) edge[->,"$\trind{}{-}$"] ( 2B);
    \draw ( 1B) edge[->,"$\trtype{}{-}$" '] ( 2C);
    \draw ( 1B) edge[->,"$\trspec{}{-} $"] ( 2D);

\end{tikzpicture}
\vspace{-0.1cm}
    \caption{The Realizability Translation}
    \label{fig:translation}
\end{figure}

\subsection{Polymorphism}
Our candidate realizer $r$ for the proposition $\Phi$ does not rely on a particular choice for the propositions $\Phi_1$ and $\Phi_2$. At the typing level, this can be reflected using polymorphic types, in particular, in an expressive enough type-system we could type $r:\forall \tau_1,\tau_2.\tau_1\times \tau_2 \to \tau_2 \times \tau_1$. The logical counterpart of this amounts to the universal quantification on propositions provided by second-order logic, allowing for a refinement of the proposition compatible with any choice of $\Phi_1$ and $\Phi_2$:
\[ \Phi'\eqdef\forall \Phi_1,\Phi_2. \left(\Phi_{1} \wedge \Phi_{2}\right) \supset \left(\Phi_{2} \wedge \Phi_{1}\right) \]
While there exist numerous logical systems featuring such a quantification, \eg Girard-Reynold's System F~\cite{ProofsAndTypes}%arguably being the most famous one~\cite{ProofsAndTypes}
, it is well-known that this quantification introduces nuanced challenges when interpreting it within a model~\cite{reynolds1984polymorphism}. 
This work is no exception: since we want propositions in the source language to be interpreted both as a specification on their realizers and as the type of these realizers, it means that via a sound translation, the image of a universal quantification over any possible proposition should range over any possible specification over several possible types, while the corresponding realizers may be assigned polymorphic types. Hence, the type system for terms includes a quantification over types, while the specification language features both
a quantification over types and over specifications. 
With these, we define the translations to types and specifications of a proposition $\forall X.\Phi(X)$ as follows
\begin{align*}
\trtype{}{\forall X.\Phi(X)} &\eqdef \forall \tau:\star.\trtype{}{\Phi(\tau)}
\\
\trspec{}{\forall X.\Phi(X)} &\eqdef \forall \tau:\star.\forall p:\tau\to\mathsf{Prop}.\trspec{}{\Phi(p)}
\end{align*}
where $\tau:\star$ (resp. $p:\mathsf{Prop}$) denotes that $\tau$ is a type (resp. $p$ a proposition). In particular, as specifications, propositional variables are interpreted as predicates on their realizer's type.

\subsection{Higher-Order Logic}
The realizability interpretation we provide is, in fact, of stronger logical expressiveness in that it interprets higher-order logic ($\hol$), which adds a few more technicalities. 
In (multi-sorted) $\hol$, besides mere propositions, formulas also include predicates ranging over terms of a given \textit{sort} $s$ (acting like sets of such terms), predicates ranging over such predicates (acting like sets of sets), etc.  
In the general case, $\hol$ includes predicates of sorts $\pred{\sort_1,\ldots,\sort_n}$, ranging over terms of sorts $s_1$, ... $s_n$.
This makes the underlying theory expressive enough to internalize any (open) formula $\Phi(x_1,....,x_n)$ ranging over variables $x_1,...x_n$ respectively of sorts $s_1,...,s_n$ through a comprehension predicate $\{x_1:s_1,...,x_n:s_n|\Phi(x_1,....,x_n)\}$ of sort $\pred{\sort_1,\ldots,\sort_n}$. 
The fact that this predicate holds for certain terms $t_1,...,t_n$ can then be expressed using a membership proposition $t_1,...,t_n\in \{x_1:s_1,...,x_n:s_n|\Phi(x_1,....,x_n)\}$
that is logically equivalent to the formula $\Phi(t_1,...,t_n)$.

The realizability interpretation requires the target language $\hopl$ to also encompass higher-order logic. Specifically, since any proposition of $\hol$ is translated both as a type and as a specification, both components need to be equipped with a counterpart of the sort system: the translation to types uses a \textit{kind system} for types while the translation to specifications uses something analogous to kinds for specifications which we call \textit{indices}.
To summarize, our target language $\hopl$ is a language of \textit{specifications}, whose structures are reflected by \textit{indices}, expressing properties of \textit{terms}. In turn, these terms come with a \textit{type}, whose structure is expressed by means of a \textit{kind} system. 
The realizability interpretation, as illustrated by \Cref{fig:translation}, consists of four translations, two translations $\trkind{\cdot}$ and $\trind{}{\cdot}$ from sorts to kinds and indices, and two translations $\trtype{}{\cdot}$ and $\trspec{}{\cdot}$ to types and specifications respectively.
% \footnote{In fact, to make the system even more precise, we will introduce the distinction between propositions and comprehension terms in $\hol$, and correspondingly between specifications and what we called \textit{expressions} in $\hopl$, and we will therefore have 6 translations\emnote{do we want to say this?}}. 

In particular, for $\hol$ formulas that quantify over predicates of a given sort, \eg $\forall X:s.\Phi(X)$, through the realizability interpretation, predicates (here the predicate variable $X:s$) are turned into specifications expressing which terms define adequate realizers. Such specifications therefore range over one extra variable for terms of the corresponding type, and as such their index should be refined to reflect this. We write $\refpred{\type}{\indice_1,...,\indice_n}$ to denote the index of a specification (over specifications of indices $\indice_1,...,\indice_n$) whose realizers are of type $\type$. Hence, the translation of a sort $\sort$ to an index takes the corresponding type (say $\tau$) of the intended realizers as a parameter, written as $\trind{\type}{\sort}$. Overall, the translations of a formula $\forall X:s.\Phi$ to types and specifications is given by:
\begin{align*}
\trtype{}{\forall X:s.\Phi(X)} &\eqdef \forall \tau:\trkind{s}.\trtype{}{\Phi(\tau)}
\\
\trspec{}{\forall X:s.\Phi(X)} &\eqdef \forall \tau:\trkind{s}.\forall p:\trind{\tau}{s}.\trspec{}{\Phi(p)} 
\end{align*}
The full translation, with extra technicalities for handling open propositions and various contexts, is given in \Cref{sec:HOLtoHOPL}.

\subsection{Spectrum}
Following the modern tradition of realizability interpretations, we aim at a framework
that allows for effectful realizers. 
In this paper, we choose to approach effects via monads, as is done in Pitts' Evaluation Logic~\cite{pitts1991evaluation} and Moggi's Computational $\lambda$-Calculus~\cite{moggi1989computationallamdba}. Specifically, our language features a type $\typecomp{\tau}$ denoting computations of type $\tau$, while terms account for the usual \textit{return} and \textit{bind} constructions of monads.
This choice has the strong benefit that the resulting system is compatible with \textit{any} choice of a monad%\footnote{NOPE WE DON'T: We provide in \Cref{sec:semantics} a semantics for $\hopl$ which interprets monad in terms of oplax post-module\emnote{correct if needed} \emnote{cite \cite{moggi1995semanticsEL} ?}
%, so that any monad equipped with the corresponding semantics will result in a sound instance of $\hopl$.}
, endorsing the subsequent language of realizers with the corresponding effects. 
However, this is only a choice made for the purpose of providing a ready-to-use generic framework. Indeed, we could also consider variants of $\hopl$ where terms, instead of monadic constructs, allow for effectful instructions in direct-style, as is done for instance in Krivine classical realizability with the control operator \texttt{call/cc}~\cite{Krivine09}.\emnote{any idea for an extra ref using states?} Nonetheless, such an approach would have the significant drawback that any choice of a particular effect would require to adapt the operational semantics (\eg with stacks for control-operators, heaps for states, etc.), and therefore the target language $\hopl$ of our realizability translation and its semantics. 
While a monad-free approach to effects would be more general, ensuring soundness and usefulness for a specific instance would require significantly more effort. 
Thus, to provide a robust generic framework, we here focus on the monadic approach.
%In a way, while a setting that does not enforce using monads for handling effects would be indubitably more general, any specific instance would require much more work to ensure its soundness and no guarantee on its usefulness, hence it would have been unfair for us to introduce such a setting while claiming it to provide a generic framework. 
%\emnote{to improve/rephrase}\lcnote{better?}
%
This already allows us to cover a wide spectrum of realizability interpretations, as we shall demonstrate.

%\section{Intuitionistic Higher-Order Logic}

\begin{figure}[!t]
\begin{small}
\begin{center}

%$\begin{array}{c}
%\begin{minipage}{0.5\textwidth}
$\begin{array}{l@{\hspace{0.1in}}l@{\hspace{0.1in}}l@{\hspace{0.1in}}l@{\hspace{0.3in}}l}
\textbf{Sorts} & \sort 
            & ::=  & \STAR \mid \predcon{\sort} & %\mbox{(predicate)\lcnote{same as kinds. change?}}
            \\[\sskip]

\textbf{Terms} & \sterm 
            & ::=  & \spred \mid \comprehend{\spred}{\sort}{\sprop} \mid \compbase{\sprop}& %\mbox{(predicate)}
            \\[\sskip]

\textbf{Propositions} & \sprop 
            & ::=  & \smembase{\sterm}\mid \smem{\sterm}{\sterm}\mid %& %\mbox{(membership)}
           % \\
           % & & \mid &  
            \simplies{\sprop}{\sprop} 
           %& %\mbox{(implication)}
            % \\
            % & &\mid &  
            \mid \sforall{\spred:\sort}{\sprop} %& %\mbox{(universal quantification)}
            \\[\sskip]
            
\textbf{Sort context}& \scontext& ::= &\emptyset \mid \scontext,\spred:\sort
% \\[\myskip]
    \\[\myskip]
% \hdashline\\
\end{array}$
%\dknote{How about we use the same logical notation in HOL and HOPL (as we reuse the universal quantifier)?}\lcnote{I'd rather not}
% \end{minipage}
% \begin{minipage}{0.5\textwidth}
 %\end{minipage}
%\begin{minipage}{0.5\textwidth}
% \\[\myskip]
% $$\infer[\idrule]{ \selsequent{\scontext}{\sprops}{\sprop} }{\sprop\in \sprops}$$
\scalebox{0.9}{
    \begin{tabular}{@{}c@{\hspace{-0.05in}}c}
 %        $\infer[(\text{Identity})]{ \elsequent{}{\tprops , \tprop}{\tprop} }{}$ &\\[\myskip]
        $\infer[\!\!\!\impintrorule]{\selsequent{\scontext}{\sprops}{\simplies{ \sprop_{1} }{ \sprop_{2} }}}{\selsequent{\scontext}{\sprops,\sprop_{1}}{\sprop_{2}} }$ &
   %      Implication Intro & $\dfrac{ \elsequent{}{\tprops,\tprop_{1}}{\tprop_{2}} }{ \elsequent{}{\tprops}{\tprop_{1} \supset \tprop_{2}} }$\\[\myskip]
        $\infer[\!\!\!\impelimrule]{\selsequent{\scontext}{\sprops}{\sprop_{2}}}{\selsequent{\scontext}{\sprops}{\simplies{\sprop_{1}}{\sprop_{2}}} \quad \selsequent{\scontext}{\sprops}{\sprop_{1}}}$\\[\sskip]
    %     Implication Elim & $\dfrac{ \elsequent{}{\tprops}{\tprop_{1} \supset \tprop_{2}} \qquad \elsequent{}{\tprops}{\tprop_{1}} }{ \elsequent{}{\tprops}{\tprop_{2}} }$\\[\myskip]

        $\infer[\!\!\!\holuniintrorule]
        {\selsequent{\scontext}{\sprops}{\tspecindprod{\spred:\sort} \sprop}}
        {\selsequent{ \scontext, \spred : \sort}{\sprops}{\sprop}}$ %\lcnote{y should be fresh here? ie not in $\tprops$}
        &
        $\infer[\!\!\!\holunielimrule]
        {\selsequent{ \scontext}{\sprops}\sprop\subst{\spred:=\sterm}}
        {\selsequent{\scontext}{\sprops}{\tspecindprod{\spred:\sort} \sprop}
       % & \eltype{\scontext}{\sterm:\sort}
        }$
        \\[\sskip]

        $\infer[\!\!\!\memintrorule]
        {\selsequent{\scontext}{\sprops}{ \smem{\sterm}{\comprehend{\spred}{\sort}{\sprop}}}}
        {\selsequent{\scontext}{\sprops}{\sprop\subst{\spred := \sterm}}& %\eltype{\scontext}{\sterm:\sort}
        }$
        & \quad
        $\infer[\!\!\!\memelimrule]
        {\selsequent{\scontext}{\sprops}{\sprop\subst{\spred := \sterm}}}
        {\selsequent{\scontext}{\sprops}{ \smem{\sterm}{\comprehend{\spred}{\sort}{\sprop}}}}
        $\\[\sskip]
\end{tabular}}
\\
\scalebox{0.95}{
    \begin{tabular}{c@{\hspace{0.2in}}c@{\hspace{0.2in}}c}
        $\infer[\!\!\!\memunintrorule]
        {\selsequent{\scontext}{\sprops}{ \smembase{\compbase{\sprop}}}}
        {\selsequent{\scontext}{\sprops}{\sprop}}$
        &
        $\infer[\!\!\!\memunelimrule]
        {\selsequent{\scontext}{\sprops}{\sprop}}
        {\selsequent{\scontext}{\sprops}{ \smembase{\compbase{\sprop}}}}
        $
        &
        $\infer[\!\!\!\idrule]{ \selsequent{\scontext}{\sprops}{\sprop} }{\sprop\in \sprops}$
    \end{tabular}
}
\end{center}
\end{small}
\vspace{-0.2cm}
   \caption{$\ihol$ Syntax and Theory \coqdoc{HOL.html}}
    \label{tab:hol-syntax}
\end{figure}

\label{sec:ihol}
\section{Higher-Order Logic}\label{sec:HOL}

To begin, we fix the specific version of higher-order logic employed in this paper. %, as the literature presents numerous variations for it.
To generate realizability models, we focus on a constructive variant, namely, intuitionistic, many-sorted, monadic, higher-order logic,  denoted by $\ihol$.
%\lcnote{do we need to explain why many sorted?}
To keep our system as general as possible, we opt for a very minimalistic formalization of $\ihol$. 
For one, we only consider a core of logical constructs, namely, implication and universal quantification.
%\lcnote{do a better job at explaining why we choose this very simplified HOL}
%Hence, from now on, the abbreviation $\ihol$ will refer to this specific formulation of intuitionistic many-sorted higher-order logic.
Furthermore, our language emulates propositional application and abstraction through comprehension terms and membership propositions.
We do this instead of the alternative formalization using  $\lambda$-terms to ensure we do not commit to any specific language construct that goes beyond the bare minimum required for $\ihol$.

\vshort{\Cref{tab:hol-syntax} presents the $\ihol$ framework, where all inference rules assume well-formedness.}
\vlong{\Cref{tab:hol-syntax} presents the $\ihol$ framework.
The typing rules, along with their explicit mentions in the inference rules are detailed in the appendix, while here all inference rules assume well-formedness.}
As a many-sorted logic, the syntax of $\ihol$ consists of \emph{propositions} and \emph{terms}, where each term has a \emph{sort}.
Terms are either variables or (base) comprehension terms. %of propositions.
%\lcnote{why use the notion of predicate? fix the next sentence}The only sorts we consider are \emph{predicate sorts}, which are the sorts of \emph{predicate terms}, or simply \emph{predicates}.
Propositions can make statements about properties of terms by having terms appear inside propositions.
The comprehension terms provide a syntactic machinery for terms to refer to propositions, which, in turn, allows propositions to  make statements about properties of other propositions.
This is the core source of the higher-order structure of the logic.
`Extracting' the inner proposition from a comprehension term, is then done via the membership proposition.
Concretely, for a term $\sterm_{1}$ of sort $\sort$ and a comprehension term $\sterm_{2}=\comprehend{\spred}{\sort}{\sprop}$ of  sort $\predcon{\sort}$, 
the proposition $\smem{\sterm_{1}}{\sterm_{2}}$ states that  $\sterm_{1}$ is a member of $\sterm_{2}$, which amounts to saying that $\sprop$ with $\sterm_1$ for $\spred$ holds .
This mechanism may be seen as a representation of standard function abstraction and application in programming languages, which allows 
%a programming language to have 
higher-order functions. 
Comprehension binds a  variable within a proposition and turns it into a comprehension term, just as abstraction binds a variable within an expression to turn it into a function.
Membership applies a comprehension term to an argument, yielding a proposition, just like a function applied to an argument yields an expression.

Propositions include only implication, universal quantification, and membership, from which the other logical constructs are standardly  derivable. 
But, while comprehension terms $\comprehend{\spred}{\sort}{\sprop}$ of composite sort $\predcon{\sort}$ are in principle enough to encode the other logical connectives and construct the higher-order hierarchy, our language also requires base terms and propositions. 
For this, we include comprehension terms $\compbase{\sprop}$ of base sort $\STAR$ with corresponding embedding of terms $t:\STAR$ as formulas $\smembase{t}$ to effectively simulate quantification over propositional variables.
Thus, for instance, a falsity constant $\bot$ can be defined naturally by $\tspecindprod{\spred : \STAR}{\smembase{\spred}}$.
%, otherwise a more indirect encoding with predicates and membership would become necessary.

We use $\scontext$ to denote a sort context, i.e., a list of unique sorted variables $ \spred_{1} : \sort_{1} , \ldots , \spred_{n} : \sort_{n} $.
%\lcnote{here too we should have well-founded rules. or at least say they are standard}
%
We have standard typing rules for the judgments that, under a sort context $\scontext$ a term is of a specific sort, written $\eltype{\scontext}{\sterm : \sort}$, and that $\sprop$ is a well-formed proposition, written $\eltype{\scontext}{\sprop}$.
We opt for a representation where well-formed propositions are a separate construct,  equivalently, a unique sort could have been used to denote well-formed propositions.

For simplicity, we use a monadic presentation of $\ihol$, i.e.\ using only unary sorts $\predcon{\sort}$ instead of $n$-ary  sorts of the form $\predcon{(\sort_1,\dots\sort_n)}$.
% and unary indices $\refpred{\type}{\indice}$ instead of $n$-ary  indices of the form $\refpred{\type}{\indice_1,\dots\indice_n}$.
While in second-order logic, the monadic fragment is less expressive, with arbitrary orders, expressivity is equivalent since $n$-ary sorts can be encoded at higher orders~\cite{scott2008pairs}.
%In a restricted system like second-order logic, the monadic fragment is much less expressive than the system with arbitrary arities.
%However, if arbitrary orders are admitted, there is no difference in expressivity as $n$-ary sorts of a certain order can be encoded at a higher order~\cite{scott2008pairs}.
Thus, our sort structure is, in fact, equivalent to the natural numbers but is presented syntactically to allow for potential extensions to more complex sort structures.
%Therefore, in our setting the sort structure is, in fact, equivalent to the natural numbers. Nonetheless, we here present it in a syntactic way that could in principle be extended to accommodate more complex sort structures.

The main judgments of the theory of  $\ihol$ are of the form $\scontext \vdash \sprops\Rrightarrow\sprop$,
where $\sprops$ is a finite set of propositions and $\scontext$ is an adequate  sort context.
%
%The propositions on the left-hand-side of a sequent are collectively called the \emph{antecedent} of the sequent, while the proposition on the right-hand-side $\tprop$ is called the \emph{consequent} of the sequent.
A sequent represents \emph{entailment}, that is, the judgment that within the specified context, assuming all the propositions in the left-hand-side hold entails the right-hand-side holds.
The theory of  $\ihol$ is inductively generated by the inference rules in~\Cref{tab:hol-syntax}. Standardly, in rule \holuniintrorule, $\spred$ is a fresh variable.
The membership rules simply identify membership in comprehension terms $\smem{\sterm}{\comprehend{\spred}{\sort}{\sprop}}$ with substitutions $\sprop\subst{\spred := \sterm}$ of the inner formula as expected.
($\sprop\subst{\spred := \sterm}$ and $\sterm'\subst{\spred := \sterm}$ denote  standard substitutions of $\sterm$ for $\spred$ in $\sprop$ and $\sterm'$, resp.)
As noted, we omit the typing premises, but, for example, rule \memintrorule~has the additional premise $\eltype{\scontext}{\sterm:\sort}$.
Standard structural rules like Exchange, Weakening, Contraction, and Cut are derivable from those in~\Cref{tab:hol-syntax}, given that $\sprops$ is treated as a set in the \idrule\ rule.
%\revnote{we should note that in Mem-I t:s}

\section{\feffhol}
\label{sec:HOPL}

%\lcnote{TERMINOLOGY Q: we use sort/kind/type/index context, but really it is of predicates/type/term/assertions.
% The same for quantification.
% Should we switch/keep. why?}

% \agnote{Rename HKPEL to Higher Evaluation Logic}
% As a general, complex example. explain why it is important and the difference from previous examples in the overview.
% Formal definition of the  logic.
% \agnote{no need to deal with basic types}

% As an example of a program logic which supports higher-kinded polymorphism, we introduce \emph{Higher-Kinded Polymorphic Evaluation Logic}.

%\agnote{Use different notation, symbols, etc. in the src language and tgt language, use macros}

To provide an abstract, general framework for handling a wide range of program languages, this section describes~\feffhol, $\hopl$.
$\hopl$ is a form of higher-order logic based on two key features:  higher-kinded polymorphism and computational term language.
The higher-kinded polymorphic type system, inspired by Girard's System $\fomega$~\cite{Fomega}, is used to type higher-order specifications. 
The computational term language 
 %, which can be seen as a simplification of Pitts' Evaluation Logic~\cite{pitts1991evaluation}, 
 enables the treatment of effectful programs.
It can be seen as a simplification of Pitts' Evaluation Logic~\cite{pitts1991evaluation} that only has a single monotonic modality with the values and composition laws
holding only in the left-to-right direction.
In particular, the effectful aspect of the language is captured through monads.
%In~\ref{sec:conc} we discuss how effects can be handled in a variant of $\hopl$ in a more direct style.\lcnote{TODO:do this!}
%
The design of $\hopl$ invokes similar design choices as  $\ihol$, for example, opting for a monadic presentation and using only unary kinds and indices.
%
%\lcnote{do we want to say that EL does only monads and we can do others? ,maybe not in the current form with computations}\emnote{I would emphasize that we \textbf{choose} to put the focus of effects through monads and to then follow the lines of Pitts' evaluation logic, but that we could have also considered effects in direct-style. Then I'd recall this in this section for the components specifics to monads}
%It has terms to describe mathematical objects, along with a higher-kinded polymorphic type system with a special type denoting effectful computations, and a logical apparatus with a higher-order expressive power.

\begin{figure*}[!t]
\hspace{-0.38in}
\begin{small}
\begin{center}
\input{section/figures/HOPL_syntax2}
\end{center}
\end{small}
\vspace{-0.4cm}
   \caption{$\hopl$ Syntax \coqdoc{EffHOL.html}}
    \label{tab:effhol-syntax}
\end{figure*}

\subsection{The Language of $\hopl$}
%\emnote{should we say predicate variable or specification variable?}
%\paragraph{Syntax}
%\lcnote{discuss free and bound vars in the various expressions.}
The syntax of $\hopl$, formally given in~\Cref{tab:effhol-syntax}, consists of the following components:  \emph{kinds}, \emph{types}, \emph{programs}, \emph{indices}, \emph{expressions} and \emph{specifications}.  
Kinds and types are used to provide the types of realizers associated with typed programs, while indices, expressions and specifications hold the logical counterpart of our realizability interpretation by describing the properties of these programs.

% we use a monadic presentation of $\ihol$ for simplicity, i.e.\ only unary predicates of sort $\predcon{\sort}$ instead of $k$-ary predicates of a sort like $\predcon{(\sort_1,\dots\sort_k)}$.
% While in a pruned system like second-order logic the monadic fragment is much less expressive than the system with aribtrary arities, if arbitrary orders are admitted there is no difference in expressivity as $k$-ary predicates of a certain order can be encoded with unary predicates at a higher order~\cite{scott2008pairs}.
% This choice influenced the design of $\hopl$ which then also just has to accommodate monadic kinds and indices.

As is standard, \emph{kinds} are used as types for \emph{type constructors}. 
Well-formed type constructors are of kind 
%$\star$, while the kind 
$\typecon{\kind}$, which denotes the kind of type constructors that take type constructors of kind $\kind$ and return a type.
Concrete types take no arguments and have kind $\star$. %, meaning they take no arguments. %, which, for readability, we introduce explicitly as $\star$.%\emnote{still equalising them with $\star$ ?}
%\dknote{That's confusing, isn't the base kind of proper types exactly $\star$?}
As for sorts in $\ihol$, we take a minimalistic kind system, equivalent to the natural numbers. 

The syntax of \emph{types} is also minimal, and includes variables, application, abstraction, functions, universal, and computation types.
As is standard, the type abstraction and universal 
bind a free variable in a type.
%take a type with a free variable and bind it. %the type variable $\typevar$.
%
The former returns a type constructor, while the latter returns a polymorphic type. 
\revnote{can you explain a bit more about type abstraction versus universal types? the program syntax looks the same. how are these different?}
%
% The language does not have arbitrary function types. Instead, for every type $\type$ there is a type $\typecomp{\type}$ of computations of type $\type$, and functions must return computations.
% Function types are of the form $\typefun{\type_{1}}{\type_{2}}$. 
Intuitively, the $M$ stands for the underlying monad of $\effhol$.
This approach ensures that the computational behavior of programs is made explicit in the type system, and allows for easy extensions with constructs that exhibit various computational effects. %, such as non-termination.
%
% The function type takes a type and outputs a computation type. \lcnote{explain more about functions}
% A computation type of the form $\typecomp{\type}$  intuitively stands for the type of computations on type $\type$. This follows the lines of Evaluation Logic, focusing on representing effects through monads. \lcnote{elaborate and change wrt the change in the function type}
%
%\lcnote{elaborate: a variant of untyped lambda}
%
Since the type system has higher-order polymorphism it does not include 
 base types.
 %, which makes all the types polymorphic.
%
Each type variable has a kind.
To keep track of the kinds of the free type variables in a type, all judgments in $\hopl$ depend on a \emph{kind context}, denoted by $\kcontext$,  containing declarations of the form $\typevar : \kind$. 
%
% \lcnote{not true anymore!}{\color{red}Note that the application type is slightly more minimalistic than the standard one because we only allow applications of type variables. 
% We opt for this form in order to simplify the dependencies in our term syntax, which, in turn, allows us to simplify our translation. }
%

The syntax of programs combines constructs for pure and effectful computations. 
It contains the standard constructs of pure programs, as in System $\fomega$, such as variables, term abstraction and application, and type abstraction and application.
In addition, it contains standard constructs for handling effectful computations through monads via return and bind, as in Pitts' Evaluation Logic~\cite{pitts1991evaluation} (or rather, Moggi's Computational $\lambda$-Calculus~\cite{moggi1989computationallamdba}).
The variable $\typevar$ is bound in type abstraction  and the variable $\termvar$ is bound in term abstraction and bind program.
Programs depend on types, which, in turn, depend on kinds. Therefore, they further require a \emph{type context}.
%Since variables are terms, and terms have type application $\eltypapp{e}{\type}$, where $\type$ is a type that may contain variables, then terms depend on both a kind context and a type context. % that depends on that kind context, denoting the free variables in the term.
A type context, denoted by $\tcontext$, contains declarations of the form $\termvar : \type$. That is, it is a list of typed program variables, where each type may have free variables denoting type constructors, hence a type context depends on a kind context.
%, which holds  the free variables in the types.

%\emnote{MAKE SURE THAT THE MONADIC APPROACH HAS BEEN DISCUSSED EARLIER}
To reason about programs, the logical part of $\hopl$ relies on 
\emph{specifications}, which in turn are defined using \emph{expressions}
which are typed with \emph{indices}. 
Indices mark the logical order of statements. 
First-order statements about a program of type $\type$ are indexed with the base refinement  $\refpredN{\type}$. 
Higher-order statements can also refer to some other statements of different types, in which case they are of the form $\refpred{\type}{\indice}$.
The syntax of indices also includes indices of the form $\tindprod{\typevar:\kind} {\indice}$, which bind $\typevar$, that capture that the index is polymorphic over the type variable $\typevar$ of kind $\kind$, together with the corresponding expressions to abstract over types. %The type variable $\typevar$ is bound in $\tindprod{\typevar:\kind} {\indice}$.
\revnote{some examples here could help}

Expressions include variables, (base) comprehension, and type abstraction and application. 
The type abstraction (in which $\typevar$ is bound) and type application reflect type polymorphism.
%The type variable $\typevar$ is bound in type abstraction. 
%\lcnote{say more? }
To enable stating higher-order properties of programs, the comprehension expression binds a program variable along with an optional expression variable.
The program variable specifies the program satisfying the property, while the optional expression variable describes a property that is related to the program in some manner via the inner specification.\revnote{some examples here could help}
A comprehension expression that takes an expression argument is denoted by  $\ecomp{\termvar}{\type}{\exprsvar}{\indice}{\tprop}$, whereas, like in  $\ihol$, base comprehension expression with no arguments is denoted by $\ecompbase{\termvar}{\type}{\tprop}$.
Each expression is typed with an index, and we use an \emph{index context} $\icontext$,  to keep track of such declarations $\epred : \indice$, where $\epred$ is an expression variable. % Concretely, an \emph{index context}, denoted by $\icontext$, contains declarations of the form $\epred : \indice$, where $\epred$ is an expression variable.

%intuitively stands for the set of all programs $\termvar$ and properties $\exprsvar$ such that $\tprop$ holds over $\termvar$ and $\exprsvar$.

%The variables $\termvar$ and $\tpred$ are bound in a comprehension term.

The membership specification,  $\tmem{\term}{\exprs_1}{\exprs_2 }$,  intuitively states that $\exprs_1$ (with argument $\exprs_2$) holds on $\term$. When $\exprs_1$ takes no arguments, we use the base expression $\tmembase{\term}{\exprs_1 }$.
As in $\ihol$,  a membership specification is a form of application in that stating that an expression is a member of a comprehension expression amounts to the inner comprehension specification applied to that expression.  
The language of specifications includes standard implication, and has three different forms of universal quantification:
%Our language also contains standard specifications, namely, implication and type (universal) quantification. 
%
% While the former enables quantifying over programs, the latter allows quantifying over properties thereof. 
%There are three different forms of universal quantification.
 type quantification quantifies over programs, index quantification quantifies over properties thereof, and kind quantification allows the specifications to be polymorphic over arbitrary types. 
%\agnote{NOTE: Added kind quantification.}
%
Last, to describe specifications of effectful programs, we use a modality $\after{\term}{\termvar}{\tprop}$, which  %(as in~\cite{pitts1991evaluation}) 
% For this specification to be well-formed, for $x$ of type $\type$, we have that $\term : \typecomp{\type}$, and the specification $\tprop$ in well-scoped in the type context $\termvar : \type$.
intuitively states that the property $\tprop$ holds when $\termvar$ is the result of running the (effectful) program $\term$.
This is the core source of effectful computations in $\hopl$. Specifically, if only considering pure programs, this would collapse into standard substitution.

\begin{figure*}[t] 
\begin{center}
\begin{small}
\input{section/figures/HOPL_rules}
\end{small}
\end{center}
\caption{Typing Rules for $\hopl$ \coqdoc{EffHOL.html\#has_kind}}
% \agnote{ $\termvar : \type \in \tcontext$ in the variable rule collides with the membership symbol}\lcnote{why? it is the set theoretical membership}
    \label{fig:typing-full}
\end{figure*}

\begin{example}\label{ex:norm_pres}
To illustrate the role of each of the logical components,
consider, e.g., a specification $p\downarrow^\type := \timplies{\left(\after{p}{\termvar}{\bot}\right)}{\bot}$ stating that a program $p$ of type $\type$ is not non-terminating.
%($p\downarrow^\type$ is definable in $\effhol$ as $\timplies{\left(\after{p}{\termvar}{\bot}\right)}{\bot} $.)
Using a comprehension term, one can build the first-order expression $\mathtt{norm}_\type\eqdef \ecompbase{\termvar}{\type}{\termvar\downarrow^\type}$ which corresponds to the set of programs of type $\type$ which are not non-terminating, or, in a more type-theoretic fashion, as a function that associates to each such program $p:\tau$ the specification $p\downarrow^\type$. 
Such an expression is of index $\refpredN{\type}$, while in the higher-order case, an index $\refpred{\type}{\indice}$ denotes expressions that define a specification for programs of type $\type$ using an expression of index $\indice$.

For a program of type $\typefun{\type}{\type}$, a natural example of such a higher-order statement is to state that some property that holds for the input also holds for the output. Using comprehension, such property can be formalized using
the expression 
${\mathtt{pres}_\type\eqdef \ecomp{\termvar}{\typefun{\type}{\type}}{\exprs}{\refpredN{\type}}{\tspectypprod{\termvar':\type}{\left(\timplies{\tmembase{\termvar'}{\exprs}}{\tmembase{\termapp \termvar{\termvar'}}{\exprs}}\right)}}}$ of index $\refpred{\typefun{\type}{\type}}{\refpredN{\type}}$.
%For such higher-order expressions, we use the notation $\tmem{\term}{f}{\exprs}$ to denote that $\term$ satisfies the expression $\exprs$ specialized to the sub-expression $f$.
Then, ${\tmem{\left(\termabs{\termvar}{\type}\termvar\right)}{\mathtt{pres}_\type}{\mathtt{norm}_\type}}$ expresses that the identity function preserves normalization.
\end{example}

$\hopl$ supports the following standardly defined capture-avoiding substitutions:~%, which are standardly defined by mutual recursion: 
%\begin{itemize}
%    \item 
$\!A\subst{\typevar := \type}, B\subst{\termvar := \term}, C\subst{\epred := \exprs}$ for $A\!\in\!\{\qu{\type},\qu{\indice},\qu{\term},\qu{\exprs}, \qu{\tprop} \}, B\!\in\!\{\qu{\term},\qu{\exprs},\qu{\tprop} \}$, and $C\!\in\!\{\qu{\exprs}, \qu{\tprop} \}$.

To summarize, the computational components of $\hopl$ are kinds, types and programs. Each type has a kind and each program has a type. 
The logical components of $\hopl$ are indices, expressions and specifications. 
Each expression has an index.
Specifications and expressions depend on both programs and indices, which, in turn, both depend on types, which, in turn, depend on kinds. 
%This is all made concrete next by providing the complete type system of $\hopl$.
%

%\lcnote{add a summary of contexts}

% Our judgments employ three distinct sets of contexts.
% %
% A \emph{kind context}, denoted by $\kcontext$,  contains kind declarations of the form $X : k$. All judgments in $\hopl$ depend on the kind context $\kcontext$ which is used to keep track of the kinds of the free variables in a type.
% %
% A type context, denoted by $\tcontext$, contains type declarations of the form $x : \type$. It is a list of typed variables, denoting terms, where each type may have free variables denoting type constructors, hence a type context always depends on a kind context, which holds all the free variables in the types.
% %
% A index context, denoted by $\icontext$, which contains index declarations of the form $p : \indice$.
% \lcnote{TODO:explain}

\subsection{Operational Semantics and Type System}\label{sec:operational}

% \begin{figure}[t]
%     \centering
%     \small
%     \input{section/figures/HOPL_beta}
%     \caption{Reduction Rules for $\effhol$ \revnote{if I understand correctly, the reduction is a parameter. so maybe this material should go later?}
% %    \revnote{Fig 4: "Conversion", using a symmetric symbol like this makes me think that conversion is closed under symmetry and transitivity, but I don't think that is the case based on how it is used later. Perhaps use a more reduction like symbol?}}
% }
%     \label{fig:beta}
% \end{figure}
%\lcnote{revised this section}
Since $\effhol$ is agnostic to the details of the computational effects, the operational semantics also has to be as generic and modular as possible, allowing for different computational effects to be plugged in.
% Therefore, we allow any choice of a 
% reduction relation $\reduction$ on programs that extends the following (one-step) $\beta$-reduction $\betared$ on programs:
Therefore, we only invoke a minimal  (one-step)
 $\beta$-reduction relation $\reduction$ on programs:
$$\begin{array}{c}
\termbind{\termvar}{\termret{\term_1}}{\term_2} \reduction \term_2 \subst{\termvar:=\term_1}
\\[\sskip]
\termtypeapp{(\termtypeabs{\typevar}{\kind}{\term})}{\type} \reduction \term\subst{\typevar:=\type}
\qquad
\termapp{(\termabs{\termvar}{\type}{\term})}{\termvalue} \reduction   \term\subst{\termvar:=\termvalue}
\end{array}$$
% \begin{minipage}{0.33\textwidth}
% $\begin{array}{@{}l@{\hspace{0.05cm}}l@{\hspace{0.05cm}}l}
%  \text{where:}~~\termvalue & ::= & \termvar \\
%  &\mid& \termtypeabs{\typevar}{\kind}{\term} \\
%  &\mid&
%   \termabs{\termvar}{\type}{\term}
% \end{array}$
% \end{minipage}\\[\myskip]
% for
%  \termvalue  ::= \termvar \\
%  \mid \termtypeabs{\typevar}{\kind}{\term}\\
% \mid \termabs{\termvar}{\type}{\term}
%\end{array}
% For the computational side of $\hopl$, we globally fix a reduction relation $\rightsquigarrow$ on programs at least containing (one-step) $\beta$-reduction $\betared$ on programs as defined in~\Cref{fig:beta}. 
% %
% Since $\effhol$ is agnostic to the details of the computational effects, the operational semantics also has to be as generic and modular as possible, allowing for different computational effects to be plugged in.
%
where $\termvalue ::= \termvar \!\mid \!\termtypeabs{\typevar}{\kind}{\term} \!\mid \!
  \termabs{\termvar}{\type}{\term}$.
%
%As is standard, 
An abstraction applied to a type is reduced to type substitution (in types and expressions).
We define a conversion relation $\conv$  as the (reflexive, symmetric and transitive) contextual closure of $\conv_\type\cup\conv_\exprs$ for 
$$\typeapp{(\typeabs{\typevar}{\kind}{\type_1})}{\type_2}\conv_\type\type_1\subst{\typevar:=\type_2}
\qquad
\typeapp{(\eforall{\typevar:\kind}{\exprs})}{\type}\conv_\exprs \exprs\subst{\typevar:=\type}$$
in all syntactic categories of the equivalence relation induced by these reductions.

%The type system of $\hopl$ is inductively defined by the rules in~\Cref{fig:typing}. %which are mostly the standard rules in $\fomega$ and computational $\lambda$-calculus.
The typing rules of $\hopl$  use four different judgments:
%, for kind context $\kcontext$, type context $\tcontext$ and index context $\icontext$: 
$\eltype{\kcontext}{\type:\kind}~$ for $\type$ being a well-formed type of kind $\kind$ in the context, $\eltrm{\kcontext \! \mid \! \tcontext}{\term}{\type}~$ for $\term$ being a well-formed program of type $\type$, 
$\eltrm{\context}{\exprs}{\indice}~$ for  $\exprs$ being a well-formed expression of index $\indice$, and $\eltype{\context}{\tprop}~$ for $\tprop$ being a well-formed specification.
The type system of $\hopl$ is inductively defined by the rules in~\Cref{fig:typing-full}. %which are mostly the standard rules in $\fomega$ and computational $\lambda$-calculus.
For readability, we elide the typing for $\indice$ being a well-formed index, as those are easily obtained by requiring the constituent types to be of base kind (for full details, see the Coq formalization\vlong{ or the appendix}).
%is inductively defined by a set of typing rules, which are mostly the standard rules in $\fomega$ and computational $\lambda$-calculus.
The typing rules are closed under context extension and under (well-formed) substitution.
Moreover, they are closed under conversion and $\beta$-reduction, and type preservation for programs holds.
%
% we require that any choice of it must satisfy type preservation.\lcnote{do we really need it? It is a major obligation for an instantiation.}\dknote{don't think we need it, can be delegated to the deduction system}
% %As we leave the actual reduction relation $\reduction$ underspecified on purpose, type preservation as in item (2) above cannot be proved but is simply stipulated as a property of $\reduction$.
% Similarly, as soon as $\reduction$ is taken to extend $\betared$ in an realistic way, further standard properties like progress and therefore full type safety can be obtained.
% For example, if we employ a leftmost weak-head call-by-value reduction, we obtain full type safety.
% \emnote{Is it true? (thinking of M(-) expresing non-termination I have some doubt. If that is true it implicitly says that there is no way to capture the untyped setting right? (which doesn't mean we can not define a collapse or something, but that wouldn't be vanilla $\effhol$ directly)}
%\emnote{should be rephrased now that return are no longer values (and also restricted to empty contexts $\tcontext$ I guess}
%\begin{proposition}[Progress]\label{progress}
%The following hold:\lcnote{(checked most) }
%    \begin{itemize}[leftmargin=*]
%        \item If $\eltype{\kcontext}{\type:\kind}$ 
%        then either $\type$ is a value or there exists some $\type'$ such that 
%        $\type \conv \type'$
%        \item If $\eltrm{\kcontext \mid \tcontext}{\term}{\type}$  then either $\term$ is a value or there exists some $\term'$ such that  $\term \betared \term'$
%    \end{itemize}
%\end{proposition}

\subsection{Deductive Apparatus}

\begin{figure*}[t]
\hspace{-0.5em}
\input{section/figures/HOPL_theory-notyping}
\caption{The Theory of $\hopl$ \coqdoc{EffHOL.html\#HOPL_prv}
}
%\revnote{Mem-I: $\tau$ and $\sigma$ are unconstrained? Memo-I: $\tau$ is unconstrained here?}}
%\lcnote{We  need to be uniformly explicit with the context.}}
%For example, mem needs to add x to the type context,right? }}
%\lcnote{in kind uni intro, what's going on in the context? can it be a premise?}}
%\lcnote{we can probably omit the names of the rules and to name them in the text instead }\agnote{We're definitely going to use them in the soundness proof}}
\label{fig:theory}
\end{figure*}

The theory of $\hopl$ is obtained through  inference rules that manipulate sequents of specifications. 
Judgments are of the form $\elsequent{ \context }{\tprops}{\tprop}$, where $\tprops$ is a finite set of specifications.
%
%Sequents in the language are defined with an extra \emph{kind context} $\kcontext$ on which everything else depends.
%
%\begin{definition}[Theory]\label{def:eltheory}
%The theory of $\hopl$ is derived through inference rules between sequents of specifications.
%Sequents in the language are defined with an extra \emph{kind context} $\kcontext$ on which everything else depends. So we have:
%\[ \elsequent{\icontext \mid \kcontext \mid \tcontext}{\tprop_{1} , \ldots , \tprop_{n}}{\psi} \]
%where $\tprop_{1} , \ldots , \tprop_{n} , \psi$ are specifications in contexts $\icontext \mid \kcontext \mid \tcontext$.
%
%\Cref{fig:theory} presents the inference rules that inductively generate the theory of $\hopl$, which, apart from the handling  of modalities and polymorphism, reflect that of $\ihol$.
%
\Cref{fig:theory} presents the inference rules that inductively define the theory of $\hopl$ (omitting the typing premises). 
%
%As in $\hol$, the presentation of the rules omits the typing premises. % (which can be found in the appendix).
The rules closely resemble those of $\ihol$, except for those handling modalities and polymorphism.
Standardly, the theory includes an identity axiom, implication introduction and elimination, and universal introduction and elimination for expressions.  
Additionally, it includes introduction and elimination rules for program and type universal quantification. 
As is standard, in \proguniintrorule~and \expuniintrorule, $\termvar$ and $\epred$, resp., are fresh. 

In addition, $\hopl$ includes rules for modalities. The modality introduction rule, \modintrorule, states that whenever a specification holds for a value, it holds for the computation that does nothing except return that value. The modality elimination rule, \modelimrule, states that a nesting of two modalities, where the specification depends only on the last computation, can be collapsed into a single modality with the same specification over the nesting of the computations.
The Monotonicity rule, \monrule, states that the modality respects deductions, so whenever $\tprop_{1}$ entails $\tprop_{2}$, if $\tprop_{1}$ holds for the value of some computation, so does $\tprop_{2}$. 
%
%As in $\ihol$, 
The (base) membership introduction and elimination rules are completely dual and acount for the mutual dependency of expressions and specifications.
%As in $\ihol$, the rules identify a membership specification  with an `application' of the inner specification in the comprehension expression on the expression, where this so called `application' takes the form of substitution. 
%\emnote{to avoid introducing their counterpart for base comprehension and membership, we abuse notation by assuming that these rules also handle the base case using a dummy substitution $\subst{\unit:=\unit}$ for the predicate variables.}\lcnote{yes! I want to do that everywhere. also in the typing rules. does that make sense or we should just explicitly state those?}
%

Lastly, $\hopl$  includes two computational rules that enable the use of the computational reductions and conversions in the logical side of $\hopl$.
The rule \convrule~allows for replacing specifications that are equivalent by conversions.
The rule \antiredtermrule~enforces a key property of the realizability interpretation, namely the closure of specifications under anti-reduction, by stating that if a program $\term_1$ reduces to a program $\term_2$, then any valid specification for $\term_2$ is also valid for $\term_1$. 
By the assumption that all judgments are well-typed, rule \antiredtermrule~can only be applied  when the reduction preserves typing. 
%In the $\antiredtermrule$ rule the implicit assumption that all judgments are well-typed indicates that it can be applied only when the reduction preserves typing. 
\emnote{This is a strong requirement in that any user-added specifications will have to satisfy this, so normalization is ok, equality of programs is not}\lcnote{indeed}
Standard structural rules are admissible in $\effhol$, with the exception that the substitution rule for programs is only valid for values.

Since $\ihol$ is a mere fragment of $\effhol$, there is a trivial forgetful translation function $\trform{-}$ from the latter to the former, erasing all program, type, and kind structure. 
Thus, all $\effhol$ deductions can be replayed in $\ihol$. 
% \lcnote{add text stating proof irrelevance allowing for the sound erasing translation. we only synthesis the realizers, don't compute with them.}
% \dknote{Note sure what you mean by that, isn't it that we do compute with the realizers but the translation deletes after completely?}
%
% \begin{lemma}[\coqdoc{}]\label{soundesstohol}
%     If\ $\elsequent{\context}{\tprops }{\tprop},\text{ then } \trform\icontext \vdash \trform\tprops \Rrightarrow \trform\tprop $.
% \end{lemma}
In particular, as any specification $\tprop$ expressing a contradiction in $\effhol$ is mapped to a formula $\trform{\tprop}$ expressing a contradiction in $\ihol$. This reduces the consistency of $\effhol$ to that of $\ihol$.

\begin{proposition}[\coqdoc{EffHOL_to_HOL.html\#Consistency}]\label{lem:con}
         $\effhol$ is consistent.
\end{proposition}

\paragraph*{Hoare Triples}
%\lcnote{say a bit more}
%A key goal of $\hopl$ is to provide a framework that can capture a diverse range of effectful programming languages. 
The theory of $\hopl$  allows expressing properties of programs in the familiar style of Hoare logic~\cite{Hoare69}.
% Given a set of specifications $\tprops$ in context $\context$, a type $\type$, a term $\term : \typecomp{\type}$, and a specification $\tprop$ in context $\kcontext \mid \tcontext , \termvar : \type \mid \icontext$, we use the \emph{triple} form
% \[ \htriple{\context}{\tprops}{\termvar}{\term}{\tprop}  \]
% for the judgment 
% \[ \sequent{\context}{\tprops}{\after{\term}{\termvar}{\tprop} }. \]
The known \emph{triple form} is defined as an abbreviation of a sequent with a modality on the right-hand-side:
$$
\ehtriple{\tprops}{\termvar}{\term}{\tprop} \quad \eqdef \quad 
\tprops \Rightarrow \after{\term}{\termvar}{\tprop} 
$$ 
%\revnote{This explanation omits x}
% This intuitively stands for the statement that, assuming all the formulas in $\tprops$ hold, then after computing $\term$, $\tprop$ holds on its value.
Intuitively, this states that, assuming all the formulas in $\tprops$ hold, then after binding $x$ with the result obtained by computing $\term$, the formula $\tprop$ holds. %\emnote{better ? at least referring to $x$}
%We say that $\term$ has \emph{preconditions} $\tprops$ and  a \emph{post-condition} $\tprop$.
The triple form will be used later when discussing the realizability translation. % to show that the translation is sound. 
That is, we show that for every provable theorem of $\ihol$ there is a corresponding program, called the \emph{evidence} of the theorem, for which the translation of said theorem forms a provable \emph{triple} in $\hopl$.
%\emnote{when the judgment is derivable?}\lcnote{why are we talking judgments and not sequents?and if so, shouldn't it be $\sequent{\kcontext \mid \tcontext,\termvar : \type,\term : \typecomp{\type} \mid \icontext }{\tprops}{\after{\term}{\termvar}{\tprop} }$?}

Importantly, the consistency of~\Cref{lem:con} only concerns the core logical system of $\effhol$ and not the derived triples.
That is, while $\bot$ is underivable,  $\after{\term}{\termvar}\bot$ may be derivable for some $\term$s.
Indeed, we allow for instantiations of the computational system, including ones realizing no meaningful logic. \lcnote{we just need to have a fail p}
%A formal way to define instantiations is described in the next section.

\vspace{-0.1cm}
\subsection{Instances of\ $~\hopl$}
\label{sec:instances}

$\hopl$ can be seen as a generic framework relying on several parameters, that can be instantiated 
in different ways to capture a wide range of computational behaviors.
The most direct family of instantiations is given by syntactic translations of $\hopl$ into itself, using its effect-free fragment $\hoplinst$ in the target, i.e., the subsystem of $\hopl$ without the constructs $\typecomp{\type}$, $\termret{\term}$, $\termbind{\termvar}{\term_{1}}{\term_{2}}$, and $\after{\term}{\termvar}{\tprop}$. %\footnote{$\hoplinst$  amounts to a deductive formulation of System $\fomega$ which combines higher-order propositions with the types of System $\fomega$.}%, making propositions use type variables in addition to term variables, and allowing types to use type variables.}
%Concretely, $\hoplinst$ is obtained by identifying any computation type with its underlying type (i.e.,~$M\left(\type\right) = \type$), identifying return with its underlying term (i.e.,~$\termret{e} = e$), identifying bind with term substitution (i.e.,~$\termbind{x}{e_{1}}{e_{2}} = e_{2}\subst{x := e_{1}}$), and identifying the modality with specification substitution (i.e.,~$\after{e}{x}{\tprop} = \tprop \subst{x := e}$).
%
In the target language, we allow any choice of a 
reduction relation $\reduction$ on programs that extends the $\beta$-reduction of $\hopl$. 
That is, different operational semantics for programs can be invoked by, e.g., extending the notion of values or adding congruence rules describing a particular evaluation strategy.  
% if $\hopl$ needs programs to be equipped with an operational semantics, the chosen evaluation strategy is a parameter of each possible instance\lcnote{we have not defined instances yet. need to rephrase this sentence} of the system. 
%
\Cref{sec:examples} provides an example that invokes a call-by-name evaluation strategy.
As standard, meta-properties of $\hopl$, such as type preservation for programs,  might fail in such extensions. 

\begin{definition}[Pure instance \coqdoc{EffHOL_to_Fw.html}]
\label{def:instance}
   A pure instance of $\effhol$ is an interpretation of $\hopl$ in $\hoplinst$ which interprets the non-pure constructs in $\hopl$ as pure. 
   That is, a pure instance of $\effhol$: (1) assigns to each $\type, \term, \term_1, \term_2, \termvar, \tprop$, an interpretation in $\hoplinst$ of $\typecomp{\type}$ as a type, $\termret{\term}$ as a program, $\termbind{\termvar}{\term_{1}}{\term_{2}}$ as a program, and $\after{\term}{\termvar}{\tprop}$ as a specification, such that the typing and inference rules of $\typecomp{\type}$, $\termret{\term}$, $\termbind{\termvar}{\term_{1}}{\term_{2}}$, and $\after{\term}{\termvar}{\tprop}$ are satisfied by their respective interpretations, 
and (2) picks a (potentially extended) evaluation strategy, such that the reduction and conversion rules are preserved.  
% \begin{itemize}[leftmargin=*]
%     \item
    
%     For each type $\type$ in $\hoplinst$, an interpretation of $\typecomp{\type}$ as a type in $\hoplinst$
%   %  $I(M(\tau)) \in Types(\hoplinst)$  
% %    a type~$\icomp{\type}$ in $\hoplinst$, to interpret $\typecomp{\type}$ for each $\type$ in $\hopl$;
%     \item For each term $\term$ in $\hoplinst$, a term~$\iret{\term}$ in $\hoplinst$, to interpret $\termret{\term}$ for each $\term$ in $\hopl$; 
%     \item For each ,  in $\hoplinst$, a term~$\ibind{\termvar}{\term_{1}}{\term_{2}}$ in $\hoplinst$, to interpret $\termbind{\termvar}{\term_{1}}{\term_{2}}$ for each $\termvar$, $\term_{1}$, and $\term_{2}$ in $\hopl$;
%     \item For each term variable $\termvar$, term $\term$, and specification $\tprop$ in $\hoplinst$, a specification~$\iafter{\termvar}{\term}{\tprop}$ in $\hoplinst$, to interpret $\after{\term}{\termvar}{\tprop}$ for each $\termvar$, $\term$, and $\tprop$ in $\hopl$;
%     \item
%     A choice of a (potentially extended) evaluation strategy such that reduction and conversion are preserved.
% \end{itemize}
% where the typing rules and inference rules of $\typecomp{\type}$, $\termret{\term}$, $\termbind{\termvar}{\term_{1}}{\term_{2}}$, and $\after{\term}{\termvar}{\tprop}$ are satisfied by $\icomp{\type}$, $\iret{\term}$, $\ibind{\termvar}{\term_{1}}{\term_{2}}$, and $\iafter{\termvar}{\term}{\tprop}$, respectively.
\end{definition}

% The essential property of pure instances is that they preserve probability. That is, if  $\elsequent{\context}{\tprops}{\tprop}$ is derivable in $\effhol$,  then for every pure instance, the interpreted sequent is derivable in $\hoplinst$ (\coqdoc{EffHOL_to_Fw.html\#HOPL_Fw}). 
%Thus, a pure instance is essentially a syntactic interpretation of $\effhol$ in $\hoplinst$ such that the following lemma holds.
\begin{lemma}[\coqdoc{EffHOL_to_Fw.html\#HOPL_Fw}]
    \label{lem_instance}
    If  $\elsequent{\context}{\tprops}{\tprop}$ is derivable in $\effhol$,   then for every pure instance, the interpreted judgment is derivable in $\hoplinst$.
\end{lemma}

After introducing the realizability interpretation (\Cref{thm:soundness}), \Cref{sec:app} provides illustrative examples of pure instances and showcases how the effectful programs provided by the monad can be used to realize additional logical principles.
In the context of \Cref{thm:soundness} we will also discuss a way in which the computational behavior can be axiomatized within $\hopl$ itself, without an explicit syntactic instantiation.

Interestingly,  the above restricted notion of pure instances is already sufficient for capturing all those examples. 
Nonetheless, one might invoke a more general notion of an instance, allowing for different target languages outside of $\hopl$ itself. 
To define such an instance, one has to choose a monad and an evaluation strategy for programs, and then to define the interpretation of the modality %specifications\emnote{is that a terminology we use elsewhere? }\dknote{found modality in 4.1} 
$\after{\term}{\termvar}{\tprop}$.
Any choice for these parameters can be taken, provided that substitution, reduction and conversion are respected, 
%\emnote{and type preservation?}\dknote{think we can get rid of type preservation everywhere}\lcnote{I am not sure}
and the corresponding typing judgments and logical rules remain valid. %in particular the \modintrorule, \modelimrule\ and \monrule\ rules.
%
%For example, deductive System $\fomega$, which is a combination of System $\fomega$~\cite{Fomega} (a total and purely functional programming language) with higher-order logic, can be recovered as an instance of $\hopl$ by taking the trivial identity monad, allowing us to describe specifications over System $\fomega$ programs. 
%
%We refer to the Rocq code for full detail on such instantiations and their corresponding correctness properties.

\emnote{Do we want to say here that in any instance the theory can be extended with user-defined specifications and deduction rules, as long as they are valid in this instance. It is a bit tautological, but might be worth being observed, for instance one could do as in \Cref{ex:norm_pres} and add a specification $p \downarrow $, which is true for any typed term, since they normalize (if I'm correct, even to represent non-termination using a partiality monad $1+A$, in the monadic approach we still only have terminating computations since a non-terminating one is represented as a term in $1$, no?)  }
\dknote{Added a little remark above hinting at that, maybe a good compromise?}
%\agnote{Shouldn't this be in its own subsection/definition?}
\lcnote{I think it is good to explicitly mention the notion of extension here as Etienne suggested (especially given that I did not get the hint :) ), but I am not sure about space}

%\paragraph{$\lambda$-calculus}\label{ex:pca}\lcnote{todo}
% \lcnote{switch to a syntactic embedding}
% While in $F_\omega$ the meaningful components are the types, in PCAs the types are constant and specifications are the key.
% In fact, given a PCA $\mca$, the class of program types is $\{ \mca \}$ and so there is no need for a kind system.
% To support non-termination we take the monad $M$ to be the subsingleton monad, i.e.  $M A$ is the set of subsets of $A$ in which all elements are equal.
% A specifications in a PCA $\mca$ is a  function of the form $X \rightarrow \P\left(\mca\right)$, where $X$ is a set and $\P\left(\mca\right)$ is the power set of $\mca$.
% \end{example}

% \begin{example}[TBD]
% \lcnote{something in the middle?}
% \emnote{sth like $\lambda$-calculus + some selected effect (state or control for instance)?}
% \lcnote{tiny ML from the EL paper?}
% \lcnote{Haskell?/C?/Ocaml?}
% \dknote{Kreisel realisability or the dialectica interpretation?}
% \end{example}

\section{The Realizability Translation}\label{sec:translations}

%This translation comprises of several other translations relating the various components of the languages. \lcnote{be more concrete when done}
%
%Also, with a trivialising translation in the converse direction, we establish the consistency of $\effhol$.

%\subsection{Realizability Translation: From $\ihol$ to $\effhol$}
\label{sec:HOLtoHOPL}

% \subsubsection{First-Order Fragment}
% base types translation
% \subsubsection{Higher-Order Fragment}
% extending with sorts (predicates), membership (propositions), and comprehension.

This section established how the $\effhol$ framework can be used to model $\ihol$. To this end, we provide a syntactic realizability translation of $\ihol$ judgments into $\effhol$ judgments, which  
%In this section we define the full formal translation of higher-order judgments into program judgments. 
 ensures that the programming language described by $\effhol$ is a realizability model of $\ihol$. 
As mentioned in the overview, the general translation mainly consists of the following four syntactic translations. 
\begin{align*}
   \trkind{-} &\;:\; \mathtt{sort} \to \mathtt{kind}
   &\quad
       \trind{-}{-} &\;:\; \mathtt{sort} \to \mathtt{type} \to \mathtt{index}
    \\
    \trtype{}{-} &\;:\; \mathtt{prop} \to \mathtt{type}
    &\quad
    \ttrspec{\scontext}{-}{-} &\;:\, \mathtt{prop} \to \mathtt{prog}\to \mathtt{spec}
    % \tretype{\scontext}{-}{\tau} &: \trm{\scontext}{\tau} \rightarrow \con{\trkind{\scontext}}{\trkind{\tau}} \\
    % \trtrm{\scontext}{-}{\tau} &: \prod_{\trm{\scontext}{\tau}} \trm{\trkind{\scontext} ; \trind{\trkind{\scontext}}{\scontext}}{\trind{\scontext}{\tau}} \\
\end{align*}
To accommodate the term language with comprehension terms, there are two additional (canonical) translations: %, namely from terms to types and from terms to expressions:
\begin{align*}
    \ttretype{}{-} &\;:\; \mathtt{term} \to \mathtt{type}
    \hspace{1.5cm}
   & 
   \ttrtrm{}{-} &\;:\; \mathtt{term} \to \mathtt{expr}
\end{align*}
The recursive translations are given in~\Cref{fig:trans2}.
%, here we provide some explanations and intuitions.
The proposition-to-specification translation, $\trspec{p}{\sprop}$, is the main realizability translation, converting any $\ihol$ proposition $\sprop$ to a specification describing what it means for a program to realize the proposition.
The realizer is internalized as an extra input $p$ of the appropriate type $\trtype{}{\sprop}$ of realizers.
The sort-to-index translation, $\trind{\type}{\sort}$ reflects the proposition-to-specification translation at the sort level, and the sort-to-kind translation, $\trkind{\sort}$, reflects the proposition-to-types translation at the sort level.
In fact, as the types of sorts and kinds are isomorphic, the last translation is essentially the identity mapping.

The translations of terms are required to translate variables. Since we translate higher-order propositions to higher-order specifications, a variable $\spred$ that appears in a proposition gets translated into a type variable, $\typevar_{\spred}$, that fills the same role variables serve in higher-order logic, namely, to fill in for arbitrary specifications. Similarly, when translating $\spred$ to an expression we use an expression variable, $\epred_{\spred}$.
%
%In HOL we can derive most logical constants from a few basic ones, so the translation will only focus on them. 
% The proposition translation is similar to the one for FOL\lcnote{which one? in Pitts?}\agnote{Did I write this? I don't recall writing that or what it means. Probably should delete it}, with the addition of the membership constructor.
%

The translation of implication states that a realizer of implication takes a realizer of the assumption and computes a realizer of the conclusion, and the type of the realizer reflects this fact as a type of functions from the type of realizers of the assumption to the type of computations of realizers of the conclusion.
The translation of universal quantification says that a realizer of a universally quantified proposition computes a realizer of the proposition for any possible instantiation of the quantified variable, reinterpreted as an assertion variable ranging over type variables of the appropriate kind.
The translation of membership is simply membership in $\effhol$, with the realizer bundled with the element translation and the predicate translation instantiated to the element's type.

% The sorts translations reflect the proposition translations.
% Sorts specify the range of predicate variables, which stand for arbitrary propositions.
% To reflect how every proposition is translated into a type of realizers, every sort is recursively translated into a kind of type constructors for the type of realizers of predicates of that sort.
% The kind of realizers of a predicate over sorts $\sort_{1}$ to $\sort_{n}$ is the kind of type constructors over the kinds of realizers of $\sort_{1}$ to $\sort_{n}$.
% To reflect how every proposition is translated into a specification of realizers, every sort is recursively translated into an index specifying the range of assertions of realizability conditions of predicates of that sort, taking as a subscript the kind of realizers given by the sort to kind translation.
% A predicate variable over sorts $\sort_{1}$ to $\sort_{n}$ is realized by any realizer of any possible type given by instantiating the types of realizers of $\sort_{1}$ to $\sort_{n}$.

% \input{section/figures/translationFig}
\begin{figure*}[t]
  %  \centering
  \hspace*{-0.5cm}
\begin{footnotesize}
   %\scalebox{0.9}{
   \begin{tabular*}{\textwidth}{@{}l@{\hspace{-0.03in}}|@{\hspace{-0.15in}}l@{\hspace{-0.03in}}|@{\hspace{-0.15in}}l@{\hspace{0.5in}}}%{@{}l|@{\hspace{-0.15in}}l|@{\hspace{-0.15in}}l@{\hspace{0.5in}}}
   % \toprule
   %%%%%%%%%%%%%%%%%% SORT -> KIND %%%%%%%%%%%%%%%%%%%%%%
   $ \begin{array}[t]{@{\hspace{-0.1in}}l@{\hspace{0.2in}}l@{\hspace{0.07in}}l@{\hspace{0.07in}}l}
       \multicolumn{4}{c}{ \boldsymbol{\ttrkind{-} : \mathtt{sort} \to \mathtt{kind}}}
        \\[\sskip]
        & \ttrkind{\STAR} &\eqdef& \star
        \\[0.1cm]
        &\ttrkind{\predcon{\sort}} &\eqdef& \typecon{\ttrkind{\sort}}  
    \end{array}$
    &
    $\begin{array}[t]{@{}l@{\hspace{0.2in}}l@{\hspace{0.07in}}l@{\hspace{0.07in}}l}
    %%%%%%%%%%%%%%%%%% TERM -> TYPE %%%%%%%%%%%%%%%%%%%%%%
      \multicolumn{4}{c}{\boldsymbol{\ttretype{\scontext}{-}{\sort} : \mathtt{term} \to \mathtt{type}}}\\[\sskip]
    %%% var
        &\ttretype{\scontext}{\spred}{\sort} &\eqdef& \typevar_{\spred}\\[0.1cm]
    %%% comprehension
        &\ttretype{\scontext}{\comprehend{\spred}{\sort}{\sprop}}{\predcon{\sort}} &\eqdef& \typeabs{\typevar_\spred}{\ttrkind{\sort}}{\ttrtype{\scontext , \spred : \sort}{\sprop}}
        \\[0.1cm]
        &\ttretype{\scontext}{\compbase{\sprop}}{\predcon{\sort}} &\eqdef& {\ttrtype{\scontext , \spred : \sort}{\sprop}}\\[0.5em]
    \end{array}$
&
   $\begin{array}[t]{@{}l@{\hspace{0.2in}}l@{\hspace{0.07in}}l@{\hspace{0.07in}}l}
      %%%%%%%%%%%%%%%%%% TERM -> EXPRESSION %%%%%%%%%%%%%%%%%%%%%%
        % \ttrtrm{\scontext}{-}{\sort}  : \left(\sterm : \mathtt{term}_{\scontext}\left(\sort\right)\right) \to \mathtt{expression}\left(\ttrind{\trtype{}{\sterm} }{\sort} \right)&&&\\
        \multicolumn{4}{c}{\boldsymbol{\ttrtrm{\scontext}{-}{\sort}  : \mathtt{term} \to \mathtt{expr}}}\\[\sskip]
    %%% var
        &\ttrtrm{\scontext}{\spred}{\sort} &\eqdef& \tpred_{\spred}\\[0.1cm]
    %%% comprehension
        &\ttrtrm{\scontext}{\comprehend{\spred}{\sort}{\sprop}}{\predcon{\sort}} &\eqdef& \eforall{\typevar_{\spred} : \ttrkind{\sort}}{
            \tcomp{\termvar}{\ttrtype{\scontext , \spred : \sort}{\sprop}}
               {\tpred_\spred}{\ttrind{\typevar_\spred}{\sort}}
               {\ttrspec{\scontext , \spred : \sort}{\sprop}{x} }
            }     \\[0.1cm]
          &\ttrtrm{\scontext}{\compbase{\sprop}}{\predcon{\sort}} &\eqdef& {
            \tcompbase{\termvar}{\ttrtype{\scontext , \spred : \sort}{\sprop}}
               {\ttrspec{\scontext , \spred : \sort}{\sprop}{x} }
            }
\end{array}$
\\
\midrule
$\begin{array}[t]{@{\hspace{-0.1in}}l@{\hspace{0.2in}}l@{\hspace{0.07in}}l@{\hspace{0.07in}}l}
   %%%%%%%%%%%%%%%%%% SORT -> INDEX %%%%%%%%%%%%%%%%%%%%%%
        \multicolumn{4}{c}{ \boldsymbol{\ttrind{-}{-} : \mathtt{sort} \to \mathtt{type}\to \mathtt{index}}}
        \\[\sskip]
        & \ttrind{\type} {\STAR} &\eqdef& \refpredN{\type} 
        \\[0.1cm]
        &\ttrind{\type} {\predcon{\sort}} &\eqdef& \tindprod{\typevar_{0} : \ttrkind{\sort}}{
        \refpred{\typeapp{\type}{\typevar_{0}}}{\ttrind{\typevar_{0} }{\sort}}}
    \end{array}$
    &
    $\begin{array}[t]{@{}l@{\hspace{0.2in}}l@{\hspace{0.07in}}l@{\hspace{0.07in}}l}
 %%%%%%%%%%%%%%%%%% PROPOSITION -> TYPE %%%%%%%%%%%%%%%%%%%%%%
    \multicolumn{4}{c}{   \boldsymbol{\ttrtype{\scontext}{-} : \mathtt{prop} \to \mathtt{type}}}
        \\[\sskip]
    %%% implication
        &\ttrtype{\scontext}{\simplies{\sprop_{1}}{\sprop_{2}}} &\eqdef& \typefun{\ttrtype{\scontext}{\sprop_{1}} } {\typecomp{ \ttrtype{\scontext}{\sprop_{2}}}}\\[0.1cm]
    %%% forall
        &\ttrtype{\scontext}{\forall_{\spred : \sort}  \sprop} &\eqdef&  \ttypprod{\typevar_\spred :\ttrkind{\sort}} {\typecomp{\ttrtype{\scontext , \spred : \sort}{\sprop}}}\\[0.1cm]% \lcnote{should be X_u'?}\\
    %%% membership
        &\ttrtype{\scontext}{\smem{ \sterm_{1} }{\sterm_{2}} } &\eqdef& \typeapp{\ttretype{\scontext}{\sterm_{2}}{\predcon{\sort}}}{\ttretype{\scontext}{\sterm_{1}}{\sort}}
        \\[0.1cm]
       &\ttrtype{\scontext}{\smembase{\sterm} }{} &\eqdef& {\ttretype{\scontext}{\sterm}{}}
        \\
    \end{array}$
    &
    $\begin{array}[t]{@{}l@{\hspace{0.2in}}l@{\hspace{0.07in}}l@{\hspace{0.07in}}l}
   %%%%%%%%%%%%%%%%%% PROPOSITION -> SPECIFICATION %%%%%%%%%%%%%%%%%%%%%%
      \multicolumn{4}{c}{   \boldsymbol{\ttrspec{\scontext}{-}{-} : \mathtt{prop} \to \mathtt{prog}\to \mathtt{spec}}}
        \\[\sskip]
    %%% implication
        &\ttrspec{\scontext}{\simplies{\sprop_{1}}{\sprop_{2}}}{\term} &\eqdef& \tspectypprod{\termvar_{1} : \ttrtype{\scontext}{\sprop_{1}} }{ \ttrspec{\scontext}{\sprop_{1}}{\termvar_{1}} \supset \after{\termapp{\term}{\termvar_{1}}}{\termvar_{2}}{\ttrspec{\scontext}{\sprop_{2}}{\termvar_{2}}}}\\[0.1cm]
    %%% forall
        &\ttrspec{\scontext}{\forall_{\spred : \sort} \sprop}{\term} &\eqdef& \tspeckinprod{\typevar_\spred : \ttrkind{\sort}} \forall_{\tpred_\spred : \ttrind{\typevar_\spred}{\sort} } .
        \after{\termtypeapp{\term}{\typevar_\spred}}{\termvar_{0}}{\ttrspec{\scontext , \spred : \sort}{\sprop}{\termvar_{0}}}\\[0.1cm]
     %   {\color{red}\subst{ \tpred_{\spred} \eqdef\tpred }}}\lcnote{I think we don't need it}\\
    %%% membership
        &\ttrspec{\scontext}{\smem{ \sterm_{1}}{\sterm_{2}} }{\term} & \eqdef& \tmem{ \term }{ \eapp{\ttrtrm{\scontext}{\sterm_{2}}{\predcon{\sort}}}{\ttretype{\scontext}{\sterm_{1}}{\sort}} }{ \ttrtrm{\scontext}{\sterm_{1}}{\sort} }
        \\[0.1cm]
        &\ttrspec{\scontext}{\smembase{ \sterm}}{\term} & \eqdef& \tmembase{ \term }{ \ttrtrm{\scontext}{\sterm}{\predcon{\sort}}}
        \\
    \end{array}$
\end{tabular*}
%}
\end{footnotesize}  
% \emnote{Due to the way we translate $u$, I think we \textbf{have} to be careful with  name variables $X_u$/$y_u$ (and not only $X$/$y$)}
    \caption{The Realizability Translation \coqdoc{HOL_to_EffHOL.html}}
    \label{fig:trans2}
\end{figure*}

%Next,  we establish the soundness of our translation. 
%Then we provide a translation in the converse direction, thus proving the equi-consistency of $\ihol$ and $\effhol$.
%
%We start by establishing some preservation properties of the translation. 
%First, we show that our translation preserves substitutions. 
%
We lift the sort translations $\trkind{-}$ and $\trind{-}{-}$ to  contexts: % as follows:
%\begin{figure}[H]
%  $$ \centering
%     \begin{tabular}{ccc}
%          $\trkind{\emptyset} := \emptyset$ & ; &$\trkind{\scontext , \spred : \sort} := \trkind{\scontext} , \typevar_{\spred} : \trkind{\sort}$ \\
%          $\trind{\emptyset}{\emptyset} := \emptyset$ & ; & $\trind{\kcontext , \typevar : \kind}{\scontext , \spred : \sort} := \trind{\kcontext}{\scontext} , \tpred_{\spred} : \trind{\typevar : \kind}{\sort}$
%     \end{tabular}
%     $$  
% \emnote{new version }
 $$ \centering
    \begin{tabular}{ccc}
         $\ttrkind{\emptyset} := \emptyset$ &  &$\ttrkind{\scontext , \spred : \sort} := \ttrkind{\scontext} , \typevar_{\spred} : \ttrkind{\sort}$ \\[0.1cm]
         $\ttrind{}{\emptyset} := \emptyset$ &  & $\ttrind{}{\scontext , \spred : \sort} := \ttrind{}{\scontext} , \tpred_{\spred} : \ttrind{\typevar_\spred }{\sort}$
    \end{tabular}
    $$    
%    \lcnote{FIX! the I translation takes a type!}
%    \caption{Context Translation}
%    \label{fig:cxt_trans}
%\end{figure}
%With these definitions in place, %\footnote{This also requires a preservation property over substitutions given in the appendix.}, 
This allows us to prove that our translation preserves well-formedness judgments. 
%
% \begin{lemma}The following rule is derivable/admissible \emnote{\st{might go to the appendices later} useless I think}
% \[\infer[(\text{Monotonicity'})]{\elsequent{\context}{\tprops }{ \after{\term}{\termvar}{\tprop_{2}} } }{\elsequent{\kcontext \mid \tcontext, \termvar:\type \mid \icontext}{\tprops}{ \tprop_{1}\supset\tprop_{2}} 
%          \qquad 
%          \elsequent{\context}{\tprops}{ \after{\term}{\termvar}{\tprop_{1}}}}\]
% \end{lemma}
%
With this we next state and prove the soundness of our translation.
That is, we show that for every sequent $\sprops \Rrightarrow \sprop$ in $\ihol$, our translation produces a computation program %(\ie,~a program that does not realize $\bot \eqdef\typeabs{\typevar}{\kind}{\typevar}$), 
$\term: M\ttrtype{\scontext}{\sprop}$, such that the triple 
 $\ehtriple{\ttrspec{\scontext}{\sprops}{}}{\termvar_{r}}{\term}{\ttrspec{\scontext}{\sprop}{\termvar_{r}}}$ is provable in $\effhol$.

\begin{theorem}[Soundness \coqdoc{HOL_to_EffHOL.html\#soundness}]\label{thm:soundness}
For each $\ihol$ theorem 
$$\scontext \vdash \sprop_{1} , \ldots , \sprop_{n} \Rrightarrow \sprop $$
there is a program $\term$ with
% $$ \term \in \trm{\trkind{\scontext} ; \trind{\trkind{\scontext}}
% {\Gamma}}{\trtype{\Gamma}{\varphi_{1}} \times \hdots \times \trtype{\Gamma}{\varphi_{n}} \rightarrow M\trtype{\Gamma}{\psi} } 
% $$
$$ \eltrm{\trkind{\scontext}\mid \ttrtype{}{\sprops}}{\term}{ M\trtype{}{\sprop} }$$
% \lcnote{$\trtype{\trkind{\scontext}}{\scontext} $ does not type check. The type context does not come from the sort. 
% Also related, I changed $\trtype{\trkind{\scontext}}{\scontext}$ to $\trtype{\scontext}{\scontext}$, to coincide with the definition, but we should make sure the def works}
such that for any collection of specifications $\context\vdash\Phi$:
\[ \htriple{\ttrkind{\scontext},\kcontext
\mid \ttrind{}{\scontext}, \icontext
\mid \ttrtype{}{\sprops} ,\tcontext}
{ \ttrspec{}{\sprops}{},\Phi }
{\termvar_{r}} {\term  }
{\ttrspec{\scontext}{\sprop}{\termvar_{r}}}. \]
where we use $\sprops:=\sprop_1,...,\sprop_n$,  $\ttrspec{}{\sprops}{} := \trspec{\termvar_{1}}{\sprop_{1}} , \ldots , \trspec{\termvar_{n}}{\sprop_{n}}$, and $\ttrtype{}{\sprops} := \termvar_{1} :\ttrtype{\scontext}{\sprop_{1}} , \ldots , \termvar_{n} :\ttrtype{\scontext}{\sprop_{n}}$.
\end{theorem}

In this setting, we can then view the translation as defining a realizability model for $\hol$ by saying that a theorem of $\hol$ is valid whenever there exists a program $p$ 
satisfying the statements of the soundness theorem.
Importantly, the proof is constructive, in the sense that it actually constructs the realizer $p$ from the derivation in $\hol$. In fact, the proof shows how each rule in $\hol$ corresponds to a program construct that is sound for the realizability translation, \ie which builds a realizer of the conclusion from realizers of the premises. As a consequence, the programs extracted from the soundness proof are modular with respect to $\hol$ derivations. In particular, if a specific instance of $\hopl$ allows us to provide an effectful realizer for some proposition of $\hol$ through the translation, then this proposition could be taken as an extra axiom in $\hol$. Indeed, since it has a realizer this proposition would be valid in the induced realizability model, as would any theorem derived using this theorem, by modularity of the extracted realizer. We shall illustrate this in ~\Cref{sec:examples}, using an instance of $\hopl$ corresponding to Krivine realizability, relying on control operators to validate classical reasoning principles.

Importantly,  the program $p$ obtained in the translation is uniform in that it does not rely on the specifics of the monad and the modality. Hence, in particular, it must also work for consistent modalities that do not permit evidence of falsity, i.e., where 
$\elsequent{\term : \typecomp{\type}}{\after{\term}{\termvar}{\bot}}{\bot}$ is provable.
%and thus it cannot be “truly” inconsistent. 
Therefore, if the modality itself is consistent, the only way to give evidence for $\top \Rightarrow \bot$  is if it holds in the meta-theory, thus ensuring the object theory is consistent as long as the meta-theory is. 

By the general formulation of the soundness theorem, the realizability translation of $\ihol$ sequents is valid in the presence of arbitrary well-typed assumptions $\Phi$. %$\context\vdash\Phi$.
This in itself is not surprising, due to weakening, but it offers additional applicativity in that the additional set of assumptions can be used, for example, to introduce new constants in computation types with axiomatized behavior.
This generality enables treating %$\effhol$ as an open-ended system and 
the translation of $\ihol$ as a generic realizability translation into particular computational settings, enabling consistency and independence proofs.
This can be done via, e.g., pure instances of $\hopl$, as the instance interpretation can be composed with the realizability translation to obtain a sound realizability translation into the pure instance of $\effhol$ (\coqdoc{EffHOL_to_Fw.html\#soundness_Fw}).
% in composition with \Cref{lem_instance}:
% \begin{corollary}[\coqdoc{EffHOL_to_Fw.html\#soundness_Fw}]
%     The statement of \Cref{thm:soundness} can be replayed in every $p$-instance of $\effhol$.
% \end{corollary}

%To illustrate this versatility, we can stipulate that there are some types $\type_i$ and programs $\term_i:M(\type_i)$ whose behavior is specified via (consistent) assumptions of the form $\Phi=\{\after{\term_1}{\termvar}{\tprop_1},\dots,\after{\term_k}{\termvar}{\tprop_k}\}$.
%Now suppose, given some $\ihol$ proposition $\sprop$, it turns out to be possible to derive
%$\htriple{}{ \Phi }
%{\termvar_{r}} {\term  }
%{\ttrspec{\scontext}{\sprop}{\termvar_{r}}}$ for some program $\term$ potentially using $\term_i$, i.e.~that under the given assumptions $\term$ happens to act as a realizer for $\sprop$\lcnote{the translation of?}.
%Then from the soundness theorem we obtain that $\sprop$ is consistent\lcnote{what does that mean?}, as otherwise if $\vdash\neg\sprop$ were derivable in $\ihol$, soundness would yield a realizer $\term'$ with $\htriple{}{ \Phi }
%{\termvar_{r}} {\term'  }
%{\ttrspec{\scontext}{\neg\sprop}{\termvar_{r}}}$, thus deriving an inconsistency in the realized logic.

%\lcnote{Also consistency: no realizer for $\bot$.}

\lcnote{Can we have an ``extensibility'' result? allowing us to go from $\effhol_{id}$ to extensions with other programs not necessarily with monads? still guaranteeing soundness}

\lcnote{See in what sense we can claim compositionality : can we have a core type system for which we prove the soundness, and then allow for additional features in some of the components provided a simple check regarding the translation (compare with the rust situation !)}\emnote{tried to do this above}

\lcnote{If we have time and we think it is interesting, add an example/discussion of how we can do pairs, demonstrating the translation and resulting realize. Kindda like what's in the commented old version}
%\input{section/OldpairExm}

%\subsection{Consistency: From  $\effhol$ to $\ihol$}
%\label{sec:HOPLtoHOL}
%\input{section/modelHOPL}

% \input{section/instantiation}
\section{Illustrative Applications}\label{sec:app}
This section illustrates the versatility of the $\hopl$ framework through various applications. First, we highlight its uniformity by seamlessly reproducing classical realizability results. Next, we demonstrate its use in a simple (relative) consistency proof of Markov's Principle. 
Finally, we showcase its generality, enabling robust computational interpretations across diverse effectful instances. 
Due to space limitations, we focus on the first application and briefly outline the other~two.

\subsection{Krivine Classical Realizability}
\label{sec:examples}

This section illustrates how Krivine classical realizability can
be seen as a particular pure instance of $\hopl$. 
Krivine realizability was introduced as a complete reformulation of standard
intuitionistic realizability,  which was inherently incompatible with classical logic, by building on Griffin's seminal observation that the control operator
\callcc~can be typed with Peirce's law~\cite{griffin90,Krivine09}.
Its original presentation uses control operators in a direct-style fashion, based on the $\lambda_c$-calculus, an extension of the $\lambda$-calculus with \callcc. To that end, the corresponding operational semantics is then expressed using processes $\cut{p}{\pi}$ of an abstract machine, where $p$ is a program and $\pi$ is a stack. 
The instruction \callcc~allows programs to backtrack, as shown by its operational semantics:
% \[\begin{array}{rcl}
% \cut{\callcc}{p\cdot \pi}       &\triangleright& \cut{p}{\throw{}{\pi}\cdot\pi}\\
% \cut{\throw{}{\pi}}{p\cdot \pi'}&\triangleright& \cut{p}{\pi}
% \end{array}
% \]
\[
\cut{\callcc}{p\cdot \pi}       \triangleright \cut{p}{\throw{}{\pi}\cdot\pi}\quad
\cut{\throw{}{\pi}}{p\cdot \pi'}\triangleright \cut{p}{\pi}
\]
When applied to a stack $\pi$, \callcc~provides its first argument with a program $\throw{}{\pi}$, which, at any future point, can drop the current stack to restore $\pi$.
Such instructions can be compiled to the pure $\lambda$-calculus using a continuation-passing style (CPS) translation.
At the type level, this translation corresponds to a negative translation embedding classical logic into intuitionistic logic~\cite{griffin90}.
Oliva and Streicher later demonstrated that Krivine realizability can be obtained by combining a standard  intuitionistic realizability interpretation with a CPS translation~\cite{OlivaStreicher08}. 
More generally, this suggests that effectful interpretations can  be defined via an indirect presentation of effects via a well-chosen monad, rather than  relying on a direct-style representation. %\footnote{\emnote{blabla on combination of monads and this only being illustrative, still can go with the direct route}}. 
To illustrate this, we show how Oliva and Streicher's formulation of Krivine realizability can be naturally expressed in our framework using the continuation monad.
We here only sketch certain salient features of Krivine realizability, and for more details, we refer the readers to papers focusing on Oliva and Streicher's approach~\cite{OlivaStreicher08,Miquel11,GardelleMiquey23}.

We define $\hoplk$ as a pure instance of $\hopl$ based on the continuation monad.
Namely, we first define the monad $\typecomp{\type}  \eqdef  \neg \neg \type$ at the level of types
where $\neg \tau$ denotes the type $\typefun{\type}{\bot}$ (using the standard encoding $\bot \eqdef\ttypprod{\typevar:\star}{\typevar}$).
Intuitively, one can think of a computation of type $\neg\neg \type$ as the CPS translation
of a $\lambda_c$-term, therefore waiting for a continuation of type $\neg \type$ which  accounts for the translation of a stack. As such, via the continuation monad a process $\cut{p}{\pi}$ formed by the interaction of a computation $p$ and a continuation $\pi$ simply corresponds to the application $\termapp{p}{\pi}$.

\let\oldtypecomp\typecomp
\renewcommand{\typecomp}[1]{\neg\neg#1}
% \revnote{ there is a lot of confusion here: the CbN nature of Krivine realizability is given by the interpretation of the effects in the source language. Since you use Moggi's monadic encoding, the system you get is actually CbV. It is thus not Krivine realizability you end up with but rather its CbV equivalent.}
% \emnote{I rephrased to adress the reviewer}
We assume programs are evaluated in a call-by-name fashion\footnote{This will have the benefit of easing some definitions and proofs that rely on $\beta$-reductions or substitutions, which are easier to handle in call-by-name.}, taking advantage of the fact that $\effhol$ is parametric with respect to the choice of an evaluation strategy. 
% \lcnote{@Etienne - can we rephrase like this: 
% Then, to coincide with the original Krivine realizability, we take a call-by-name evaluation for programs, taking advantage  of the fact that $\effhol$ is parametric with respect to the choice of an evaluation strategy. }
%\lcnote{I am not sure I got the footnote.}\emnote{made it more explicit}
% \lcnote{I am not sure this will make a lot of sense to the reader}
% Nonetheless, to further draw the comparison with a CPS translation, we should emphasize that the call-by-name evaluation that we pick here only corresponds to the reduction in the target of the CPS translation. In turn, since we use Moggi's monadic encoding of continuations, the underlying CPS itself rather amounts to a call-by-value source system, and as such we will obtain a call-by-value variant of Krivine realizability.  
We define the return and bind of the monad as expected by:
\[
\termret{p}       \eqdef  \termabs{k}{\neg \type}{\termapp{k}{p}}
\quad
\termbind{\,\termvar}{\term_1}{\,\term_2} \eqdef  \termabs{k}{\neg \type_2}{\termapp{\term_1}{(\termabs{x}{\type_1}{\termapp{\term_2}{k}})}} 
 \]
% \[
% \begin{array}{rcl}
% \typecomp{\type} & \eqdef & \neg \neg \type\\
% \termret{p}      & \eqdef & \termabs{k}{\type}{\termapp{k}{p}}\\
% \termbind{\termvar}{\term_1}{\term_2} & \eqdef & \termabs{k}{\neg \type_2}{\termapp{\term_1}{(\termabs{x}{\type_1}{\termapp{\term_2}{k}})}} 
% \end{array}\]
 where $p$ (resp. $p_1$, $p_2$) is of type $\type$ (resp. $\typecomp{\type_1}$, $\typecomp{\type_2}$). It is then easy to verify that they satisfy the associated typing rules. 
%  \begin{equation*}
%  \scalebox{0.9}{
%     \infer{\eltrm{\kcontext \mid \tcontext}{\termret{\term}}{\typecomp{\type}}}{\eltrm{\kcontext \mid \tcontext}{\term}{\type}}
%    \quad
%     \infer{\eltrm{\kcontext \mid \tcontext}{\termbind{\termvar}{\term_{1}}{\term_{2}}}{\typecomp{\type_{2}}}}{
%     %\begin{array}{l}
%     \eltrm{\kcontext \mid \tcontext}{\term_{1}}\typecomp{\type_{1}}\quad
%     \eltrm{\kcontext \mid \tcontext, \termvar : \type_{1}}{\term_{2}}{\typecomp{\type_{2}}}
%     %\end{array}
%     }
%     }
% \end{equation*}

Lastly, we need to instantiate the modality $\after{\termvar}{\term}{\tprop}$. 
For this, we follow the intuitions of Krivine's classical realizability, where realizers are defined using an orthogonality relation with respect to a set of stacks acting as opponents~\cite{Krivine09}. 
In a call-by-value setting, the latter is itself defined by orthogonality to a set of values acting as realizers~\cite{munch09}. 
The orthogonality relation is typically parameterized by a set of processes $\pole$, which intuitively encompasses the set of intended correct computations in the chosen realizability model. Then, the orthogonal of a set of stacks $A$ is defined as the sets of programs that successfully compute in front of any stack in $A$, \ie~$A^\bot =\{p \mid \forall \pi\in A.  \cut p \pi \in \pole\}$.
The realizability translation of an $\hol$ proposition $\sprop$ can be seen
as the expression $\ecompbase{x}{\ttrtype{}{\sprop}}{\ttrspec{}{\sprop}{x}}$ of index $\refpredN{\ttrtype{}{\sprop}}$.
We thus define an orthogonality relation on expressions $e$ of index $\refpredN{\type}$
\[e^\bot \eqdef \ecompbase{x}{\neg \type}{
    \tspectypprod{x':\type}{
        \timplies
            {(\tmembase{x'}{e})}
            {(\tmembase{\termapp x {x'}}{\pole})}
    }
}%\eqno\text{(for any $e$ of index $\refpredN{\type}$)}
\] 
where, for simplicity, we take 
%\footnote{\lcnote{is it necessary @Etienne?}\emnote{no, it's nice but not necessary}The specific definition of $\pole$ will not be used in the proofs of this section thus, any other expression of index $\refpredN{\bot}$ could be used, emphasizing that it is a parameter of the interpretation as is usual in classical realizability.} 
$\pole\eqdef \ecompbase{x}{\bot}{\bot}$.
Observe that if $e$ is an expression of index $\refpredN{\type}$, then $e^\bot$ is of 
index $\refpredN{\neg\type}$, and therefore $e^{\bot\bot}$ has index $\refpredN{\typecomp{\type}}$.
The biorthogonality relation thus allows us to lift expressions on a given type $\tau$ to expressions on the corresponding computational type $\typecomp{\tau}$.
Viewing a well-formed specification $\tprop$ for $\termvar:\type$ as a proposition defining a set of values, for a program $p$ of type $\typecomp{\type}$, the modality $\after{\term}{\termvar}{\tprop}$ can be viewed as expressing that $p$ is a valid computation with respect to $\tprop$, \ie a realizer. Following Krivine realizability, we can define it by biorthogonality:
\[
\after{\term}{\termvar}{\tprop}\eqdef \tmembase{\term}{\ecompbase{\termvar}{\type}{\tprop}^{\bot\bot}}
\] 
\emnote{one reviewer found the former version of the last paragraphe mysterious, I hope it's now clearer} 

\begin{proposition}\label{lm:hopl_classical}
$\hoplk$ is a valid pure instance of $\hopl$. 
% , \textit{i.e.} the following rules are derivable:
% \begin{mathpar}
%     \infer[\modintrorule]
%         {\elsequent{\context}{\tprops}{ \after{\termret{\term}}{\termvar}{\tprop} } }
%         {\elsequent{\context}{\tprops}{\tprop\subst{\termvar := \term}}}

%     \infer[\modelimrule]
%         {\elsequent{\context}{\tprops}{ \after{\left(\termbind{\termvar_{1}}{\term_{1}}{\term_{2}}\right)}{\termvar_{2}}{\tprop} }}
%         {\elsequent{\context}{\tprops}{ \after{\term_{1}}{\termvar_{1}}{ \after{\term_{2}}{\termvar_{2}}{\tprop} } }}
        
%     \infer[\monrule]{\elsequent{\context}{\tprops }{ \after{\term}{\termvar}{\tprop_{2}} } }{\elsequent{\kcontext \mid \icontext \mid \tcontext, \termvar:\type }{\tprops , \tprop_{1}}{\tprop_{2}} 
%          \qquad 
%          \elsequent{\context}{\tprops}{ \after{\term}{\termvar}{\tprop_{1}}}}
% \end{mathpar}
\end{proposition}

Since $\hoplk$ is a valid pure instance of $\hopl$, we can take advantage of the continuation monad to obtain effectful realizers for classical reasoning principles.  
Recall that Peirce's law is defined by $\mathrm{Peirce}\eqdef\sforall{a:\STAR}{\sforall{b:\STAR}{\simplies{(\simplies{(\simplies{a}{b})}{a})}{a}}}$, abusing notations to write $a$ for the proposition $\smembase a$. Through the translations, we get:
\[\scalebox{0.93}{$\ttrtype{}{\mathrm{Peirce}} \! = \typeabs{X}{\star}{\typecomp{\big(\typeabs{Y}{\star}{\typecomp{
    \typefun{(\typefun{(\typefun{X}{\typecomp{Y}})}{\typecomp X})}{\typecomp{X}}
}}\big)}}$}\]
% \[
% \begin{array}{@{}l@{~~}r}
% \ttrtype{}{\mathrm{Peirce}}  = \scalebox{0.9}{$\typeabs{X}{\star}{\typecomp{\big(\typeabs{Y}{\star}{\typecomp{
%     \typefun{(\typefun{(\typefun{X}{\typecomp{Y}})}{\typecomp X})}{\typecomp{X}}
% }}\big)}}$}& \\
% \ttrspec{}{\mathrm{Peirce}}{p}  = & \\
% \quad\scalebox{0.95}{$\tspeckinprod{\typevar_a :\star} \forall_{\tpred_a :\refpredN{\typevar_a}} .
%     \after{\termtypeapp{\term}{\typevar_a}}{\termvar_{a}}{
%         \Big(\big(\tspeckinprod{X_b :\star}{\forall_{\tpred_b :\refpredN{\typevar_b}}
%             \after{\termtypeapp{\termvar_a}{\typevar_b}}{\termvar_{b}}{
%             }    }
%         }$}&\\
% \quad \scalebox{0.95}{$\tspectypprod{\termvar : \typefun{(\typefun{X_a}{\typecomp{X_b}})}{\typecomp {X_a}} }{ \ttrspec{\scontext}{\simplies{(\simplies{a}{b})}{a}}{\termvar} \supset \after{\termapp{\termvar_b}{\termvar}}{\termvar'}{
%             \ttrspec{\scontext}{a}{\termvar'}    
%         }}\big)\Big)$}&\\ 
% \multicolumn{3}{l}{\big(
%             \after{\termtypeapp{\termvar_a}{\typevar_b}}{\termvar_{b}}{
%             \tspectypprod{\termvar : \typefun{(\typefun{X_a}{\typecomp{X_b}})}{\typecomp {X_a}} }{ \ttrspec{\scontext}{\simplies{(\simplies{a}{b})}{a}}{\termvar} \supset \after{\termapp{\termvar_b}{\termvar}}{\termvar'}{
%             \ttrspec{\scontext}{a}{\termvar'}    
%         }}}    
%             \big)\Big)}
% \end{array}
% \]
To obtain a realizer of Pierce's law via the realizability interpretation of $\hol$ in  $\hoplk$, we define the programs:\\
    $\begin{array}{>{\hspace{-2mm}}r@{~}c@{~}l}
      \callcc &\eqdef &
        \termtypeabs{X}{\star}{\termret{
            \termtypeabs{Y}{\star}{\termret{
                \callcc^{X,Y}           
            }}
        }}\\
    \callcc^{\tau,\tau'} &\eqdef &
        
               \termabs{z}{\typefun{(\typefun{\type}{\typecomp{\type'}})}{\typecomp \type}}
                    {\termabs{k}{\neg \type}{\termapp{\termapp{z}{\throw{\type,\type'}{k}}}k}}\\
    \throw{\type,\type'}{k}& \eqdef& \termabs{x}{\type}{\termabs{k'}{\neg \type'}{\termapp{k}{x}}}               
    \end{array}
  $
Observe that $\callcc^{\tau,\tau'}$ and $\throw{\type,\type'}{k}$
are essentially the CPS translations of the corresponding instructions in Krivine's $\lambda_c$-calculus presented earlier (in particular, they have the same computational behavior up to the CPS), while $\callcc$ handles the polymorphic aspect.
Thus, it is no surprise that $\callcc$ serves as a realizer for Peirce's law.
\begin{theorem}\label{thm:callcc} 
The following hold:
\begin{enumerate}[leftmargin=*]
    \item $\eltrm{}{\callcc}{\ttrtype{}{\mathrm{Peirce}}}$
    \item $\elsequent{}{\top}{\after{\termret{\callcc}}x{\ttrspec{}{\mathrm{Peirce}}x}}$
\end{enumerate}
% $$
% \eltrm{}{\callcc}{\ttrtype{}{\mathrm{Peirce}}}
% ~~ \text{and} ~~
% \elsequent{}{\top}{\after{\termret{\callcc}}x{\ttrspec{}{\mathrm{Peirce}}x}}
% $$
\end{theorem}
% \begin{proof}
% The first item follows from a straightforward typing derivation.
% The proof of the second item follows the line of the usual proof of adequacy for \callcc. For any types $\type_a,\type_b$, 
% we first prove that 
%  $\tmembase{\throw{\type_a,\type_b}{k}}{\ecompbase{x}{\typefun{\type_a}{\typecomp{\type_b}}}{\ttrspec{\scontext}{\simplies{a}{b}}{\termvar}}}$ for any $\tmembase{k}{\ecompbase{x}{\tau_a}{\ttrspec{}{a}{x}}^\bot}$.
% This is proved as usual by closure under anti-reduction, since
% for any $p_a$ and $k'$ such that $\tmembase{p_a}{\ecompbase{x}{\tau_a}{\ttrspec{}{a}{x}}}$ and $\tmembase{k'}{\ecompbase{x'}{\tau_b}{\ttrspec{}{b}{x'}}^\bot}$, we have $\throw{\type_a,\type_b}{k}\,p_a\,k' \betared k\,p_a$ and $\tmembase{k\,p_a}{\pole}$. 
% The proof for \callcc~then follows easily, using anti-reduction again to conclude.
% \end{proof}

% \emnote{Do we want to say that we can pull this backwards : that means that we can safely consider an extension of $\hopl$ with axioms stating the existence of a program \callcc such that \Cref{thm:callcc} holds. In particular, since the realizability translation is by essence compatible with such extra-hypothesis, we also get that Peirce's law is consistent with HOL, since we can add it as an axiom which will then be validated by the realizability translation.}
\let\typecomp\oldtypecomp

\subsection{Consistency of Markov's Principle}\label{sec:MP}
\newcommand{\MP}{\mathsf{MP}}
In~\Cref{sec:examples}, since we defined a realizability model that validates Peirce's law, we obtain that $\hol$ is consistent with it and all its fragments, for instance Markov's principle ($\MP$) allowing double negation elimination of $\Sigma_1$ formulas.
In fact, using partiality as a computational effect in $\effhol$, this (relative) consistency result can be obtained in a more direct way.
Note that it is well-known that Kreisel's modified realizability provides a model for the negation of Markov's Principle~\cite{kreisel1958non}, so the principle is in fact independent.

Arguing informally, we stipulate a type variable $\N$ of base kind together with constants $O:\N$, $S:\N\to \N$ and $\mathsf{find}:(\N\to \N)\to  M(\N)$.
To specify their intended behavior, we  assume a set $\Phi$ expressing the usual axioms of (higher-order) arithmetic and including a specification of the form
\[ \ehtriple{\exists n:\N.\, f\,n=O}{\termvar}{\mathsf{find}\,f}{f\,x = O} \]
expressing that $\mathsf{find}$ implements some sort of linear search.
These few ingredients now allow one to quickly explore which consequences the presence of linear search has on the realized logic.
For instance, one could set out to verify that under certain circumstances (for instance assuming $\MP$ on the meta-level) a realizer for the translation of $\MP$, stating that $\neg\neg(\exists n:\N.\, f\,n=O)$ implies $\exists n:\N.\, f\,n=O$, can be constructed from $\mathsf{find}$.
So if indeed there is a $p$ such that
\[ \ehtriple{\Phi}{\termvar_r}{p}{\ttrspec{}{\MP}{\termvar_{r}}} \]
we immediately obtain that $\hol$ cannot derive $\neg\mathsf{ MP}$, as otherwise by~\Cref{thm:soundness} there would be evidence $p'$ with
\[ \ehtriple{\Phi}{\termvar_r}{p'}{\ttrspec{}{\neg\MP}{\termvar_{r}}} \]
yielding contradicting realizability triples. To obtain formal certainty that such a contradiction cannot follow from the axiomatic extension itsef, one can routinely verify that a consistent instance of $\effhol$ %(in the sense \Cref{sec:instances}) 
based on the monad $M(\tau):=\mathbb{N}\to\tau$ implemented by step-indexing allows to define a deterministic search function $\mathsf{find}$ as axiomatized~(cf.~\cite{Richman1983}), and that assuming $\MP$ on the meta-level is consistent, as common in Kleene realizability.

%For an informal example, we could stipulate a type variable $\mathbb B$ of base kind together with constants $\mathsf{true}$ and $\mathsf{false}$ of type $\mathbb B$ as well as $\mathsf{flip}$ of type $M(\mathbb B)$.
%If we now further assume specifications
%$$\Phi := \{\after{\mathsf{flip}}{\termvar}{\termvar=\mathsf{true}},\after{\mathsf{flip}}{\termvar}{\termvar=\mathsf{false}}\},$$
%we characterise a very simple notion of non-determinism with direct impact on the realized logic.
%For instance, following~\cite{cohen2019effects}, in this setting we can obtain a realizer $p$ with \lcnote{$r$/$x_r$?}$\htriple{}{ \Phi }
%{\termvar_{r}} {\term  }
%{\ttrspec{\scontext}{\neg \mathsf{CC}}{\termvar_{r}}}$ for $\mathsf{CC}$ being the axiom of Countable Choice.\lcnote{is it that easy? can we spell out the realizer?maybe this example can be expended with Mem-SCA (As Ariel is working hard on the mere formalization of CC)} From this, we may conclude by the soundness theorem that HOL does not derive $\mathsf{CC}$, given that it is easy to instantiate the addditional structure in the form of a syntactic translation as described in \Cref{sec:instances}.\dknote{should this syntactic instantiation also be sketched?}\lcnote{yes}

\subsection{Memoizing Countable Choice}
\label{sec:memcc}
 This section demonstrates an application of $\hopl$ that relies on its uniformity to identify core computational capabilities needed for realizing specifications. 
We illustrate this using the principle of Countable Choice ($\cc$), %which is a key feature in, for example, constructive real analysis as it unifies the standard constructive formulations of the reals~\cite{ Troelstra+VanDalen:1988,Lubarsky:2007}.
whose validity status drastically changes w.r.t. the underlying effectful capabilities~\cite{cohen2019effects}.
%It was shown in~\cite{cohen2019effects}  that internalizing effects into constructive type theory can drastically change the validity status of $\cc$. 
While deterministic computation guarantees $\cc$, non-deterministic computation can negate $\cc$, and stateful computation can validate $\cc$ via  memoization. %even in the presence of non-determinism.
In fact, different forms of memoization techniques have been used in various settings to prove $\cc$ (\eg,~\cite{Berardi+Bezem+Coquand:jsyml:1998,Herbelin:DC,Blot:barRecursion2013,MiqueyDC18,krivine:2016}), suggesting it is somehow a more robust model of $\cc$.
The generality of $\hopl$  enables the formalization of such a robustness property,  proving that memoization techniques indeed guarantee the validity of $\cc$ even in the presence of other effects.

To simplify the formalization of $\cc$ we again take $\N$ to be a standard encoding of the natural numbers in $\hol$, and use a predicate $\isnat$ verifying that $n$ is a natural number\footnote{This enables relativized quantifiers over $\N$, \ie formulas of the shape $A \equiv \forall n:\N.\isnat \supset A'$ where $\isnat$ acts as a formula realized by the encoding of $n$.}%, thus providing realizers of $A$ with the value of $n$ for which they should realize $A'$.}. 
%We further use the fact that they translate to isomorphic structures via the realizability translation. 
%
For readability, we also use pairs, which can be standardly encoded in $\hol$.
$\cc$ over a type $\type$ can be axiomatized in $\hol$ as:
$$\begin{array}{l@{\hspace{-2.3cm}}c}
\cc:= \sforall{\spred_{1} : \predcon{\N \times \type}}&\\ 
&{\simplies{ \Total\left(\spred_{1}\right) }{ \sexists{\spred_{2} : \predcon{\N \times \type}}{\sconj{\spred_{2} \subseteq \spred_{1}}{\sconj{\Det\left(\spred_{2}\right)}{\Total\left(\spred_{2}\right)}}} }} 
\end{array}$$
where $\Total(\spred) = \forall n : \N.\; \exists i : \type.\; \isnat \supset (n,i)\in \spred $ expresses that $\spred$ is a total relation, $\Det(\spred)$ that $u$ is deterministic and $\spred_1 \mathop{\subseteq} \spred_2$ that $\spred_1$ a sub-relation of $\spred_2$. 
% where $\Total(\spred) = \forall n : \N.\; \exists i : \type.\; \isnat \supset (n,i)\in \spred $ (resp. $\Det(\spred)$, $\spred_1 \mathop{\subseteq} \spred_2$) expresses that $\spred$ is a total relation (resp. deterministic,  $\spred_1$ a sub-relation of $\spred_2$.
% $$\begin{array}{@{}l@{{}={}}l@{\qquad}r@{{}={}}l@{}}
% \Total(\spred) & \forall n : \N.\; \exists i : \type.\; \isnat \supset (n,i)\in \spred 
% \\
% \Det(\spred) & \forall n : \N, i, i' : \type.\; (n,i) \in \spred  \mathrel{\wedge} (n,i) \in \spred'  
% \mathrel{\supset} i =_\tau i'\\
% \spred_1 \mathop{\subseteq} \spred_2 & \forall v.\; v \in \spred_1 \mathrel{\supset} v \in \spred_2 \\
% \end{array}$$

Crucially, in the statement of $\Total(\spred)$ the universal quantification over natural numbers is relativized via the $\isnat$ predicate. Through the translation, when provided with a value $n$, a realizer $e_{\spred_1}$ of $\Total(\spred_1)$ 
will compute a realizer, say $p^i_n$ of $(n,i) \in \spred_1$ for some index $i \in I$ (without specifying what $i$ is).
Yet, in the presence of effectful computations, $e_{\spred_1}$ may not behave as a function (\ie,  different applications of $e_\spred$ to $n$ may result in different $p^i_n$), and the relation may indeed need not be functional. 
To recover a functional sub-relation ${\spred_2}\subseteq {\spred_1}$ with a realizer of $\Total({\spred_2})$, one can use memoization techniques to store the $p^i_n$ at a given location, making sure that for each $n$ only one $p^i_n$ is ever computed.  Those serve as realizers for a subrelation ${\spred_2}\subseteq {\spred_1}$ which is now functional.

% \begin{align*}
% e_{t,l_{0},n}=&m_{i}\leftarrow\text{lookup}\left(l_{0},n\right);\\&m_{i}\ \left(\text{some}\left(i\right)\mapsto\left[i\right];\\
% &\text{none}\mapsto\left(t'\leftarrow t;r\leftarrow t'\;n;\text{update}\left(\left(l_{0},n\right),r\right);\left[r\right]\right)\right)
% \end{align*}

To define a realizer $\emem$ of $\cc$ along these lines, one essentially needs to be provided with computational features that are axiomatized by a monad with state and location, \eg,~\cite{bauer2019algebraiceffects,plotkin2002Notions}.
Again, assume that $\Phimem$ is a set of specifications which axiomatizes the structure necessary for the standard computational constructs of $\lookup$, $\alloc$ and $\update$. 
Then, using memoization as in~\cite{cohen2019effects} (and assuming $\cc$ on the meta-level as done with $\MP$ in the previous example) we can define a program $\emem:\ttrtype{}{\cc}$
such that:
\[\ehtriple{\Phimem}{\termvar}{\emem}{\ttrspec{}{\cc}{x}}\]

The interesting by-products of such an axiomatic presentation is that any instance $\hopl^{\mathsf{mem}}$ of $\hopl$ providing an instantiation of the additional computational features such that the axioms of $\Phimem$ hold will then have the corresponding $\emem$ also realizing of $\cc$. 
In other words, this captures the robustness of the interpretation of $\cc$ based on memoization techniques.
A concrete instance $\hopl^{\mathsf{mem}}$ can be obtained with the state monad and its standard demonic (\ie, necessity) modality, mimicking the mem-SCA framework in~\cite{cohen2019effects}.
However, this suggests that a realizer of $\cc$ can be obtained in a larger class of instantiations, and it remains to be seen if works realizing $\cc$ in different settings, such as~\cite{Berardi+Bezem+Coquand:jsyml:1998,Herbelin:DC,blot22update,MiqueyDC18,krivine:2016} can be obtained as such instances.

\newcommand{\llangle}{\langle\!|}
\newcommand{\rrangle}{|\!\rangle}
\newcommand{\alsu}{{\mathtt{as}}}
\newcommand{\xle}[1]{\xrightarrow{\!#1\!}}
\newcommand{\eid}{e_{\mathtt{id}}}
\newcommand{\asid}{e_{\mathtt{id}}^\alsu}
\newcommand{\ecompos}[2]{#1 \mathop{;} #2}
\newcommand{\acomp}[2]{#1 \mathop{;^\alsu} #2}
\newcommand{\etrue}{e_{\scriptscriptstyle\top}}
\newcommand{\atrue}{e_\top^\alsu}
\newcommand{\efalse}{e_\bot}
\newcommand{\afalse}{e_\bot^\alsu}
\newcommand{\epair}[2]{\llangle #1, #2 \rrangle}
\newcommand{\tpair}[2]{( #1, #2 )}
\newcommand{\efst}{e_{\mathtt{fst}}}
\newcommand{\afst}{e_{\mathtt{fst}}^\alsu}
\newcommand{\esnd}{e_{\mathtt{snd}}}
\newcommand{\asnd}{e_{\mathtt{snd}}^\alsu}
\newcommand{\elambda}[1]{\lambda #1}
\newcommand{\aslambda}[1]{\lambda^\alsu #1}
\newcommand{\eeval}{e_{\mathtt{eval}}}
\newcommand{\aseval}{e_{\mathtt{eval}}^\alsu}
\newcommand{\efa}[1]{\Pi #1}
\newcommand{\asforall}[1]{\bigsqcap^\alsu #1}
\newcommand{\eproj}{e_\Pi}
\newcommand{\asproj}{e_\sqcap ^\alsu}
\newcommand{\inter}{\prod}
\newcommand{\eexists}[1]{\bigsqcup #1}
\newcommand{\einj}{e_\sqcup}
\newcommand{\union}{\bigcup}
\newcommand{\imp}{\supset}
\newcommand{\power}{\mathscr{P}}
\newcommand{\instance}{\textbf{\textrm{H}}}
\newcommand{\Prog}{\mathcal{P}}
\newcommand{\typeval}{T_{\mathsf{val}}}
\newcommand{\erase}[1]{\lfloor #1 \rfloor}
\newcommand{\lift}[1]{{\overline{#1}}}%^{\mathrm{M}}}
\newcommand{\progset}{\mathrm{P}}
\newcommand{\progtype}{\progset}
\newcommand{\progtypex}[1]{\progset}
\newcommand{\progcomp}{\progset}

\newcommand{\efprop}{\Phi_{\mathrm{ef}}}
\newcommand{\efevd}{\mathrm{E}_{\mathrm{ef}}}
\newcommand{\Vprop}{\Lambda_\type}
\newcommand{\sat}{\textsc{\textbf{Sat}}_\type}

\section{From syntax to semantics : the induced evidenced frame}\label{sec:EF}

% Intuitively speaking, an evidenced frame can be regarded as a frame~\cite{johnstone1983point}, or (more accurately) a complete Heyting algebra, in which one provides evidence that an element is smaller than another, and there can be multiple forms of evidence for the same fact.
% In Kleene's realizability model~\cite{kleene1945interpretation}, the idea is that the elements of the frame represent the propositions in the model, and evidence that $\phi_1$~is smaller than~$\phi_2$ (i.e., $\phi_1$~entails~$\phi_2$) is a computation that can convert proofs of~$\phi_1$ into proofs of~$\phi_2$.
% This computation is not necessarily typed though, in that the same computation can be used as evidence of a variety of entailment relationships.
% For example, the evidence~$\eid$ in the following definition serves as evidence that any $\phi$ entails itself.

An approach to unify several established forms of realizability, including effectful ones, was suggested by Cohen, Miquey, and Tate~\cite{CohMiqTat21}.
That work abstracts away the core of various constructions of realizability models through a single structure, called \emph{evidenced frame}, which focuses solely on the relationship between propositions and their evidence, leaving out the computational specifics of particular models.
This is done semantically, by considering any model as an evidenced frame and providing a uniform construction of a realizability tripos from an evidenced frame.
Similarly, $\hopl$ unifies several established forms of realizability, but through a syntactic translation rather than a semantic abstraction, while focusing on effectful programming languages.

Next, we show how evidenced frames can be constructed from our $\hopl$ framework. 
More concretely, we provide a method that takes a pure instance of $\hopl$ and defines an evidenced frame.
%\footnote{\lcnote{@Etienne - since we are not providing the proof here, maybe we can move this comment to the appendix and here simply avoid talking about pairs?}Since evidenced frames are defined using pairs, to avoid irrelevant encoding details,  we here assume that we work in a slight extension of $\hopl$ with pairs $\tpair{p_1}{p_2}$ of programs, projections $\pi_i$, and a product type $\tau_1\times\tau_2$.}
%
%\lcnote{I think we don't have the space for the def. Maybe instead give a short intuitive definition? @Etienne} 
Roughly speaking, an evidenced frame is a triple $( \Phi, E,  \mbox{$\cdot \xle{\cdot} \cdot$} )$, where $\Phi$ is a set of propositions,~$E$ is a set of evidence, and~\mbox{$\phi_1 \xle{e} \phi_2$} is a evidence relation on $\Phi \times E \times \Phi$, with some required relationships between the operations on propositions and the operations on evidence.

The definition of evidenced frames is at its core very semantical, and in particular, the universal implication requires to define a proposition that acts like an intersection on any set of propositions, which the syntax of $\hol$ cannot account for: we can not directly define 
propositions of an evidenced frame as mere $\hol$ propositions.
In turn, as is usual, the set of (semantics evidenced frame) propositions $\efprop$ we are looking for should reflect the structure of the set of realizers we obtain through the realizability translation. For a (closed) $\hol$ proposition $\sprop$, the set of its realizers is given by a set of closed programs $\vdash p: \ttrtype{}{\sprop}$ such that $\ttrspec{}{\sprop}{p}$ holds\lcnote{holds? is provable?}. The evidencing relation $\varphi \xle e \psi$ should reflect the triple $\ehtriple{\ttrspec{}{\varphi}{p}}{x}{\termapp{e}{p}}{\ttrspec{}{\psi}{x}}$. 
% Yet, if we were to define propositions as typed programs, the definition of evidenced frames requires the existence of an evidence $\varphi \xle \eid \varphi$ compatible with any proposition. In the presence of typed programs, $\eid$ should be polymorphic to be well-behaved. In particular, programs should then come with their types in propositions and applying an evidence should define a computation that also produces a program with its types.  This, in turn, would require computational features on types to compute, required by other components of the evidenced frame (\emph{e.g.} transitivity and conjunction).
Nonetheless, as was observed in earlier work, in typed settings such as those coming from modified realizability, one obtains a tripos only if the type system admits a universal type, \ie if the setting is essentially untyped \cite{lietz02}. The present work faces similar issues, and to circumvent this, we consider a type erasure function $\erase{\cdot}$ and define propositions as erasures of set of programs.

Consider a pure instance of $\effhol$, %i.e., a choice of a monad together with a valid modality $\after{p}{x}{\varphi}$ and an evaluation strategy $\reduction$, 
% To ease the presentation, we assume it is given as a translation to $\effhol^*$, 
with $\Prog$ its set of programs (which we assume to be extended with pairs and projections to handle the conjunction).
We denote by $\erase{\cdot}:\Prog\to\Lambda$ the type erasure function from $\Prog$ to the untyped computational $\lambda$-calculus, \ie the function erasing type annotations on variables, type abstractions and type applications. We write $\Lambda$ for the set of $\lambda$-terms equipped with the same reduction (up to erasure) as~$\Prog$, and we write $\termvalue$ for its set of values. Then, propositions are defined as sets of values obtained as erasures of programs, while evidence are (untyped) $\lambda$-terms:
% \vspace{-0.1cm}
\[
\efprop \eqdef \{\erase{\progset}\mid \progset\subseteq\Prog ~\land~ 
%\forall p\in\progset.\erase{p}\in
\erase{\progset} \subseteq\termvalue\}
\qquad\quad
\efevd \eqdef \Lambda
\qquad
\]
% \revnote{I also had trouble following this part. --- from here until "By considering"}Therefore, we relegate types to the meta-theoretic aspect of the interpretations by considering a type erasure function $\erase{\cdot}$ from $\instance$ programs to the untyped computational $\lambda$-calculus (with pairs), \textit{i.e.} the function erasing types annotations on variables, type abstractions and type applications.
% We can then consider as usual saturated sets, \textit{i.e.} sets $X$ such that for any $p_1,p_2\in\progtype$, if $p_1\betared p_2$ and $p_2\in X$, then $p_1 \in X$, 
% to define $\sat\eqdef\{\erase{X} \mid X\subset\progtype \,\text{saturated}\}$. EF propositions are then defined as the union (over any type) of these sets.
% For any $X\subset\progtype$, we write $\erase{X}$ to denote its image through the erasure translation and we say that $X$ is saturated if it is closed under anti-reduction, \textit{i.e.}, if for any $p_1,p_2\in\progtype$, if $p_1\betared p_2$ and $p_2\in X$, then $p_1 \in X$.
% We define $\sat \eqdef\{\erase{X} \mid X\subset\progtype \text{saturated}\}$.

% \vspace{-0.1cm}
As highlighted in~\Cref{sec:examples},  defining the evidence relation requires lifting a proposition $\varphi\in\efprop$ (defined in terms of values) to a set $\lift{\varphi}$ intuitively comprising of computations returning values in $\varphi$.
For this we use the modality, %$\after{p}{x}{\tprop}$, 
which transforms an  expression $e$ of index $\refpredN{\type}$ into an expression $\lift e$ of index $\refpredN{\typecomp{\type}}$ by 
$\lift e \eqdef \ecompbase{x'}{\typecomp{\type}}{\after{x'}{x}{\tmembase{x}{e}}}$.
For a pure instance of $\effhol$, this modality is defined internally as a specification. 
By considering the (meta) set-theoretic counterparts of $\effhol$'s logical constructs (replacing comprehension terms by comprehension, logical membership by membership, and quantification by meta-quantification), the definition of the modality carries over to sets of (erased) programs, thus we define, for any  $A\subseteq \Prog$, the set $\lift A \eqdef \{p\in{\Prog}\mid   \after{p}{x}{x\in A}\}$. We extend this definition along the erasure function by simply taking $\lift{\erase{A}} \eqdef \erase{\lift{A}}$.\lcnote{not sure how to read the rhs of the equation}\emnote{is $\erase{(\lift{A})}$ clearer? Technically the lifting is defined on set of typed values, we thus need to say how it carries over EF propositions aka set of untyped values, hence this dumb (but I think necessary) definition} With this, we can  define the evidence relation to mimic the syntactic realizability translation:
\[\phi_1 \xle e \phi_2\eqdef \forall p_1\in\phi_1.e\,p_1\in\lift{\phi_2}\]

%%% Removed for space saving
\begin{comment} 
As illustrated in the previous section, once carried over sets of programs,
the lifting operation $\lift{(\cdot)} $  satisfies properties coming from the theory of $\hopl$ :
\begin{enumerate}
\item by \monrule, we get that if $A\subset B$, then $\lift A \subset \lift B$,
\item by \modintrorule, we get that if $p\in A$, then $\termret{p}\in \lift A$,
\item by \modelimrule, if $p_1\in \lift{\scompbase{x}{\Prog}{p_2\subst{x_1:=x}\in\lift{A}}}$, then $\termbind{x_1}{p_1}{p_2}\in \lift A$,
\item by \antiredtermrule, $\lift{A}$ is always closed under anti-reduction.
\end{enumerate}
\end{comment}

% We now have all the necessary ingredients to define the evidenced frame induced by the realizability translation:

% \begin{itemize}
% \item propositions are defined as erasures saturated subsets of $\progtype$: $\efprop\eqdef\bigcup_{\tau:\star}\sat$
% \item evidences are erasure of programs originally of type $\tau_1\to\typecomp{\tau_2}$ : $\efevd\eqdef\bigcup_{\tau_1,\tau_2:\star}\bigcup_{p\in\progtypex{\typefun{\tau_1}{\typecomp{\type_2}}}} \erase{p}$
% \item the evidencing relation $\phi_1 \xle e \phi_2\eqdef \forall p_1\in\phi_1.e\,p_1\in\lift{\phi_2}$ mimics the syntactic realizability translation  (where for any $V=\erase{X}\in\sat$, we define $\lift{V}\eqdef \erase{\lift{X}}$).
% \end{itemize}

% \revnote{I think I am missing something basic about realizability and the evidenced frame setup. why is this proposition interesting? what can we conclude from it?}
\begin{theorem}\label{thm:ef}
$(\efprop,\efevd,\cdot \xle \cdot \cdot)$ is an evidenced frame.
\end{theorem}
\vlong{
 \begin{proof}
 The definition of the different components of the evidence frame is along the line of the realizability translation, while the proofs that they satisfy the expected properties use
 properties for the operation $\lift{A}$ coming for the rules
 \modintrorule/\modelimrule/\monrule/\antiredtermrule, in a manner that is
 similar to the use of these rules in the soundness proof of the realizability
 translation (see Appendix \ref{app:ef} for more details).
 \end{proof}
 }
In particular, since any evidenced frame defines a tripos~\cite{CohMiqTat21}, connecting the dots we obtain:
\begin{corollary}
Any  pure instance of $\effhol$ induces a tripos and a realizability topos.
\end{corollary}
In fact, this construction applies to any (not necessarily pure) instance for which the operation $\lift{A}$ to lift sets of values to sets of programs  using the modality satisfies the meta-theoretic counterpart of the rules \modintrorule, \modelimrule, \monrule, \antiredtermrule~for sets of programs. Interestingly, the resulting evidenced frame and tripos  neither match the usual definitions coming from a fully typed realizability setting such as modified realizability (where types would be accounted for in all the components of the evidenced frame), nor the ones in an untyped setting~\cite{CohGruKirMiq25mca}. In a sense, by taking types into account while erasing them, this definition lies somewhere in between these two situations, and the precise connection between them is left for future work.

% \emnote{add discussion on values vs computations}
% \emnote{add discussion on Kreisel tripos for typed programs ?}
\section{Related Work}\label{sec:related}
%We here review works related to various aspects of $\hopl$.

\paragraph*{Evaluation Logic}
Evaluation logic~\cite{pitts1991evaluation} is a deductive system for reasoning about program evaluation, extending~\cite{moggi1991notions} with modalities. % linking propositions to computation results. 
It resembles $\hopl$ but with a simplified polymorphic and kind structure and it uses two modalities. Its higher-order extension~\cite{moggi1995semanticsEL} still lacks polymorphism and a separation of logical and computational components, limiting its ability to support \emph{typed} higher-order realizability.

\paragraph*{Syntactic Realizability}

Syntactic realizability, pioneered by Gödel~\cite{godel1958bisher} and
Kreisel~\cite{kreisel1959interpretation}, translates Heyting Arithmetic into to a simply typed purely functional programming language (System T).
Modified realizability was extended  by Paulin-Mohring~\cite{paulin_mohring89} to verify Rocq's extraction mechanism, with further advancements in~\cite{letouzey2002new,forster2024verified}. 
Realizability as a translation between type systems was also studied in~\cite{bernardy+lasson11}.
%
% The syntactic approach to realizability has been pioneered by Gödel~\cite{godel1958bisher} and
% Kreisel~\cite{kreisel1959interpretation}, suggesting a realizability translation from Heyting Arithmetic to a simply typed purely functional programming language (called System T).
% The translation assigns each proposition with a set of \emph{potential realizers}, and then derives a subset of \emph{actual realizers} out of the potential realizers, by comprehension over formulas describing what it means for a program to be a realizer.
% Later, modified realizability was extended  by Paulin-Mohring~\cite{paulin_mohring89} to prove the soundness of  the proof extraction mechanism of Coq, by relating the extracted programs in $\fomega$ to their original types in the Calculus of Constructions (CC).
% The Coq extraction mechanism and its correctness have been extended to incorporate inductive types by Letouzey~\cite{letouzey2002new} and only very recently been verified in Coq itself by Forster, Sozeau, and Tabareau~\cite{forster2024verified}.
% This presentation of realizability as a translation between type systems was further studied by Bernardy and Lasson, to connect it with parametricity translations~\cite{bernardy+lasson11}.
% Our work extends these frameworks in another direction, by allowing the realizability interpretation of higher-order logic through programming languages which exhibit computational effects.
In the context of interactive realizability, Birolo~ \cite{birolo2013interactive} introduces monadic realizability %with a type translation and a realizability relation between formulas and realizers. % typed by their type translation.
%
% However, that framework is merely used as a stepping stone for discussing interactive realizability, using the exception-state monad. Furthermore, it only handles first-order Heyting arithmetic, while our framework can handle higher-order logic.
%However, this framework
which focuses on the exception-state monad and is restricted to first-order Heyting arithmetic.
Our work extends these frameworks to higher-order logic using effectful programming languages.

% \paragraph{Classical Realizability.}
% Most works on classical realizability~\cite{Krivine09} are based on Krivine abstract machines~\cite{TODO}.
% This, however, offers an entirely different model of computation than the one used here.
% By using the continuation monad, in a way similar to the example in~\Cref{sec:examples}, we can put classical realizability on a computational model similar to that of intuitionistic realizability and explore its relation to the one based on Krivine abstract machines.
% \emnote{I think this is already explained with more detailed in sec 6. Also, the exploration has somewhat been made in already 3 papers I think\cite{OlivaStreicher08,Miquel11,GardelleMiquey23}. I'd be in favor of dropping this paragraph and just a sentence in the "effectful one"}

\paragraph*{Effectful Curry-Howard}
%Although such effectful implementations of axioms can be integrated into extensional type theories, integrating them into intensional type systems is more challenging due to their heavy reliance on normalization for decidability.
%it is not so clear how to integrate them into intensional type systems due to their heavy reliance on normalization for decidability. 
%
Many recent works extend formal systems with logical principles and computational effects, e.g.~\cite{jaber2016definitional,Boulier+Pedrot+Tabareau:cpp:2017,Pedrot+Tabareau:esop:2018,Pedrot+Tabareau+Fehrmann+Tanter:icfp:2019,Pedrot+al:lics:2020,Pedrot+Tabareau:popl:2020,baillon_et_al:2022,Escardo:2013,Coquand+Jaber:2010,Coquand+Jaber:2012,FCS2018,Cohen+Rahli:fscd:2022,Bickford+Cohen+Constable+Rahli:csl:2021,continuityRef2023,cohen23inductivecontinuity,bauer2022reals}. 
In particular, P\'edrot and Tabareau~\cite{Pedrot+Tabareau:popl:2020}~prove that any \emph{observably} effectful type theory (with other standard properties) is inconsistent. 
While these extensions are usually guided by specific principles or computational capabilities such as exceptions, our work defines a generic syntactic translation for $\hol$, without dependencies on or direct references to the effectful realizers at play.
Recently,  PCAs were extended to Monadic Combinatory Algebras (MCAs) and used to derive an alternative formulation of effectful realizability~\cite{CohGruKirMiq25mca}. 
However, while $\effhol$ yields a typed syntactic notion of realizability, MCAs lead to an untyped semantic notion of realizability. %, in the spirit of PCA-based realizability.

\lcnote{Krivine realizability}
\lcnote{dump tons of effectful realizability works}

%\subsection{Added after review}

% \paragraph*{Interactive Realizability}
% %A related effectful realizability framework is described in \cite{birolo2013interactive}, in the context of Interactive Realizability. 
% %In chapter 3, Birolo described \emph{Monadic Realizability}, where he provides a type translation for formulas along with a realizability relation between formulas and realizers typed by their type translation.
% In the context of Interactive Realizability, Birolo~ \cite{birolo2013interactive} introduces Monadic Realizability with a type translation and a realizability relation between formulas and realizers. % typed by their type translation.
% %
% % However, that framework is merely used as a stepping stone for discussing interactive realizability, using the exception-state monad. Furthermore, it only handles first-order Heyting arithmetic, while our framework can handle higher-order logic.
% However, this framework focuses on the exception-state monad and is restricted to first-order Heyting arithmetic, while $\effhol$ generalizes realizability to higher-order logic and supports broader computational effects.

% \paragraph{Predicate Transformer Semantics}
% \cite{dijkstra1975guarded}
% \agnote{Hoare triples and Weakest precondition}

\lcnote{TBD: cite and compare against~\cite{VSJ25carte}}

%\paragraph*{Dijkstra Monads}\lcnote{TBD: remove this completely and just cite~\cite{ahman2017dijkstra,maillard2019dijkstra}somewhere }
%Another framework relating monads with effects is the one of Dijkstra monads~\cite{ahman2017dijkstra,maillard2019dijkstra}. 
%In particular, the semantics of the F* language \cite{swamy2013verifying} are based on this notion. Dijkstra monads over the $W^{\text{Pure}}$ specification monad are equivalent to monads equipped with a (strict) modality, as in the semantics of Evaluation Logic, as well as HOPL.\lcnote{we need to differentiate ourselves from these works. The reviewer says "They also use a monad and a kind of post-condition modality to interpret effectful arrows in their system, but they actually have a richer setting."}
\section{Conclusion and Future Work}\label{sec:conc}
%\lcnote{TBD: SHORTEN!!!}
In this paper, we presented a purely syntactic account of effectful realizability within a higher-order setting. 
$\effhol$ is a \emph{highly expressive}, \emph{unified} system for \emph{internally} reasoning about \emph{effectful} program logics in a \emph{natural} manner.
The reasoning is done internally as the language of $\effhol$ includes programs that act as realizers for $\hol$ theorems via the syntactic realizability translation. % and get carried throughout the judgment of the framework.
%
% The uniformity of $\hopl$ stems from the fact that it is parameterized by a monad, which in turn holds the  effectful behaviour of the language.
$\hopl$'s strength lies in its parameterization by a monad, capturing effectful behavior, and its natural integration of standard programming language features, such as typed realizers and effectful programs, enabling reasoning about realizers in a way akin to reasoning about programs.
%
%  $\effhol$'s naturality stems from the fact that it contains internal support for standard program language features, for example, by having typed realizers and effectful programs, making our framework applicable to a broader range of languages
% and allowing for reasoning about realizers in a manner similar to the way one would reason about programs.
% That is, we provided a  
% sound syntactic realizability translation from $\ihol$ into $\effhol$, which is a \emph{unified}, \emph{highly expressive} system for \emph{internally} reasoning about \emph{effectful} languages in a \emph{natural} manner.

% Systems like $\fomega$ and computational $\lambda$-calculus offer two examples of languages that can be internalized in $\effhol$ by trivializing one element of the framework, either the computational aspect or the logical aspect.
% But $\effhol$ offers support for more subtle languages, and one should be able to plug into $\effhol$ any programming language of choice and (using the type-constructor polymorphism) generate a realizability model for $\ihol$.
\lcnote{$\fomega$ and computational $\lambda$-calculus offer two examples of languages that can be internalized in $\effhol$ by trivializing one element of the framework, either the computational aspect or the logical aspect. 
Moreover, the examples we present in this paper all have a trivial kind system. 
But $\effhol$ offers support for more subtle languages, and
one should be able to plug into $\effhol$ any programming language of choice and (using the type-constructor polymorphism) generate a realizability model for $\ihol$.}
\lcnote{Many programming languages do not have types, let alone kinds, but $\effhol$ can be fit onto even untyped languages}

% While this paper focuses on syntactic realizability,   completing the picture requires complementing the work with a semantic exploration of $\hopl$. 
% For this we plan (and already started) to develop a categorical semantics for $\hopl$.
%
% \paragraph{Uniform Handling of  Effects}
% To explore the breadth of the framework, realizability will be investigated under several specific examples of computational effects, specified by their corresponding monads. \lcnote{examples. too semantic?}
%
% For $\effhol$ to be closer to program logic we can  remove return, and restrict specification  to just variables rather than arbitrary expressions. (like in dynamic logic)
% For both mod-intro and representability, you can just use another program variable instead of an arbitrary program expression. Specification can be restricted to just program variables.
% \rtnote{Right now you’re using computational types to keep all programs “pure” so that you can use substitution, ie $\phi[x := e]$.
% There’s only one rule there that uses substitution: the rule for return (also term uni elim). Just restrict that to variables, and you’re good.
% Similarly, restrict the grammar for “specification” to program variables.}
%
%
%\input{section/semanticDiscussion}

Since this paper focuses on syntactic realizability, due to space constraints, we leave the presentation of a categorical semantics for $\hopl$ to future paper. This involves combining the semantics of Higher-Order Logic, System $F_\omega$, and Evaluation Logic, using indexed categories for polymorphism~\cite{seely1987categorical}, a strong monad for computational types, along with  a tripos and a T-modality for the logical components, as in~\cite{moggi1991notions, pitts1991evaluation}. %The full construction, ensuring consistency and soundness, will be detailed in a subsequent paper.

Future work includes direct effect handling beyond monads in $\hopl$. This can be done by replacing the dependency on the monad by restricting some language constructs to variables in the spirit of dynamic logic~\cite{Harel2000dynamic}.
Additionally, applying $\hopl$ to a case study in a complex programming language with effects is planned, though this requires addressing auxiliary details, which we reserve for future efforts.
We also note that while $\effhol$ follows traditional PCA-based realizability by invoking a  CbV evaluation strategy~\cite{hofstra2004partial,Lepigre16,Richman1983,seely1987categorical}, we plan to explore alternative strategies like call-by-push-value~\cite{levy1999call}, which may reveal new computational behaviors.

\lcnote{smt hinting at future work looking into separation logics?/program logics? beyond monads}
\lcnote{
Note that in $\effhol$ effects are being handled through monads. While standard and structurally convenient, the handling of effects can be done in a more direct fashion, which can also facilitate the combination of multiple effects.
This can be done by removing the dependency on the monad via computational types and the return program, and instead restricting some of the language constructs to variables, rather than arbitrary terms or expressions, in the spirit of dynamic logic~\cite{Harel2000dynamic}.}

\lcnote{We understand that a case study of a specific programming language with effects could have added to the paper. Unfortunately, we currently do not have one to provide. Obtaining such an example will require substantial amount of work in handling many auxiliary details that are irrelevant to the specific effect we wish to capture.  Accordingly, relating the framework to program development in a realistic programming language is something we left for future work.}

A particularly interesting example is traditional realizability interpretations \textit{à la} Kleene, relying on partial combinatory algebras (that is a partial variant of the untyped $\lambda$-calculus).
However, since the very core of such models is untyped, obtaining them as an instance of $\effhol$ in a natural way requires further structure and is left for future work.
%
%This can be done, e.g., using a ``dynamic'' type $D$, along with an equivalence $D \cong \typefun{D}{\typecomp{D}}$, as the untyped $\lambda_{c}$ in \cite{moggi1989computationallamdba}, for $M$ reflecting non-termination.
%by trivializing the type system (to account for untyped programs) and considering a monad reflecting non-termination (to account for partiality).
% \newpage
%%
\section*{Acknowledgment}
We are deeply grateful to Ross Tate for his valuable ideas and insights, which significantly shaped the direction of this work.
We also thank the anonymous reviewers for their constructive feedback and thoughtful suggestions, which greatly improved the clarity and quality of this paper.
\bibliographystyle{IEEEtran}
\bibliography{ref}

\vlong{
\newpage
\onecolumn
\begin{appendices}
%\subsection*{Appendix}
%\appendix

As noted, most of the syntactic results, and in particular, the systems $\ihol$, $\effhol$ and the syntactic realizability translation between them have been formalized and verified in Coq.
This appendix presents some omitted details and proofs. 

\section{\text{$\hol$}}
\Cref{fig:hol-typing} provides the complete set of typing rules for $\hol$ and the theory with explicit typing premises in the inference rules. 
\begin{figure}[h]
    \centering
\begin{tabular}{@{}c@{\hspace{0.15in}}c@{\hspace{0.15in}}c}
    \infer{\eltype{\scontext}{\spred : \sort}}{\left(\spred : \sort\right) \in \scontext}
        &
    \infer{\eltype{\scontext}{ \comprehend{\spred}{\sort}{\sprop} : \predcon{\sort} } }{ \elprop{\scontext , \spred : \sort}{\sprop  }}
&
    \infer{ \elprop{\scontext}{\smem{\sterm_1}{\sterm_{2}}} }{ \eltype{\scontext}{\sterm_1 : \sort} \quad \eltype{\scontext}{\sterm_2 : \predcon{\sort}} }
\end{tabular}
\\[\sskip]
\begin{tabular}{@{}c@{\hspace{0.2in}}c@{\hspace{0.2in}}c@{\hspace{0.2in}}c}
    \infer{ \elprop{\scontext}{\smembase{\sterm}}}{ \eltype{\scontext}{\sterm : \STAR} }
&
   \infer{ \eltype{\scontext}{\compbase{\sprop}}:\STAR}{ \elprop{\scontext}{\sprop} }
&
    \infer{\elprop{\scontext}{\simplies{ \sprop_{1} }{ \sprop_{2} }}}{\elprop{\scontext}{\sprop_{1}} \quad \elprop{\scontext}{\sprop_{2}}} 
&
    \infer{\elprop{\scontext}{\tspecindprod{\spred : \sort}{\sprop}}}{\elprop{\scontext, \spred : \sort }{\sprop}}
    \\[\myskip]
    
\end{tabular}
\par\noindent\rule{0.6\textwidth}{0.7pt}
\\[\myskip]
\begin{tabular}{@{}c@{\hspace{-0.07in}}c}
        $\infer[\!\!\!\impintrorule]{\selsequent{\scontext}{\sprops}{\simplies{ \sprop_{1} }{ \sprop_{2} }}}{\selsequent{\scontext}{\sprops,\sprop_{1}}{\sprop_{2}} }$ &
   %      Implication Intro & $\dfrac{ \elsequent{}{\tprops,\tprop_{1}}{\tprop_{2}} }{ \elsequent{}{\tprops}{\tprop_{1} \supset \tprop_{2}} }$\\[\myskip]
        $\infer[\!\!\!\impelimrule]{\selsequent{\scontext}{\sprops}{\sprop_{2}}}{\selsequent{\scontext}{\sprops}{\simplies{\sprop_{1}}{\sprop_{2}}} \quad \selsequent{\scontext}{\sprops}{\sprop_{1}}}$\\[\sskip]
    %     Implication Elim & $\dfrac{ \elsequent{}{\tprops}{\tprop_{1} \supset \tprop_{2}} \qquad \elsequent{}{\tprops}{\tprop_{1}} }{ \elsequent{}{\tprops}{\tprop_{2}} }$\\[\myskip]

        $\infer[\!\!\!\holuniintrorule]
        {\selsequent{\scontext}{\sprops}{\tspecindprod{\spred:\sort} \sprop}}
        {\selsequent{ \scontext, \spred : \sort}{\sprops}{\sprop}}$ %\lcnote{y should be fresh here? ie not in $\tprops$}
        &
        $\infer[\!\!\!\holunielimrule]
        {\selsequent{ \scontext}{\sprops}\sprop\subst{\spred:=\sterm}}
        {\selsequent{\scontext}{\sprops}{\tspecindprod{\spred:\sort} \sprop}
        & \eltype{\scontext}{\sterm:\sort}
        }$
        \\[\sskip]

        $\infer[\!\!\!\memintrorule]
        {\selsequent{\scontext}{\sprops}{ \smem{\sterm}{\comprehend{\spred}{\sort}{\sprop}}}}
        {\selsequent{\scontext}{\sprops}{\sprop\subst{\spred := \sterm}}& \eltype{\scontext}{\sterm:\sort}
        }$
        & \quad
        $\infer[\!\!\!\memelimrule]
        {\selsequent{\scontext}{\sprops}{\sprop\subst{\spred := \sterm}}}
        {\selsequent{\scontext}{\sprops}{ \smem{\sterm}{\comprehend{\spred}{\sort}{\sprop}}}}
        $\\[\sskip]
\end{tabular}
\\
    \begin{tabular}{c@{\hspace{0.2in}}c@{\hspace{0.2in}}c}
        $\infer[\!\!\!\memunintrorule]
        {\selsequent{\scontext}{\sprops}{ \smembase{\compbase{\sprop}}}}
        {\selsequent{\scontext}{\sprops}{\sprop}}$
        &
        $\infer[\!\!\!\memunelimrule]
        {\selsequent{\scontext}{\sprops}{\sprop}}
        {\selsequent{\scontext}{\sprops}{ \smembase{\compbase{\sprop}}}}
        $
        &
        $\infer[\!\!\!\idrule]{ \selsequent{\scontext}{\sprops}{\sprop} }{\sprop\in \sprops
         &
         \eltype{\scontext}{\sprops,\sprop}}$
    \end{tabular}
    \caption{Typing Rules and Full Theory for $\hol$ \coqdoc{HOL.html}}
    \label{fig:hol-typing}
\end{figure}

~\newpage
\section{Effectful Higher-Order Logic}\label{app:conv}
\begin{figure}[h]
\input{section/figures/HOPL_beta_full}
\caption{Reduction and conversion in $\effhol$ (full definition) \coqdoc{EffHOL.html}}
\label{fig:full_conv}
\end{figure}

% \begin{figure}[t] 
% \begin{small}
% \input{section/figures/HOPL_rules}
% \end{small}
% \caption{Typing Rules for $\hopl$ \coqdoc{EffHOL.html\#has_kind}}
% % \agnote{ $\termvar : \type \in \tcontext$ in the variable rule collides with the membership symbol}\lcnote{why? it is the set theoretical membership}
%     \label{fig:typing-full}
% \end{figure}
\Cref{fig:full_conv} provides the complete definition of the conversion relations in $\effhol$.
By induction on typing derivations, it is easy to verify, that the typing rules are closed under context extensibility and under (well-formed) substitution.

\begin{lemma}[Context Weakening \coqdoc{EffHOL.html\#has_kind_weak}]
Let $\varmathbb{C}$ be any combination of contexts and let $\varmathbb{C}'$ be an extension with fresh variables of one or more of the contexts in  $\varmathbb{C}$.
If $\eltype{\varmathbb{C}}{\mathcal {J}}$, then $\eltype{\varmathbb{C}'}{\mathcal {J}}$ for $\mathcal{J}\in \{ \lq \type:\kind \rq, \lq\term:\type\rq, \lq\exprs :\indice\rq, \lq\tprop \rq\}$.
\end{lemma}

\begin{lemma}[Judgement Substitution \coqdoc{EffHOL.html\#has_kind_subst}]
The following hold:
\begin{enumerate}[leftmargin=*]
    \item 
If     $\eltype{\kcontext, \typevar:\kind'}{\type:\kind}$ and 
$\eltype{\kcontext}{\type':\kind'}$
then 
$\eltype{\kcontext}{\type\subst{\typevar:=\type'}:\kind}$
    \item 
If     $\eltrm{\kcontext, \typevar:\kind \mid \tcontext, \termvar:\type'}{\term}{\type}$ and 
$\eltype{\kcontext}{\type':\kind}$ and $\eltype{\kcontext,\typevar:\kind \mid \tcontext}{\term':\type'}$
then 
$\eltrm{\kcontext \mid \tcontext}{\term \subst{\typevar:=\type'}\subst{\termvar:=\term'}}{\type\subst{\typevar:=\type'}}$
    \item 
If     $\eltrm{\kcontext, \typevar:\kind \mid \tcontext, \termvar:\type
\mid \icontext, \epred: \indice' }{\exprs}{\indice}$ and 
$\eltype{\kcontext}{\type':\kind}$ 
and $\eltype{\kcontext ,\typevar:\kind\mid \tcontext}{\term':\type}$
and
$\eltrm{\kcontext ,\typevar:\kind\mid \tcontext,  \termvar:\type \mid \icontext}{\exprs'}{\indice'}$
then 
$\eltrm{\context}{\exprs \subst{\typevar:=\type'}\subst{\termvar:=\term'}\subst{\epred:=\exprs'}}{\indice\subst{\typevar:=\type'}}$
\item 
If $\eltype{\kcontext, \typevar:\kind \mid \tcontext, \termvar:\type
\mid \icontext, \epred: \indice }{\tprop}$  
and 
$\eltype{\kcontext}{\type':\kind}$ 
and $\eltype{\kcontext ,\typevar:\kind\mid \tcontext}{\term':\type}$
and
$\eltrm{\kcontext ,\typevar:\kind\mid \tcontext,  \termvar:\type \mid \icontext}{\exprs'}{\indice}$
then 
$\eltype{\context}{\tprop \subst{\typevar:=\type'}\subst{\termvar:=\term'}\subst{\epred:=\exprs'}}$
\end{enumerate}
\end{lemma}

\begin{lemma}[Subject Reduction \coqdoc{EffHOL.html\#red_has_type}]
        If $\eltrm{\kcontext \mid \tcontext}{\term_1}{\type}$  and $\term_1 \betared \term_2$, then $\eltrm{\kcontext \mid \tcontext}{\term_2}{\type}$.
\end{lemma}

%\subsection{Forgetful translation from $\effhol$ to $\hol$}
\newpage
We give here the trivial forgetful translation function $\trform{-}$ from $\effhol$ to $\ihol$ used in \Cref{lem:con}: 
    \begin{center}
    $\def\arraystretch{1}
\begin{array}{l@{\hspace{0.5in}}l}
        \hspace{-0.13in}
\begin{array}{l@{\hspace{0.1in}}l@{\hspace{0.1in}}l}
\trform{\refpredN{\type}} &:=& \STAR\\
\trform{\refpred{\type}{\indice}} &:=& \predcon{\trform \sigma}\\
\trform{\textstyle\tindprod{\typevar:\kind}{\indice}}&:=&\trform{\indice}\\
\trform{\epred} &:=& \spred_\epred\\
\trform{\ecomp{\termvar}{\type}{\exprsvar}{\indice}{\tprop}}
    &:=&\comprehend{\spred_\epred}{\trform \indice}{\trform\tprop}\\
\trform{\ecompbase{\termvar}{\type}{\tprop}}
    &:=&\compbase{\trform\tprop}\\
\trform{\eforall{\typevar:\kind}{\exprs}} &:=& \trform \exprs\\
\trform{\eapp{\exprs}{\type}} &:=& \trform \exprs\\
\end{array}
&
\begin{array}{l@{\hspace{0.05in}}l@{\hspace{0.05in}}l}
\trform{\tmem{\term}{\exprs_1}{ \exprs_2}}
    &:=& \smem{\trform{\exprs_1}}{\trform{\exprs_2}}\\
    \trform{\tmembase{\term}{\exprs}}
    &:=& \smembase{\trform{\exprs}}\\
    \trform{\timplies{ \tprop_1 }{ \tprop_2 }}
    &:=& \simplies {\trform{\tprop_1}} {\trform{\tprop_2}}\\
    \trform{\after{\term}{\termvar}{\tprop}} &:=& \trform{\tprop}\\
    \trform{\tspeckinprod{\typevar:\kind}{\tprop}} &:=&\trform{\tprop}\\
    \trform{\tspectypprod{\termvar:\type}{\tprop}} &:=&\trform{\tprop}\\
   \trform{\tspecindprod{\epred:\indice}{\tprop}}
    &:=& \sforall{\spred_\epred:\trform \indice}{\trform \tprop}
\end{array}
\\~
        \end{array}$
    \end{center}

\begin{lemma}[Judgement Preservation \coqdoc{EffHOL_to_HOL.html\#trans_form_is_prop}] The following hold.
\begin{itemize}[leftmargin=*]
    \item If \ $\eltrm{\context}{\exprs}{\indice}$ then $~\eltype{\trform\icontext}{\trform\exprs : \trform\indice}$
    \vspace{0.1cm}
\item If \ $\eltype{\context}{\tprop}$ then $~\eltype{\trform\icontext}{\trform\tprop}$.
    \end{itemize}
\end{lemma}

\begin{lemma}[Derivability Preservation \coqdoc{EffHOL_to_HOL.html\#HOPL_HOL}]\label{soundesstohol}
    If\ $\elsequent{\context}{\tprops }{\tprop},\text{ then } \trform\icontext \vdash \trform\tprops \Rrightarrow \trform\tprop $.
\end{lemma}

\begin{figure*}[h]
\hspace{-0.5em}
\input{section/figures/HOPL_theory}
\caption{The Theory of $\hopl$ \coqdoc{EffHOL.html\#HOPL_prv}
}
%\revnote{Mem-I: $\tau$ and $\sigma$ are unconstrained? Memo-I: $\tau$ is unconstrained here?}}
%\lcnote{We  need to be uniformly explicit with the context.}}
%For example, mem needs to add x to the type context,right? }}
%\lcnote{in kind uni intro, what's going on in the context? can it be a premise?}}
%\lcnote{we can probably omit the names of the rules and to name them in the text instead }\agnote{We're definitely going to use them in the soundness proof}}
\label{fig:theory-full}
\end{figure*}

~\newpage
\section{Realizability Translation}
%\subsubsection{Proof of~\Cref{lem:sub}:}

%First, we show that our translation preserves substitutions. 

\begin{lemma}[Substitution Preservation \coqdoc{HOL_to_EffHOL.html\#trans_I_subst}]\label{lem:sub}
The following syntactic equalities hold.
\begin{itemize}[leftmargin=*]
    \item $\ttrtype{\scontext}{\sprop\subst{\spred :=\sterm}} = 
    \ttrtype{\scontext}{\sprop}\subst{\typevar_{\spred}:=\ttretype{\scontext}{\sterm}{\sort}}$ %\lcnote{equiv: $\typeabs{\typevar_\spred}{\trkind{\sort}}{\trtype{\scontext , \spred : \sort}{\sprop}}(\tretype{\scontext}{\sterm}{\sort})$?}
    \vspace{0.1cm}
        \item $\ttrspec{\scontext}{\sprop\subst{\spred :=\sterm}}{x} = 
     \ttrspec{\scontext}{\sprop}{x}\subst{\typevar_{\spred}:=\ttretype{\scontext}{\sterm}{\sort}}\subst{\tpred_{\spred}:=\ttrtrm{\scontext}{\sterm}{\sort}}$
    %\lcnote{doesn't compile with the [xr]}
    \vspace{0.1cm}
        \item $\ttrtrm{\scontext}{\sterm '\subst{\spred :=\sterm}} == 
    \ttrtrm{\scontext}{\sterm'}{\sort}\subst{\tpred_{\spred}:=\ttrtrm{\scontext}{\sterm}{\sort}}$
    \vspace{0.1cm}
        \item $\ttretype{\scontext}{\sterm '\subst{\spred :=\sterm}}{\sort} = 
    \ttretype{\scontext}{\sterm'}{\sort}
    \subst{\typevar_{\spred}:=\ttretype{\scontext}{\sterm}{\sort}}
    \subst{\tpred_{\spred}:=\ttrtrm{\scontext}{\sterm}{\sort}}$
    %  \item $\tretype{\scontext}{\sterm '\subst{\spred :=\sterm}}{\sort} == 
    % \trtype{\scontext}{\sterm}\subst{\typevar_{\spred }:=\tretype{\scontext}{\sterm'}{\sort}}$
    % \lcnote{equiv: $\typeabs{\typevar_\spred}{\trkind{\sort}}{\trtype{\scontext , \spred : \sort}{\sterm'}}(\tretype{\scontext}{\sterm}{\sort})$?}
\end{itemize}
    
\end{lemma}

\begin{lemma}[Judgement Preservation \coqdoc{HOL_to_EffHOL.html\#trans_is_index}]\label{lem:preswf}
 %   \lcnote{needed for soundness of uni elim (and probably more)}
 The following hold:
    \begin{itemize}[leftmargin=*]
        \item 
    If \ $\eltype{\scontext}{\sterm:\sort}$, then  
    %\begin{itemize}
    %\vspace{0.1cm}
    %    \item 
    $\eltype{\ttrkind{\scontext}}{\ttretype{\scontex}{\sterm}{\sort}:\trkind{\sort}} \, \text{and} \,
    %\vspace{0.1cm}
    %\item 
    \eltrm{\trkind{\scontext} \mid
    \trind{}{\scontext}\mid \emptyset }{\ttrtrm{\scontext}{\sterm}{\sort}}{\trind{\ttretype{}{\sterm}{}}{\sort}}$
     %   \end{itemize}
        \item If \ $\eltype{\scontext}{\sprop}$, then~
    %\begin{itemize}
    %\vspace{0.1cm}
    %\item
    $\eltrm{\trkind{\scontext}}{\trtype{}{\sprop}}{\star} \, \text{~and~} \,
    %\vspace{0.1cm}
    %\item
    \eltype{\trkind{\scontext}\mid\trind{}{\scontext}
    \mid \termvar : \trtype{\scontext}{\sprop}}{\trspec{\termvar}{\sprop}}$
    %\end{itemize}
     \end{itemize}
\end{lemma}
%\subsection{Proof of~\Cref{lem:preswf}:}
% \begin{theorem}
% $$
% \infer
% {\elprop{\trkind{\scontext} \mid \termvar : \trtype{\scontext}{\sprop} \mid \trind{\trkind{\scontext}}{\scontext} }{\trspec{\scontext}{\sprop}\subst{\termvar}} }
% {\elprop{\scontext}{\sprop}}
% $$
\begin{proof}
The proof is by induction using the translation definition given in~\Cref{fig:trans2}. Here we focus on the case where $\elprop{\scontext}{\forall_{\spred : \sort} \sprop}$
, and show that 
$$
    {\elprop{\trkind{\scontext} \mid
    \termvar : \prod_{\typevar : \trkind{\sort}} \typecomp{\trtype{\scontext , \spred : \sort}{\sprop}\subst{\typevar_{\spred} := \typevar}} \mid \trind{\trkind{\scontext}}{\scontext}}{\trspec{\scontext}{\forall_{\spred : \sort} \sprop}\subst{\termvar}}}
$$
% \end{theorem}

The WF rule for universal quantification is:
$$
\infer{\elprop{\scontext}{\tspecindprod{\spred : \sort}{\sprop}}}{\elprop{\scontext, \spred : \sort }{\sprop}}
$$

So by the IH we get that if $\elprop{\scontext, \spred : \sort }{\sprop}$, then 
$
\elprop{\trkind{\scontext} , \typevar_{u} : \trkind{\sort} \mid
    \tcontext, \termvar_{h} : \trtype{\scontext , \spred : \sort}{\sprop} \mid \trind{\trkind{\scontext}}{\scontext} , \tpred_{\spred} : \trind{\trkind{\sort}}{\sort} } {\trspec{\scontext , \spred : \sort}{\sprop}\subst{\termvar_{h}} }
$.
Now we get:

\begin{center}
$$
% \begin{small}
\infer
{\elprop{\trkind{\scontext} \mid
    \termvar : \tau \mid \trind{\trkind{\scontext}}{\scontext}}{\tspeckinprod{\typevar : \trkind{\sort}} \forall_{\tpred : \trind{\typevar}{\sort} } . \after{\termtypeapp{\termvar}{\typevar}  }{\termvar_{0}}{(\trspec{\scontext , \spred : \sort}{\sprop}\subst{\typevar_{\spred} := \typevar}\subst{\termvar_{0}})\subst{ \tpred_{\spred} := \tpred }}}}
{\infer
{\elprop{\trkind{\scontext} , \typevar : \trkind{\sort} \mid
    \termvar : \tau \mid \trind{\trkind{\scontext}}{\scontext}}{ \forall_{\tpred : \trind{\typevar}{\sort} } . \after{\termtypeapp{\termvar}{\typevar}  }{\termvar_{0}}{(\trspec{\scontext , \spred : \sort}{\sprop}\subst{\typevar_{\spred} := \typevar}\subst{\termvar_{0}})\subst{ \tpred_{\spred} := \tpred }}}}
{
\infer
{\elprop{\trkind{\scontext} , \typevar : \trkind{\sort} \mid
    \termvar : \tau \mid \trind{\trkind{\scontext}}{\scontext} , \tpred : \trind{\typevar}{\sort} }{ \after{\termtypeapp{\termvar}{\typevar}  }{\termvar_{0}}{(\trspec{\scontext , \spred : \sort}{\sprop}\subst{\typevar_{\spred} := \typevar}\subst{\termvar_{0}})\subst{ \tpred_{\spred} := \tpred }}}}
{
\infer
{\elprop{ \Theta_{1} }{ \termtypeapp{\termvar}{\typevar} : \typecomp{\trtype{\scontext , \spred : \sort}{\sprop}\subst{\typevar_{\spred} := \typevar}} }}
{
\mathbb{A}
}
&\qquad %--------------------------------
\infer
{\elprop{\Theta_{2}}{(\trspec{\scontext , \spred : \sort}{\sprop}\subst{\typevar_{\spred} := \typevar}\subst{\termvar_{0}})\subst{ \tpred_{\spred} := \tpred }}}
{\mathbb{B}}
}}}
% \end{small}
$$

where $\mathbb{A}$ is:
$$
% \begin{small}
\infer
{\elprop{ \Theta_{1} }{ \termtypeapp{\termvar}{\typevar} : \typecomp{\trtype{\scontext , \spred : \sort}{\sprop}} }}
{\infer
{\elprop{ \trkind{\scontext} , \typevar : \trkind{\sort} \mid
    \termvar : \type }{ \termvar : \prod_{\typevar_{1}} \typecomp{\trtype{\scontext , \spred : \sort}{\sprop}\subst{\typevar_{\spred} := \typevar}} \subst{\typevar := \typevar_{1}} }}
{
\infer[\text{id} + \alpha]
{\elprop{ \trkind{\scontext} , \typevar : \trkind{\sort} \mid
    \termvar : \type }{ \termvar : \prod_{\typevar_{1}} \typecomp{\trtype{\scontext , \spred : \sort}{\sprop}\subst{\typevar_{\spred} := \typevar_{1}}} }}
{}
}
&
\infer[\text{id}]
{\eltype{\trkind{\scontext} , \typevar : \trkind{\sort}}{\typevar : \trkind{\sort}}}
{}
}  
% \end{small}
$$

where $\mathbb{B}$ is:
$$
% \begin{small}
\infer[\text{wk}\left(\termvar\right)]
{\elprop{\trkind{\scontext} , \typevar : \trkind{\sort} \mid
    \termvar : \type , \termvar_{0} : \trtype{\scontext , \spred : \sort}{\sprop}\subst{\typevar_{\spred} := \typevar} \mid \trind{\trkind{\scontext}}{\scontext} , \tpred : \trind{\typevar}{\sort}}{(\trspec{\scontext , \spred : \sort}{\sprop}\subst{\typevar_{\spred} := \typevar}\subst{\termvar_{0}})\subst{ \tpred_{\spred} := \tpred }}}
{
\infer[\text{rename}]
{\elprop{\trkind{\scontext} , \typevar : \trkind{\sort} \mid \termvar_{0} : \trtype{\scontext , \spred : \sort}{\sprop}\subst{\typevar_{\spred} := \typevar} \mid \trind{\trkind{\scontext}}{\scontext} , \tpred : \trind{\typevar}{\sort}}{(\trspec{\scontext , \spred : \sort}{\sprop}\subst{\typevar_{\spred} := \typevar}\subst{\termvar_{0}})\subst{ \tpred_{\spred} := \tpred }}}
{
\infer[\text{hypothesis}]
{
{\elprop{\trkind{\scontext} , \typevar_{\spred} : \trkind{\sort} \mid \termvar_{h} : \trtype{\scontext , \spred : \sort}{\sprop} \mid \trind{\trkind{\scontext}}{\scontext} , \tpred_{\spred} : \trind{\typevar}{\sort}}{\trspec{\scontext , \spred : \sort}{\sprop}\subst{\termvar_{h}} }}
}
{}
}
}
% \end{small}
$$
\end{center}

\begin{align*}
    \type &:= \prod_{\typevar : \trkind{\sort}} \typecomp{\trtype{\scontext , \spred : \sort}{\sprop}\subst{\typevar_{\spred} := \typevar}}\\
    \Theta_{1} &:= \trkind{\scontext} , \typevar : \trkind{\sort} \mid
    \termvar : \type\\
    \Theta_{2} &:= \trkind{\scontext} , \typevar : \trkind{\sort} \mid
    \termvar : \type , \termvar_{0} : \trtype{\scontext , \spred : \sort}{\sprop}\subst{\typevar_{\spred} := \typevar} \mid \trind{\trkind{\scontext}}{\scontext} , \tpred : \trind{\typevar}{\sort}
\end{align*}
\end{proof}

\subsection{Proof of Soundness (\Cref{thm:soundness} \coqdoc{HOL_to_EffHOL.html\#soundness})}
\label{app:soundness}
    \begin{proof}[Proof (outline)]

\lcnote{comment about the additional assumptions}
The proof is carried out by induction on the derivation rules of $\ihol$. 
% The full proof can be found in the formalization in Coq~\cite{todo}. Here, we present a couple of cases that demonstrate some key aspects of the translation.
For readability, we abuse the notation and write $\holhyp$ for the $\ihol$ context $\sprop_1,...,\sprop_n$, $\trtypecontext$ for the typing context obtained by translating these propositions in $\sprops$ as types and as in the theorem $\trspeccontext$ for their translations as specifications for the corresponding variables. 
We also focus on the triples,  denoting by $\trcontext$ the context $\ttrkind{\scontext}\mid \ttrind{}{\scontext} \mid  \trtypecontext $ when irrelevant, abusing notations 
to write $\trcontext,x:\tau$ for $\ttrkind{\scontext}\mid \ttrind{}{\scontext} \mid  \trtypecontext,x:\tau $.

\begin{description}
    \item[Case (Imp Elim)] 
Assume that the last rule of derivation in HOL was
\[\infer{\selsequent{\scontext}{\holhyp}{\sprop_{2}}}{\selsequent{\scontext}{\holhyp}{\simplies{\sprop_{1}}{\sprop_{2}}} \quad \selsequent{\scontext}{\holhyp}{\sprop_{1}}}\]
By induction, we get that:
\begin{itemize}
    \item there is a term $\term_0$  such that (IH0):
\[ \eltrm{\ttrkind{\scontext}\mid \trtypecontext\!}{\!\term_0}{ M(\ttrtype{\scontext}{\sprop_1}\!\to\! M\ttrtype{\scontext}{\sprop_2})} 
\qquad \htriple{\trcontext }{ {\sprops} }
{\termvar_0} {\term_0 }
{\ttrspec{\scontext}{\simplies{\sprop_1}{\sprop_2}}{\termvar_0}} \]
    \item  there is a term $\term_1$  such that (IH1):
\[ \eltrm{\ttrkind{\scontext}\mid \trtypecontext}{\term_1}{ M\ttrtype{\scontext}{\sprop_1} }
\qquad\qquad
\htriple{\trcontext }{ {\sprops} }
{\termvar_{1}} {\term_1 }
{\ttrspec{\scontext}{\sprop_1}{\termvar_{1}}} \]
\end{itemize}  
Recall that by definition, $\ttrspec{\scontext}{\simplies{\sprop_1}{\sprop_2}}{\termvar_0} = 
\tspectypprod{\termvar_{1} : \trtype{\scontext}{\sprop_{1}} } . \ttrspec{\scontext}{\sprop_{1}}{\termvar_{1}} \supset \after{\termapp{\termvar_0}{\termvar_{1}}}{\termvar_{2}}{\ttrspec{\scontext}{\sprop_{2}}{\termvar_{2}}}$.
To define a realizer for the conclusion, take $\term_2\eqdef \termbind{\termvar_0}{\term_0}{\termbind{\termvar_1}{\term_1}{\termapp{\termvar_0}{\termvar_1}}}$, for which it is easy to check that 
\[\eltrm{\ttrkind{\scontext}\mid \trtypecontext}{\term_2}{ M(\ttrtype{\scontext}{\sprop_2})}\]
We conclude this case by proving that $\term_2$ is indeed a valid realizer.
\newcommand{\Hzero}{\ttrspec{\scontext}{\simplies{\sprop_1}{\sprop_2}}{\termvar_0}}
\newcommand{\Hone}{\ttrspec{\scontext}{\sprop_1}{\termvar_1}}
\[
\infer[\modelimrule]{
\elsequent{\trcontext }{ \sprops}{\after{\term_2}{\termvar_2}{\ttrspec{\scontext}{\sprop{_2}}{\termvar_{2}}}}
}{
    \infer[\modelimrule]{
    \elsequent{\trcontext }{ \sprops}{\after{\term_0}{\termvar_0}{\after{\termbind{\termvar_1}{\term_1}{\termapp{\termvar_0}{\termvar_1}}}{\termvar_2}{\ttrspec{\scontext}{\sprop_2}{\termvar_{2}}}}}
    }{
        \infer[\monshortrule]{
        \elsequent{\trcontext }{ \sprops}{\after{\term_0}{\termvar_0}{\after{\term_1}{\termvar_1}{\after{{\termapp{\termvar_0}{\termvar_1}}}{\termvar_2}{\ttrspec{\scontext}{\sprop_2}{\termvar_{2}}}}}}
        }{
           \infer[\monshortrule]{
                \elsequent{\trcontext }{ \sprops,\Hzero}{\after{\term_1}{\termvar_1}{\after{{\termapp{\termvar_0}{\termvar_1}}}{\termvar_2}{\ttrspec{\scontext}{\sprop_2}{\termvar_{2}}}}}
            }{
              \Pi\qquad
            &
            \infer={\elsequent{\trcontext }{ \sprops,\Hzero}{\after{\term_1}{\termvar_1}{\Hone}}
            }{            (IH_1)}
            }
        & 
            \infer={\elsequent{\trcontext }{ \sprops}{\after{\term_0}{\termvar_0}{\Hzero}}}{
            (IH_0)
            }
        }
    }
}
\]
where $\Pi$ is the following subderivation
\[
  \infer[\impelimrule]{
                    \elsequent{\trcontext }{ \sprops,\Hzero,\Hone}{{\after{{\termapp{\termvar_0}{\termvar_1}}}{\termvar_2}{\ttrspec{\scontext}{\sprop_2}{\termvar_{2}}}}}
                }{
                    \infer[\progunielimrule]{
                        \elsequent{\trcontext }{ \sprops,\Hzero}{\ttrspec{\scontext}{\sprop_{1}}{\termvar_{1}} \supset{\after{{\termapp{\termvar_0}{\termvar_1}}}{\termvar_2}{\ttrspec{\scontext}{\sprop_2}{\termvar_{2}}}}}
                    }{
                        \infer[\idrule]{
                            \elsequent{\trcontext }{ \sprops,\Hzero}{\Hzero}
                        }{}
                    }
                    &
                    \infer[\idrule]{\elsequent{\trcontext }{ \sprops,\Hone}{\Hone}}{}
                }
\]

  \item[Case (Imp Intro)] 
Assume that the last rule of derivation in HOL was
\[\infer{\selsequent{\scontext}{\holhyp}{\simplies{\sprop_{1}}{\sprop_{2}}}}
        {\selsequent{\scontext}{\holhyp,\sprop_1}{\sprop_{2}}}\]
By induction, we obtain a term $\term_0$  such that (IH0):
\[ \eltrm{\ttrkind{\scontext}\mid \trtypecontext,\termvar_1:\trtype{}{\sprop_1}\!}{\!\term_0}{  M\ttrtype{\scontext}{\sprop_2}} 
\qquad 
\elsequent{\trcontext,\termvar_1:\trtype{}{\sprop_1} }{ {\sprops,\ttrspec{}{\sprop_1}{\termvar_0}} }
{\after{\term_0} {\termvar_2 }
{\ttrspec{\scontext}{{\sprop_2}}{\termvar_2}}} \]    
    
 Recall that by definition, $\ttrspec{\scontext}{\simplies{\sprop_1}{\sprop_2}}{\term} = 
 \tspectypprod{\termvar_{1} : \trtype{\scontext}{\sprop_{1}} }{ \ttrspec{\scontext}{\sprop_{1}}{\termvar_{1}} \supset \after{\termapp{\term}{\termvar_{1}}}{\termvar_{2}}{\ttrspec{\scontext}{\sprop_{2}}{\termvar_{2}}}}$.
We prove that $p\eqdef \termret{\termabs{x_1}{\ttrtype{\scontext}{\sprop_1}}{\term_0}}$ is the realizer we are looking for :
\[
\infer[\modelimrule]{
\elsequent{\trcontext}{ \sprops}{\after{\termret{\termabs{x_1}{\ttrtype{\scontext}{\sprop_1}}{\term_0}}}{\termvar}{\ttrspec{\scontext}{\simplies{\sprop_1}{\sprop_2}}{\termvar}}}
}{
    \infer[\proguniintrorule]{
    \elsequent{\trcontext }{ \sprops}    {\tspectypprod{\termvar_{1} : \trtype{\scontext}{\sprop_{1}} }
        { \simplies{\ttrspec{\scontext}{\sprop_{1}}{\termvar_{1}}}
            {\after{\termapp{(\termabs{x_1}{\ttrtype{\scontext}{\sprop_1}}{\term_0})}                   {\termvar_{1}}}
                    {\termvar_{2}}
                    {\ttrspec{\scontext}{\sprop_{2}}{\termvar_{2}}}}}}
    }{\infer[\impintrorule]{
        \elsequent{\trcontext,\termvar_{1} : \trtype{\scontext}{\sprop_{1}} }{ \sprops}
            { \simplies{\ttrspec{\scontext}{\sprop_{1}}{\termvar_{1}}}
            {\after{\termapp{(\termabs{x_1}{\ttrtype{\scontext}{\sprop_1}}{\term_0})}                   {\termvar_{1}}}
                    {\termvar_{2}}
                    {\ttrspec{\scontext}{\sprop_{2}}{\termvar_{2}}}}}
    }{
        \infer[\antiredtermrule]{
            \elsequent{\trcontext,\termvar_{1} : \trtype{\scontext}{\sprop_{1}} }
                    { \sprops,\ttrspec{\scontext}{\sprop_{1}}{\termvar_{1}}}
                    {\after{\termapp{(\termabs{x_1}{\ttrtype{\scontext}{\sprop_1}}{\term_0})}                   {\termvar_{1}}}
                    {\termvar_{2}}
                    {\ttrspec{\scontext}{\sprop_{2}}{\termvar_{2}}}}
        }{
        \infer={
                \elsequent{\trcontext,\termvar_{1} : \trtype{\scontext}{\sprop_{1}} }
                    { \sprops,\ttrspec{\scontext}{\sprop_{1}}{\termvar_{1}}}
                    {\after{\term_0}
                    {\termvar_{2}}
                    {\ttrspec{\scontext}{\sprop_{2}}{\termvar_{2}}}}
                }{(IH_0)}
        }
    }
    }
}
\]
 
% To define a realizer for the conclusion, take $\term_2\eqdef \termbind{\termvar_0}{\term_0}{\termbind{\termvar_1}{\term_1}{\termapp{\termvar_0}{\termvar_1}}}$, for which it is easy to check that 
% \[\eltrm{\ttrkind{\scontext}\mid \trtypecontext}{\term_2}{ M(\ttrtype{\scontext}{\sprop_2})}\]
% We conclude this case by proving that $\term_2$ is indeed a valid realizer.
% \newcommand{\Hzero}{\ttrspec{\scontext}{\simplies{\sprop_1}{\sprop_2}}{\termvar_0}}
% \newcommand{\Hone}{\ttrspec{\scontext}{\sprop_1}{\termvar_1}}

\item[Case (Uni Elim)] 
Assume that the last rule of derivation in $\ihol$ was
        $$\infer
        {\selsequent{ \scontext}{\sprops}\sprop\subst{\spred:=\sterm}}
        {\selsequent{\scontext}{\sprops}{\tspecindprod{\spred:\sort} \sprop}
        & \eltype{\scontext}{\sterm:\sort}}$$
By induction and the preservation of wf judgment (\Cref{lem:preswf}), we get that:
\begin{itemize}
    \item there is a term $\term_0$  such that:
\[ \eltrm{\ttrkind{\scontext}\mid  \trtypecontext\!}{\term_0}{ M\ttrtype{\scontext}{\tspecindprod{\spred:\sort} \sprop} }
\qquad \qquad 
\htriple{\trcontext }{ \sprops }
{\termvar_0} {\term_0 }
{\ttrspec{\scontext}{\tspecindprod{\spred:\sort} \sprop} {\termvar_0}}.\]
    \item  the following typing sequents hold % there is a term $\term_1$  such that (IH1):
\[  \eltype{\trkind{\scontext}}{\ttretype{\scontext}{\sterm}{\sort}:\trkind{\sort}} 
    \qquad 
    \eltrm{\trkind{\scontext} \mid\ttrind{}{\scontext}\mid\emptyset}{\ttrtrm{\scontext}{\sterm}{\sort}}{\ttrind{\ttretype{\scontext}{\sterm}{\sort}}{\sort}}
    \]
\end{itemize} 

Recall that by definition, we have:
\begin{itemize}
\item $\ttrtype{\scontext}{\tspecindprod{\spred:\sort} \sprop}  = {\ttypprod{\typevar_\spred :\ttrkind{\sort}} {\typecomp{\ttrtype{\scontext , \spred : \sort}{\sprop}}}}$
\item $\ttrspec{\scontext}{\tspecindprod{\spred:\sort} \sprop} {x_0} = \tspeckinprod{\typevar_\spred : \trkind{\sort}} {
    \tspecindprod{\tpred_\spred}{\trind{\typevar_u}{\sort} }
        {\after{\termtypeapp{\termvar_0}{\typevar_\spred}  }
            {\termvar}
            {\ttrspec{\scontext , \spred : \sort}{\sprop}{\termvar}}}}$
\end{itemize}

%
% Also,  by the induction hypothesis,  $ \eltrm{\ttrkind{\scontext}\mid  \trtypecontext\!}{\term_0}{ M\ttrtype{\scontext}{\tspecindprod{\spred:\sort} \sprop} } = \typecomp {\ttypprod{\typevar_\spred :\ttrkind{\sort}} {\typecomp{\ttrtype{\scontext , \spred : \sort}{\sprop}}}}$.
To define a realizer for the conclusion, take $\term\eqdef \termbind{\termvar_0}{\term_0}{\termtypeapp{\termvar_0}{\ttretype{\scontext}{\sterm}{\sort}}}$, for which 
it is easy to check that %\lcnote{not really. I ended up with having to identify $\typeabs{\typevar_\spred : \ttrkind{\sort}} {\typecomp{\ttrtype{\scontext , \spred : \sort}{\sprop}}}$ with $\prod_{\typevar_u : \ttrkind{\sort}} \typecomp{\ttrtype{\scontext , \spred : \sort}{\sprop}}$. smt is off. I might be too tired. will check tomorrow}
\[\eltrm{\ttrkind{\scontext}\mid \trtypecontext}{\term_1}{ M(\ttrtype{\scontext}{\sprop\subst{\spred:=\sterm}})}\]
since by~\Cref{lem:sub}, 
$\ttrtype{\scontext}{\sprop\subst{\spred:=\sterm}}:=\ttrtype{\scontext}{\sprop}\subst{\typevar_{\spred}:=\ttretype{\scontext}{\sterm}{\sort}}$.

\renewcommand{\Hzero}{\ttrspec{\scontext}{\tspecindprod{\spred:\sort} \sprop} {\termvar_0} }

Recall that, by~\Cref{lem:sub},  $\ttrspec{\scontext}{\sprop\subst{\spred :=\sterm}}{x} = \ttrspec{\scontext}{\sprop}{x}\subst{\typevar_{\spred}:=\ttretype{\scontext}{\sterm}{\sort}}\subst{\tpred_{\spred}:=\ttrtrm{\scontext}{\sterm}{\sort}}$. 
We conclude this case by proving that $\term$ is indeed a valid realizer.%\lcnote{fix}\emnote{done =)}
\[
% \begin{smaller}
\infer[\modelimrule]{
\elsequent{\trcontext }{ \sprops}
{\after
    {\termbind{x_0}{\term_0}{\termtypeapp{x_0}{\ttretype{\scontext}{\sterm}{\sort}}}}
    {\termvar}
    {\ttrspec{\scontext}{\sprop}{\termvar}
        \subst{\typevar_{\spred}:=\ttretype{\scontext}{\sterm}{\sort}}
        \subst{\tpred_{\spred}:=\ttrtrm{\scontext}{\sterm}{\sort}}}}
}{
    \infer[\monshortrule]{
    \elsequent{\trcontext }{\sprops}{
        \after{\term}{\termvar_0}
        {\after{\termtypeapp{\termvar_0}{\ttretype{\scontext}{\sterm}{\sort}} }
            {\termvar}
            {\ttrspec{\scontext}{\sprop}{\termvar}
                \subst{\typevar_{\spred}:=\ttretype{\scontext}{\sterm}{\sort}}
                \subst{\tpred_{\spred}:=\ttrtrm{\scontext}{\sterm}{\sort}}}
        }
    }
    }{
       \infer[\expshortunielimrule]{
          \elsequent{\trcontext }{\sprops,\Hzero}{
            {\after{\termtypeapp{\termvar_0}{\ttretype{\scontext}{\sterm}{\sort}} }
            {\termvar}
            {\ttrspec{\scontext}{\sprop}{\termvar}
                \subst{\typevar_{\spred}:=\ttretype{\scontext}{\sterm}{\sort}}
                \subst{\tpred_{\spred}:=\ttrtrm{\scontext}{\sterm}{\sort}}}
            }
            }
        }{
            \infer[\typeunielimrule]{
              \elsequent{\trcontext }{\sprops,\Hzero}{
                {\tspecindprod{\tpred_\spred}{\trind{\typevar_u}{\sort} }{\after{\termtypeapp{\termvar_0}{\ttretype{\scontext}{\sterm}{\sort}} }
                    {\termvar}
                    {\ttrspec{\scontext}{\sprop}{\termvar}
                        \subst{\typevar_{\spred}:=\ttretype{\scontext}{\sterm}{\sort}}}
                }
                }
                }
            }{
                \infer[\idrule]{
                    \elsequent{\trcontext }{\sprops,\Hzero}{
                        \tspeckinprod{\typevar_\spred : \trkind{\sort}} {
                            \tspecindprod{\tpred_\spred}{\trind{\typevar_u}{\sort} }
                                {\after{\termtypeapp{\termvar_0}{\typevar_\spred}  }
                                    {\termvar}
                                    {\ttrspec{\scontext , \spred : \sort}{\sprop}{\termvar}}}}
                    }
                }{}
        }
        }
        & 
        \infer={\elsequent{\trcontext }{ \sprops}{\after{\term_0}{\termvar_0}{\ttrspec{\scontext}{\tspecindprod{\spred:\sort} \sprop} {\termvar_0} }}}{
            (IH_0)
            }
        }
    }
    % \end{smaller}
\]

% ===========

% Applying Kind Elim and Exp Elim using the two wf judgments from the hypothesis obtains a proof of:
% $$
% \ttrspec{\scontext}{\tspecindprod{\spred:\sort} \sprop} {x} \Rightarrow 
%  \after{\termtypeapp{\termvar}{\tretype{\scontext}{\sterm}{\sort}}  }{\termvar_{0}}{\ttrspec{\scontext , \spred : \sort}{\sprop}{\termvar_{0}}\subst{\tpred_\spred := \trtrm{\scontext}{\sterm}{\sort}}}
% $$
% Applying Monotonicity on the above sequent and the hypothesis triple obtains a proof of 
% $$
% \sprops \Rightarrow \after{\term_0}{\termvar}
% { \after{\termtypeapp{\termvar}{\tretype{\scontext}{\sterm}{\sort}}  }{\termvar_{0}}{\ttrspec{\scontext , \spred : \sort}{\sprop}{\termvar_{0}}\subst{\tpred_\spred := \trtrm{\scontext}{\sterm}{\sort}}}}.
% $$
% From this, using Mod Elim, we obtain a proof of 
% $$
% \sprops \Rightarrow 
% { \after{\termbind{x}{\term_0}{\termtypeapp{x}{\tretype{\scontext}{\sterm}{\sort}}} }
% {\termvar_{0}}{\ttrspec{\scontext , \spred : \sort}{\sprop}{\termvar_{0}}\subst{\tpred_\spred := \trtrm{\scontext}{\sterm}{\sort}}}}.
% $$
% % Therefore, using the substitution Lemma (\Cref{lem:sub}), we have that the triple 
% % $\htriple{\ttrkind{\scontext} \mid  \ttrtype{}{\sprops}  \mid \ttrind{}{\scontext} }{ \sprops }
% % {\termvar} {\term }
% % {\ttrspec{\scontext}{\sprop\subst{\spred:=\sterm}} {x}}$
% % is provable in $\effhol$ for $\term := \termbind{x}{\term_0}{\termtypeapp{x}{\tretype{\scontext}{\sterm}{\sort}}}$.

\item[Case (Uni Intro)]
Assume that the last rule of derivation in HOL was
        $$\infer[(\text{Uni Intro})]
        {\selsequent{\scontext}{\sprops}{\tspecindprod{\spred:\sort} \sprop}}
        {\selsequent{ \scontext, \spred : \sort}{\sprops}{\sprop}}$$
By induction, there is a term $\term_0$  such that:
$$ \eltrm{\ttrkind{\scontext}, \typevar_\spred:\ttrkind{\sort} \mid \ttrtype{}{\sprops}}{\term_0}{ M\ttrtype{\scontext}{ \sprop} }$$
and 
$$
\htriple{\ttrkind{\scontext}, \typevar_\spred:\ttrkind{\sort} \mid \ttrind{}{\scontext}, \tpred_\spred:\ttrind{\typevar_\spred }{\sort}\mid  \trtypecontext }{ \sprops }
{\termvar} {\term_0 }
{\ttrspec{\scontext}{ \sprop} {x}}.$$

Recall that by definition, $\ttrspec{\scontext}{\tspecindprod{\spred:\sort} \sprop} {x} := \tspeckinprod{\typevar_\spred : \trkind{\sort}} \forall_{\tpred_\spred : \trind{\typevar}{\sort} } . \after{\termtypeapp{\termvar}{\typevar_\spred}  }{\termvar_{0}}{\ttrspec{\scontext , \spred : \sort}{\sprop}{\termvar_{0}}}$

We show that 
\[ \htriple{\Gamma}{ \sprops }
{\termvar} {\term }
{\ttrspec{\scontext}{\tspecindprod{\spred:\sort} \sprop} {x}} \]
for $p\eqdef\termret{ \termtypeabs{X_u}{\ttrkind{\sort}}{p_0}}$.

\[
\infer[\modintrorule]{
\elsequent{\trcontext}{\trspeccontext}{
\after{\termret{ \termtypeabs{X_u}{\ttrkind{\sort}}{p_0}}}{\termvar}{(\tspeckinprod{\typevar_\spred : \trkind{\sort}} \forall_{\tpred_\spred : \trind{\typevar}{\sort} } . \after{\termtypeapp{\termvar}{\typevar_\spred}  }{\termvar_{0}}{\ttrspec{\scontext , \spred : \sort}{\sprop}{\termvar_{0}}})}
}
}{
    \infer[\typeuniintrorule]{
        \elsequent{\trcontext}{\trspeccontext}{
        {\tspeckinprod{\typevar_\spred : \trkind{\sort}} \forall_{\tpred_\spred : \trind{\typevar}{\sort} } . \after{\termtypeapp{(\termtypeabs{X_u}{\ttrkind{\sort}}{p_0})}{\typevar_\spred}  }{\termvar_{0}}{\ttrspec{\scontext , \spred : \sort}{\sprop}{\termvar_{0}}}}
        }
    }{
        \infer[\expuniintrorule]{
            \elsequent{\ttrkind{\scontext}, \typevar_\spred:\ttrkind{\sort} \mid \ttrind{}{\scontext}\mid  \trtypecontext}{\trspeccontext}{
            { \forall_{\tpred_\spred : \trind{\typevar}{\sort} } . \after{\termtypeapp{(\termtypeabs{X_u}{\ttrkind{\sort}}{p_0})}{\typevar_\spred}  }{\termvar_{0}}{\ttrspec{\scontext , \spred : \sort}{\sprop}{\termvar_{0}}}}
            }
        }{
            \infer[\antiredtermrule]{
                    \elsequent{\ttrkind{\scontext}, \typevar_\spred:\ttrkind{\sort} \mid \ttrind{}{\scontext}, \tpred_\spred:\ttrind{\typevar_\spred }{\sort}\mid  \trtypecontext}{\trspeccontext}{
                    { \after{\termtypeapp{(\termtypeabs{X_u}{\ttrkind{\sort}}{p_0})}{\typevar_\spred}  }{\termvar_{0}}{\ttrspec{\scontext , \spred : \sort}{\sprop}{\termvar_{0}}}}
                    }
                }{
                   \infer={
                        \elsequent{\ttrkind{\scontext}, \typevar_\spred:\ttrkind{\sort} \mid \ttrind{}{\scontext}, \tpred_\spred:\ttrind{\typevar_\spred }{\sort}\mid  \trtypecontext}{\trspeccontext}{
                        { \after{p_0}{\termvar_{0}}{\ttrspec{\scontext , \spred : \sort}{\sprop}{\termvar_{0}}}}
                        }
                    }{
                        (IH_0)
                    }
                }
        }
    }     
}
\]

\item[Case (Mem Intro)] 
Assume that the last rule of derivation in HOL was
        $$\infer[(\text{Mem Intro})]
        {\selsequent{\scontext}{\holhyp}{ \smem{\sterm}{\comprehend{\spred}{\sort}{\sprop}}}}
        {\selsequent{\scontext}{\holhyp}{\sprop\subst{\spred := \sterm}}}$$
By induction, there is a term $\term_0$  such that:
\[
\eltrm{\trtypecontext}{\term_0}{ M\ttrtype{\scontext}{ \sprop\subst{\spred := \sterm}} }
\qquad\text{and}\qquad
\htriple{\trcontext}{ \trspeccontext }{\termvar_0} {\term_0}{\ttrspec{\scontext}{ \sprop\subst{\spred := \sterm}} {x}}.\eqno(IH)\]

Observe first that:
\begin{align*}
\ttrtype{\scontext}{\smem{t}{\comprehend{\spred}{\sort}{\sprop}}}
&~~=~~ \typeapp{\left(\typeabs{\typevar_\spred}{\ttrkind{\sort}}{\ttrtype{\scontext , \spred : \sort}{\sprop}}\right)}{\ttrtype{}{t}}\\
&~~=_\beta~~ \ttrtype{\scontext , \spred : \sort}{\sprop}\subst{\typevar_\spred:=\ttrtype{}{t}}
\qquad~~\overset{\text{Lm. }\ref{lem:sub}}{=}~~\ttrtype{\scontext}{ \sprop\subst{\spred := \sterm}}
\end{align*}
Besides, we have for any $\termvar$ that:
\begin{align*}
\ttrspec{}{\smem{\sterm}{\comprehend{\spred}{\sort}{\sprop}}}{\termvar_0} ~~
&=~~ \tmem{ \termvar_0 }{ \eapp{\ttrtrm{\scontext}{\comprehend{\spred}{\sort}{\sprop}}{\predcon{\sort}}}{\ttretype{\scontext}{\sterm_{1}}{\sort}} }{ \ttrtrm{\scontext}{\sterm_{1}}{\sort} }\\
&=~~ \tmem{ \termvar_0 }
          {\eapp{\eforall{\typevar_{\spred} : \ttrkind{\sort}}{\tcomp{\termvar}{\ttrtype{\scontext , \spred : \sort}{\sprop}}{\tpred_\spred}{\ttrind{\typevar_\spred}{\sort}}{ \ttrspec{\scontext , \spred : \sort}{\sprop}\termvar}}}{\ttretype{\scontext}{\sterm_{1}}{\sort}} }
          { \ttrtrm{\scontext}{\sterm_{1}}{\sort} }\\
&=_\beta~~\tmem{ \termvar_0 }
          {\tcomp{\termvar}{\ttrtype{\scontext , \spred : \sort}{\sprop}}{\tpred_\spred}{\ttrind{\typevar_\spred}{\sort}}{ \ttrspec{\scontext , \spred : \sort}{\sprop}{\termvar}}\subst{\typevar_{\spred}:={\ttretype{\scontext}{\sterm}{\sort}} }}
          { \ttrtrm{\scontext}{\sterm}{\sort} }\\
\end{align*}
By~\Cref{lem:sub}, $ \ttrspec{\scontext}{\sprop\subst{\spred := \sterm}} {x} :=  \ttrspec{\scontext}{\sprop}{x}\subst{\tpred_{\spred}:=\ttrtrm{\scontext}{\sterm}{\sort}}$.
We prove that the same program $\term$ (which has then the adequate type) also serves as a realizer for the conclusion.
\[
\infer[\monshortrule]{
\elsequent{\trcontext}{\sprops }{
    {\after{\term_0} {\termvar_0}{\tmem{ \termvar }
          {\tcomp{\termvar}{\ttrtype{\scontext , \spred : \sort}{\sprop}}{\tpred_\spred}{\ttrind{\typevar_\spred}{\sort}}{ \ttrspec{\scontext , \spred : \sort}{\sprop}\termvar}\subst{\typevar_{\spred}:={\ttretype{\scontext}{\sterm_{1}}{\sort}} }}
          { \ttrtrm{\scontext}{\sterm_{1}}{\sort} }}}}
}{
{\mathbf{\Pi}}\qquad\qquad
&
\infer={
\elsequent{\trcontext}{\sprops }{
    {\after{\term_0} {\termvar_0}{\ttrspec{\scontext}{ \sprop\subst{\spred := \sterm}} {\termvar_0}}}
    }
    }{IH}
}
\]
with $\mathbf{\Pi}$ is the following subderivation
\[
\infer[\memintrorule]{
    \elsequent{\trcontext,\termvar_0:\ttrtype{\scontext}{ \sprop\subst{\spred := \sterm}}}
        {\sprops,\ttrspec{\scontext}{ \sprop\subst{\spred := \sterm}} {\termvar_0} }
        {\tmem{ \termvar_0 }
          {\tcomp{\termvar}{\ttrtype{\scontext , \spred : \sort}{\sprop}}{\tpred_\spred}{\ttrind{\typevar_\spred}{\sort}}{ \ttrspec{\scontext , \spred : \sort}{\sprop}\termvar}\subst{\typevar_{\spred}:={\ttretype{\scontext}{\sterm}{\sort}} }}
          { \ttrtrm{\scontext}{\sterm}{\sort} } 
        }
    }{
    \infer[\idrule]{
    \elsequent{\trcontext,\termvar_0:\ttrtype{\scontext}{ \sprop\subst{\spred := \sterm}}}
        {\sprops,\ttrspec{\scontext}{ \sprop\subst{\spred := \sterm}} {\termvar_0} }
        {\ttrspec{\scontext , \spred : \sort}{\sprop}{\termvar} 
        \subst{\termvar:=\termvar_0}
        \subst{\tpred_{\spred}:=\ttrtrm{\scontext}{\sterm}{\sort}}
        \subst{\typevar_{\spred}:={\ttretype{\scontext}{\sterm}{\sort}} }}
        % { \termvar_0 }
        %   {\tcomp{\termvar}{\ttrtype{\scontext , \spred : \sort}{\sprop}}{\tpred_\spred}{\ttrind{\typevar_\spred}{\sort}}{ }\subst{\typevar_{\spred}:={\ttretype{\scontext}{\sterm}{\sort}} }}
        %   {  } 
        %}
    }{}
}
\]
using \Cref{lem:sub} to conclude that 
\[\ttrspec{\scontext , \spred : \sort}{\sprop}{\termvar} 
        \subst{\termvar:=\termvar_0}
        \subst{\tpred_{\spred}:=\ttrtrm{\scontext}{\sterm}{\sort}}
        \subst{\typevar_{\spred}:={\ttretype{\scontext}{\sterm}{\sort}} } = \ttrspec{\scontext}{ \sprop\subst{\spred := \sterm}} {\termvar_0} \]

% We show that 
% \[ \htriple{\ttrkind{\scontext} \mid  \ttrtype{}{\sprops}  \mid \ttrind{}{\scontext} }{ \sprops }
% {\termvar} {\term }
% {\ttrspec{\scontext}{\smem{\sterm}{\comprehend{\spred}{\sort}{\sprop}}} {x}} \]
% for $\term:= \term '$.

% By definition, $\ttrspec{\scontext}{\smem{\sterm}{\comprehend{\spred}{\sort}{\sprop}}} {x} := \tmem{ \termvar }{ \eapp{{\color{red}\eforall{\typevar_{\spred}} : \trkind{\sort}}{\tcomp{\termvar}{\trtype{\scontext , \spred : \sort}{\sprop}}{\tpred}{\trind{\trkind{\sort}}{\sort}}{ \ttrspec{\scontext}{\sprop}{x} }}}{\ttretype{\scontext}{\sterm}{\sort}} }{ \ttrtrm{\scontext}{\sterm}{\sort} }$
% \lcnote{the membership rules do note coincide with this quantifier. what am I missing? This should be derivable from mem intro and monotonicity. not sure how to bind}

\item[Case (Mem-Elim)] 
This case, which corresponds to this rule in HOL
       $$\infer[(\text{Mem Elim})]
        {\selsequent{\scontext}{\holhyp}{\sprop\subst{\spred := \sterm}}}
        {\selsequent{\scontext}{\holhyp}{ \smem{\sterm}{\comprehend{\spred}{\sort}{\sprop}}}}
        $$
is analogous to previous one (since, again, the realizer is the same for the premise and the conclusion), by simply using the \memelimrule~ rule instead in $\effhol$ in the derivation.

\item[Cases (Mem$_0$-Intro) and (Mem$_0$-Elim)] Those are analogous  to the cases for (Mem-Intro) and (Mem-Elim).
\qedhere
\end{description}
\end{proof}

% \begin{align*}
%     \trform{\refpredN{\type}} &:= \STAR &
%     \trform{\tmem{\term}{\exprs_1}{\langle \exprs_2\rangle}}
%     &:= \smem{\trform{\exprs_1}}{\trform{\exprs_2}}\\
%     \trform{\refpred{\type}{\indice}} &:= \predcon{\trform \sigma}&
%     \trform{\timplies{ \varphi }{ \psi }}
%     &:= \simplies {\trform{\varphi}} {\trform{\psi}} \\
%     \trform{\textstyle\tindprod{\typevar:\kind}{\indice}}&:=\trform{\indice}&
%     \trform{\after{\term}{\termvar}{\tprop}} &:= \trform{\tprop} \\
%     \trform{\epred} &:= \spred_\epred&
%     \trform{\tspeckinprod{\typevar:\kind}{\tprop}} &:= \trform{\tprop}\\
%     \trform{\ecomp{\termvar}{\type}{\exprsvar}{\indice}{\tprop}}
%     &:=\comprehend{\spred}{\trform \indice}{\trform\tprop}&
%     \trform{\tspectypprod{\termvar:\type}{\tprop}} &:= \trform{\tprop}\\
%     \trform{\eforall{\typevar:\kind}{\exprs}} &:= \trform \exprs&
%     \trform{\tspecindprod{\epred:\indice}{\tprop}}
%     &:= \sforall{\spred:\trform \indice}{\trform \tprop}\\
%     \trform{\eapp{\exprs}{\type}} &:= \trform \exprs\\[0.3cm] & &
% \end{align*}

%First, we note the trivial fact that the translation preserves well-formedness.

~\newpage
 \section{Krivine realizability in the continuation monad}

 We give here more details and missing proof from \Cref{sec:examples}.
 First, we pick a call-by-name evaluation strategy for our programs,
 which is defined by means of adequate evaluation contexts $\progcontext$
 to obtain the following reduction rules:\\
 \begin{mathpar}
\termtypeapp{(\termtypeabs{\typevar}{\kind}{\term})}{\type}\betared^0 \term\subst{\typevar:=\type}

\termapp{(\termabs{\termvar}{\type}{\term})}{\term'}\betared^0 \term\subst{\termvar:=\term'}

\infer{\progcontext[\term_1]\betared^0 \progcontext[\term_2]}{\term_1 \betared^0 \term_2}
\end{mathpar}
where $\progcontext ::= \hole
\mid \termtypeapp{\progcontext}{\type}
\mid \termapp \progcontext  {\term}
\mid \termabs \termvar \type \progcontext$, and we let $\betared$ be the reflexive transitive closure of $\betared^0$.
This indeed defines an evaluation strategy\emnote{note on how this should be prioritized
to actually get an ev. strat.} that obviously contains the axioms for $\betared$ (see~\Cref{sec:operational}).
Recall that constructs of the continuation monads are defined by:
\[
\termret{p}       \eqdef  \termabs{k}{\neg \type}{\termapp{k}{p}}
\qquad\qquad\qquad
\termbind{\termvar}{\term_1}{\term_2} \eqdef  \termabs{k}{\neg \type_2}{\termapp{\term_1}{(\termabs{x}{\type_1}{\termapp{\term_2}{k}})}}
 \]
In particular, the reduction rule for $\termbind{\termvar}{\termret{\term'}}\term$ is,
as expected, a consequence of the definition of these programs
in the continuation monad, and these definitions satisfy the expected typing rules:
 \begin{equation*}
 \scalebox{0.9}{
    \infer{\eltrm{\kcontext \mid \tcontext}{\termret{\term}}{\typecomp{\type}}}{\eltrm{\kcontext \mid \tcontext}{\term}{\type}}
   \qquad\qquad\qquad
    \infer{\eltrm{\kcontext \mid \tcontext}{\termbind{\termvar}{\term_{1}}{\term_{2}}}{\typecomp{\type_{2}}}}{
    %\begin{array}{l}
    \eltrm{\kcontext \mid \tcontext}{\term_{1}}\typecomp{\type_{1}}\quad
    \eltrm{\kcontext \mid \tcontext, \termvar : \type_{1}}{\term_{2}}{\typecomp{\type_{2}}}
    %\end{array}
    }
    }
\end{equation*}

\begin{remark}
Observe that even though we picked a call-by-name evaluation strategy, what we define here is essentially a call-by-value presentation of Krivine realizability, following~\cite{munch09,Lepigre16}. Indeed, to further draw the comparison between the pure instance $\hoplk$ and a CPS translation, the call-by-name evaluation that we pick here only corresponds to the reduction in the target of the CPS translation. In turn, since we use Moggi's monadic encoding of continuations (in what would be the source), the underlying CPS itself rather amounts to a call-by-value source system, and as such we will obtain a call-by-value variant of Krivine realizability. \end{remark}
\let\oldtypecomp\typecomp
\newcommand{\dneg}{{\neg\neg}}
\renewcommand{\typecomp}[1]{{\dneg #1}}

\subsection{Proof of \Cref{lm:hopl_classical}}
\label{app:hopl_classical}
Recall that we define the modality as follows:

\[e^\bot \eqdef \ecompbase{x}{\neg \type}{
    \tspectypprod{x':\type}{
        \timplies
            {(\tmembase{x'}{e})}
            {(\tmembase{\termapp x {x'}}{\pole})}
    }
}%\eqno\text{(for any $e$ of index $\refpredN{\type}$)}
\]
\[
\after{\term}{\termvar}{\tprop}\eqdef \tmembase{\term}{\ecompbase{\termvar}{\type}{\tprop}^{\bot\bot}}
\]

We prove that the following rules are derivable:
\begin{mathpar}
    \infer[\modintrorule]
        {\elsequent{\context}{\tprops}{ \after{\termret{\term}}{\termvar}{\tprop} } }
        {\elsequent{\context}{\tprops}{\tprop\subst{\termvar := \term}}}

    \infer[\modelimrule]
        {\elsequent{\context}{\tprops}{ \after{\left(\termbind{\termvar_{1}}{\term_{1}}{\term_{2}}\right)}{\termvar_{2}}{\tprop} }}
        {\elsequent{\context}{\tprops}{ \after{\term_{1}}{\termvar_{1}}{ \after{\term_{2}}{\termvar_{2}}{\tprop} } }}
        
    \infer[\monrule]{\elsequent{\context}{\tprops }{ \after{\term}{\termvar}{\tprop_{2}} } }{\elsequent{\kcontext \mid \icontext \mid \tcontext, \termvar:\type }{\tprops , \tprop_{1}}{\tprop_{2}} 
         \qquad 
         \elsequent{\context}{\tprops}{ \after{\term}{\termvar}{\tprop_{1}}}}
\end{mathpar}

\begin{proof}~
\begin{itemize}
\item The proof for the \modintrorule~rule is analogous to the usual proof of $A\subset A^{\bot\bot}$
in the context of Krivine realizability. 
Omitting the contexts and unfolding the definition of $\after{\termret{\term}}{\termvar}{\tprop}$, assuming $\tprop\subst{\termvar := \term}$ holds, we want to prove that 
 \[\tmembase{\termabs{k}{\neg\type}{\termapp k p}}{
 \ecompbase{z}{\typecomp\type}{
        \tspectypprod{k:\neg\type}{
        \timplies
            {\tmembase{k}{\ecompbase{\termvar}{\type}{\tprop}^{\bot}}}
            {\tmembase{\termapp z k}{\pole}}
        }
    }
 }
 \]
 Using \memintrorule, it suffices to show
 \[{
 {
        \tspectypprod{k:\neg\type}{
        \timplies
            {\tmembase{k}{\ecompbase{\termvar}{\type}{\tprop}^{\bot}}}
            {\tmembase{\termapp {(\termabs{k}{\neg\type}{\termapp k p})} k}{\pole}}
        }
    }
 }
 \]
 Using \antiredtermrule, we can use anti-reduction and rather prove that
 \[{
 {
        \tspectypprod{k:\neg\type}{
        \timplies
            {\tmembase{k}{\ecompbase{\termvar}{\type}{\tprop}^{\bot}}}
            {\tmembase{{\termapp k p}}{\pole}}
        }
    }
 }
 \]
 Now observe that for any $k:\neg\type  $, $\tmembase{k}{\ecompbase{\termvar}{\type}{\tprop}^{\bot}}$ entails by definition that for any $ \tmembase{q}{\ecompbase{\termvar}{\type}{\tprop}}$, $\tmembase{\termapp k q}{\pole}$.
 So in particular, it holds for $p$ since $\tprop\subst{\termvar := \term}$, which concludes this case.

 \item For the \modelimrule~rule, it follows the same structure relying on anti-reduction. 
Indeed, unfolding the definitions, assuming $\after{\term_1}{\termvar_1}{\after{\term_2}{\termvar_2}{\tprop}}$ holds, we want to prove that 
 \[\tmembase{\termabs{k_2}{\neg\type_2}{\termapp {\term_1} {(\termabs{x_1}{\type_1}{\termapp{p_2}{k_2}})}}}{
 \ecompbase{z}{\typecomp{\type_2}}{
        \tspectypprod{k_2:\neg\type_2}{
        \timplies
            {\tmembase{k_2}{\ecompbase{\termvar_2}{\type_2}{\tprop}^{\bot}}}
            {\tmembase{\termapp z {k_2}}{\pole}}
        }
    }
 }
 \]
 Again using \memintrorule~and \antiredtermrule, it suffices to show that:
 \[{
 {
        \tspectypprod{k_2:\neg\type_2}{
        \timplies
            {\tmembase{k_2}{\ecompbase{\termvar_2}{\type_2}{\tprop}^{\bot}}}
            {\tmembase{{\termapp {\term_1} {(\termabs{x_1}{\type_1}{\termapp{p_2}{k_2}})}}}{\pole}}
        }
    }
 }
 \]
 Now, unfolding the definition of $\after{\term_1}{\termvar_1}{\after{\term_2}{\termvar_2}{\tprop}}$ and using~\memelimrule, we obtain
 \[
 \tspectypprod{k_1:\neg\type_1}{
        \timplies
            {\tmembase{k_1}{\ecompbase{\termvar_1}{\type_1}{\after{\term_2}{\termvar_2}{\tprop}}^{\bot}}}
            {\tmembase{{\termapp {\term_1} {k_1}}}{\pole}}
        }
 \]
 Therefore, to conclude it is enough to show that 
 if $k_2:\neg\type_2$ such that $\tmembase{k_2}{\ecompbase{\termvar_2}{\type_2}{\tprop}^{\bot}}$
 then $\tmembase{\termabs{x_1}{\type_1}{\termapp{p_2}{k_2}}}{
 {\ecompbase{\termvar_1}{\type_1}{\after{\term_2}{\termvar_2}{\tprop}}^{\bot}}
 }$. 
 Unfolding the definitions again, this amounts to proving that for any $q_1:\type_1$
 is such that $\tmembase{q_1}{\ecompbase{\termvar_1}{\type_1}{\after{\term_2}{\termvar_2}{\tprop}}}$ (\textit{i.e.} such that
 $\after{\term_2\subst{\termvar_1:=q_1}}{\termvar_2}{\tprop}$ holds\footnote{Observe here that $x_1$ was bound in $p_2$ and cannot occur in $\tprop$.}), then 
 \[\tmembase{\termapp{(\termabs{x_1}{\type_1}{\termapp{p_2}{k_2}})}{q_1}}{\pole}\]
 By anti-reduction, it suffices to show that $\tmembase{\termapp{(p_2\subst{\termvar_1:=q_1})}{k_2}}{\pole}$. This follows from directly from the hypotheses on $q_1$ and $k_2$.
\item The derivability of the \monrule\ rule, which in terms of sets of realizers, expresses  that $A\subseteq B$ entails $A^{\bot\bot}\subseteq B^{\bot\bot}$, is a direct consequence of the following rule that rather expresses that $A\subseteq B$ entails $B^{\bot}\subseteq A^{\bot}$:
\[
 \infer[]{\elsequent{\context}{\tprops }{\tmembase{p}{\ecompbase{x}{\type}{\tprop_1}^\bot}}}{\elsequent{\kcontext \mid \icontext \mid \tcontext, \termvar:\type }{\tprops , \tprop_{1}}{\tprop_{2}} 
 &
 \elsequent{\context}{\tprops }{\tmembase{p}{\ecompbase{x}{\type}{\tprop_2}^\bot}}
}
\]
To prove that the latter is derivable, unfolding the definition of the orthogonality relation, it suffices to prove from the hypotheses that for any $k:\tau$ such that $\tprop_1\subst{x:=k}$ holds, 
$\tmembase{\termapp p k}{\pole}$. But if $\tprop_1\subst{x:=k}$ holds, then so does $\tprop_2\subst{x:=k}$ using the first hypothesis, while using the assumption on $p$, we indeed get that $\tmembase{\termapp p k}{\pole}$.\qedhere
\end{itemize}
\end{proof}

\subsection{Peirce's law}
 Through the translations, we get:
\[
\begin{array}{rcl}
\ttrtype{}{\mathrm{Peirce}} & = &\typeabs{X}{\star}{\typecomp{\typeabs{Y}{\star}{\typecomp{
    \typefun{(\typefun{(\typefun{X}{\typecomp{Y}})}{\typecomp X})}{\typecomp{X}}
}}}} \\
\ttrspec{}{\mathrm{Peirce}}{p} & = & 
\tspeckinprod{\typevar_a :\star} \forall_{\tpred_a :\refpredN{\typevar_a}} .
    \after{\termtypeapp{\term}{\typevar_a}}{\termvar_{a}}{
        \Big(\big(\tspeckinprod{X_b :\star} \forall_{\tpred_b :\refpredN{\typevar_b}}
            \after{\termtypeapp{\termvar_a}{\typevar_b}}{\termvar_{b}}{
            }    
        }\\
        &&\qquad \tspectypprod{\termvar : \typefun{(\typefun{X_a}{\typecomp{X_b}})}{\typecomp {X_a}} }{ \ttrspec{\scontext}{\simplies{(\simplies{a}{b})}{a}}{\termvar} \supset \after{\termapp{\termvar_b}{\termvar}}{\termvar'}{
            \ttrspec{\scontext}{a}{\termvar'}    
        }}\big)\Big)\\ 
\ttrspec{\scontext}{\simplies{(\simplies{a}{b})}{a}}{\termvar}&=&
\tspectypprod{z : (\typefun{X_a}{\typecomp{X_b}})}{ \ttrspec{\scontext}{(\simplies{a}{b})}{\termvar} \supset \after{\termapp{x}{z}}{z'}{
           \ttrspec{\scontext}{a}{z'}}}  
           \\
\ttrspec{\scontext}{\simplies{a}{b}}{\termvar}&=&
\tspectypprod{z : X_a}{\ttrspec{\scontext}{a}{\termvar} \supset \after{\termapp{x}{z}}{z'}{\ttrspec{\scontext}{b}{z'}}}  \\
\ttrspec{\scontext}{a}{\termvar} &=& \tmembase{x}{y_a}{}{}\\
\end{array}
\]
Recall that we define:
    \[ \begin{array}{rcl}
      \callcc &\eqdef &
        \termtypeabs{X}{\star}{\termret{
            \termtypeabs{Y}{\star}{\termret{
                \callcc^{X,Y}
            }}
        }}\\
    \callcc^{\tau,\tau'} &\eqdef &

               \termabs{z}{\typefun{(\typefun{\type}{\typecomp{\type'}})}{\typecomp \type}}
                    {\termabs{k}{\neg \type}{\termapp{\termapp{z}{\throw{\type,\type'}{k}}}k}}\\
    \throw{\type,\type'}{k}& \eqdef& \termabs{x}{\type}{\termabs{k'}{\neg \type'}{\termapp{k}{x}}}
    \end{array}
  \]

We give here the proof that $\callcc$ defines a valid realizer for Peirce's law.\\
\noindent\textbf{Theorem \ref{thm:callcc}}. We have
\begin{enumerate}
    \item $\eltrm{}{\callcc}{\ttrtype{}{\mathrm{Peirce}}}$
    \item $\elsequent{}{\top}{\after{\termret{\callcc}}x{\ttrspec{}{\mathrm{Peirce}}x}}$
\end{enumerate}

\begin{proof}
\begin{enumerate}
\renewcommand{\typecomp}[1]{\neg\neg{#1}}
    \item
    For any kind context $\kcontext$ and any well-kinded types $\type,\type'$ in this context, we can derive (using implicit weakenings to ease the derivation): 
\begin{small}    
\[
\infer{\eltrm{\kcontext\mid \,}{\lambda z.\lambda k. z\, (\texttt{throw}_k^{\type,\type'})\,k }{\typefun{(\typefun{(\typefun{\type}{\typecomp{{\type'}}})}{\typecomp \type})}{\typecomp{\type}}}}{
    \infer{\eltrm{\kcontext\mid z:\typefun{(\typefun{\type}{\typecomp{{\type'}}})}{\typecomp \type},k:\neg \type }{\termapp{\termapp{z}{\texttt{throw}_k^{\type,\type'}}}{k} }{\bot}}{
        \infer{\eltrm{\kcontext\mid z:\typefun{(\typefun{\type}{\typecomp{{\type'}}})}{\typecomp \type},k:\neg \type }{\termapp{\termapp{z}{\texttt{throw}_k^{\type,\type'}}}}{\typecomp \type}}{
            \infer{\eltrm{\kcontext\mid z:\typefun{(\typefun{\type}{\typecomp{{\type'}}})}{\typecomp \type}}{z}{{(\typefun{\type}{\typecomp{{\type'}}}){\typecomp \type} }}}{}
            &\hspace{-1.5cm}
            \infer{\eltrm{\kcontext\mid k:\neg \type }{\texttt{throw}_k^{\type,\type'}}{\type\to\typecomp {\type'}}}{{
        \infer{\eltrm{\kcontext\mid k:\neg {\type},x:\type,k':\type' }{\termapp k x}{\bot}}{    
            \infer{\eltrm{\kcontext\mid k:\neg {\type},x:\type,k':\type' }{k}{\neg \type}}{}
            &
            \infer{\eltrm{\kcontext\mid k:\neg {\type},x:\type,k':\type' }{x}{\type}}{}
        } 
    }}
        }
        &\hspace{-1.5cm}
        \infer{\eltrm{\kcontext\mid k:\neg \type }{k}{\neg \type}}{}
    }
}   
\]
\end{small}    
It is then straightforward to deduce from this that $\vdash \callcc:\ttrtype{}{\mathrm{Peirce}}$.
\item The proof follows the line of the usual proof of adequacy for \callcc~\cite{Krivine09}. For any types $\type_a,\type_b$ and expressions $\exprs_a:\refpredN{\type_a}$, $\exprs_b:\refpredN{\type_b}$ to inhabit the corresponding quantifications, we first prove that if $k:\neg\type_a$ is such that $\tmembase{k}{\exprs_a^\bot}$, 
then $\tmembase{\throw{\type_a,\type_b}{k}}{\ecompbase{x}{\typefun{\type_a}{\typecomp{\type_b}}}{\ttrspec{\scontext}{\simplies{a}{b}}{\termvar}}}$.
This essentially is a consequence of the closure under reduction. Indeed,  for any $p_a$ such that $\tmembase{p_a}{\exprs_a}$ and any $k'$ such that $\tmembase{k}{\exprs_a^\bot}$, we have \[ \throw{\type_a,\type_b}{k}\,p_a\,k' \betared k\,p_a\]
and the latter is in  $\pole$. 
Then the proof for {\callcc} is straightforward, by again using anti-reduction and the previous fact for $\throw{\tau_a,\tau_b}{k}$ to conclude.
\emnote{TODO : detail a bit more, being precise on how expressions substitute $y_a$ etc.}\qedhere

\end{enumerate}
\end{proof}
\let\typecomp\oldtypecomp

~\newpage
\section{The induced Evidenced Frame}
\label{app:ef}
\label{app:ef}

\begin{definition}[Evidenced Frame]\label{evidenced-frame}
An \emph{evidenced frame} is a triple $( \Phi, E,  \mbox{$\cdot \xle{\cdot} \cdot$} )$, where $\Phi$ is a set of propositions,~$E$ is a collection of evidence, and~\mbox{$\phi_1 \xle{e} \phi_2$} is a evidence relation on $\Phi \times E \times \Phi$,  with:
\begin{description}[leftmargin=!,labelindent=2mm,font={\normalfont\itshape}]
\item[Reflexivity.] There is evidence~$\eid \in E$ such that 
$\forall \phi.\; \phi \xle{\eid} \phi$.
\item[Transitivity.] There is~${\ecompos{}{}} \colon E \times E \to E$ s.t. for all
$\phi_1, \phi_2, \phi_3, e, e'$, $\mbox{$\phi_1 \xle{e} \phi_2 \mathrel{\wedge} \phi_2 \xle{e'} \phi_3 \implies \phi_1 \xle{\ecompos{e}{e'}} \phi_3$}$.
\item[Top.] A proposition~$\!\top \!\!\in\!\! \Phi$ such that there exists evidence~$\etrue \!\!\in\! E$: $\forall \phi.\; \phi \xle{\etrue} \top$.
\item[Conjunction.] An operator~\mbox{${\wedge}: \Phi \!\times\! \Phi \!\to\! \Phi$} such that there exists an operator~\mbox{$\epair{\cdot}{\!\cdot} \!\in\! E \!\times\! E \!\to\! E$} together with evidence~$\efst, \esnd \in E$ satisfying:
 $\forall \phi, \phi_1, \phi_2, e_1, e_2.\; \mbox{$\phi \xle{e_1} \phi_1 \mathrel{\wedge} \phi \xle{e_2} \phi_2 \!\!\implies\!\! \phi \xle{\epair{e_1}{e_2}} \phi_1 \wedge \phi_2$}$,  $\forall \phi_1, \phi_2.\; \phi_1 \wedge \phi_2 \xle{\efst} \phi_1$, and $\forall \phi_1, \phi_2.\; \phi_1 \wedge \phi_2 \xle{\esnd} \phi_2$.

\item[Universal Implication.] An operator~${\imp} : \Phi \times \power(\Phi) \to \Phi$ such that there exists an operator~\hbox{$\elambda{} \!\in\! E \!\to\! E$} and evidence~$\eeval \!\in\! E$ satisfying:
$\forall \phi_1, \phi_2, \vec{\phi}, e.\; (\forall \phi \in \vec{\phi}.\; \phi_1 \wedge \phi_2 \xle{e} \phi) \implies \phi_1 \xle{\elambda{e}} \phi_2 \imp \vec{\phi}$
and $\forall \phi_1, \vec{\phi}, \phi \in \vec{\phi}.\; (\phi_1 \imp \vec{\phi}) \wedge \phi_1 \xle{\eeval} \phi$, where we write $\vec{\phi}$ for an element of $\power(\Phi)$, \ie, a subset of $\Phi$. 
\end{description}
\end{definition}

\begin{comment}
\textbf{Ideas:}
\begin{itemize}
 \item we would like $\phi \xle e \psi$ to be induced
 by $\ehtriple{\tprops}{\termvar}{\term}{\tprop}$,
 natural definition would be propositions= $\hopl$ propositions,
 evidences = $\hopl$ programs
 \item yet universal implication requires intersections of propositions,
 very semantical: propositions = set of its realizers, as usual
 \item reflexivity requires polymorphic realizers $\termtypeabs{X}{\star}{p}$
 \item if so, $\ecompos{\cdot}{\cdot}$ should also produce polymorphic evidences,
 for instance for $\ecompos\eid\eid$, thus the second evidence should be
 provided the output type from the computation resulting from the first,
 which would thus require terms to be able to compute types
 \item then $\efst$ should thus be able to compute the first components
 of a product type, that is, our system should have dependent types
 \item conclusion : we need to erase !!
\end{itemize}
\end{comment}

Note that, while closely related, evidenced frames are not cartesian closed categories (ccc) since they neither impose nor require an equational theory.

Following the standard construction (in an untyped setting) of an evidenced frame / a realizability tripos, one could expect to define a proposition $\phi$ in $\efprop$ as the set of its realizers (obtained through the realizability translation), \ie as a set of closed programs $\vdash p: \ttrtype{}{\sprop}$ such that $\ttrspec{}{\sprop}{p}$. Accordingly, the evidencing relation $\phi \xle e \psi$ should then reflect the triple $\ehtriple{\ttrspec{}{\phi}{p}}{x}{\termapp{e}{p}}{\ttrspec{}{\psi}{x}}$. 

Nonetheless, as mentioned in \Cref{sec:EF}, it turns out that in a typed setting, these definitions do not induce an evidenced frame, let us briefly sketch why. 
Observe first that since the evidence $\eid$ has to be compatible with any proposition, in the presence of typed programs, $\eid$ should be compatible with any type to be well-behaved. One natural way to ensure this internally would be to consider polymorphic evidences. In particular, programs should then come with their types in propositions and applying an evidence should define a computation that also produces a program with its types.  This, in turn, would require computational features on types to compute, say, a product type provided with two types (due to the requirement for the conjunction of propositions) or or the first projection of a product type (for $\efst$), which $\hopl$ cannot handle. 
On the other hand, if the polymorphic status of evidences is handled externally (for instance with evidences defined as families of terms indexed by their input types), similar issues arise. 

In fact, this issues are not really surprising, since it was observed in previous work by Lietz and Streicher that in typed settings such as those coming from modified realizability, one obtains a tripos only if the type system admits a universal type, \ie if the setting is essentially untyped \cite{lietz02}. 
This is why, to overcome this, we will consider a type erasure function $\erase{\cdot}$ and define propositions as erasures of set of values.

Let us consider a pure instance of $\effhol$, we write $\Prog$ for its underlying set of programs which, for simplicity reasons, we assume to include pairs $\tpair{p_1}{p_2}$ of programs together with the corresponding projections $\pi_1$ and $\pi_2$, with the following extra axioms for the reduction $\betared$:
\[
\termvalue  ::= \ldots \mid \tpair{V_1}{V_2} 
\qquad\qquad\qquad
\pi_i\,\tpair{V_1}{V_2} \betared  V_i 
\]
We write $\Lambda$ for the set of untyped terms in the computational $\lambda$-calculus with pairs and we formally define the erasure map $\erase{\cdot}:\Prog\to\Lambda$ by:
\[
\begin{array}{rcl}
\erase{\termvar                             }&\eqdef & \termvar               \\
\erase{\termabs{\termvar}{\type}{\term}     }&\eqdef & \lambda \termvar.\erase{\term} \\
\erase{\termapp{\term_1}{\term_2}           }&\eqdef & \termapp{\erase{\term_1}}{\erase{\term_2}}\\
\erase{\termtypeabs{\typevar}{\kind}{\term} }&\eqdef & \erase{\term}                  \\
\end{array}
\qquad\vrule\qquad
\begin{array}{rcl}
\erase{\termtypeapp{\term}{\type}           }&\eqdef & \erase{\term}            \\
\erase{\termret{\term}                      }&\eqdef & \termret{\erase{\term}}\\
\erase{\termbind{\termvar}{\term_1}{\term_2}}&\eqdef & \termbind{\termvar}{\erase{\term_1}}{\erase{\term_2}}\\
\end{array}
\]
By considering the (meta) set-theoretic counterparts of $\effhol$'s logical constructs (replacing comprehension terms by comprehension, logical membership by membership, and quantification by meta-quantification), the definition of the modality on specifications carries over to sets of (erased) programs. In particular, for any set $A\subseteq \Prog$, we define ${\lift A \eqdef \{p\in{\Prog}\mid  \after{p}{x}{{x}\in A}\}}$, which is carried to set of erased programs by simply letting $\lift{\erase{A}} \eqdef \erase{\lift{A}}$.

As was illustrated in \Cref{sec:examples}, once carried over sets of programs,
the lifting operation $\lift{(\cdot)} $  satisfies properties coming from the theory of $\hopl$, namely the rules \monrule~/ \modintrorule~ / \modelimrule~ / \antiredtermrule:
\begin{property}\label{ppt}
     For any set of values $A,B\subseteq \termvalue$, we have:
\begin{enumerate}
\item \label{ppt:mon} If $A\subset B$, then $\lift A \subset \lift B$.
\item \label{ppt:ret} If $p\in A$, then $\termret{p}\in \lift A$.
\item \label{ppt:bind} If $p_1\in \lift{\scompbase{x}{\Prog}{p_2\subst{x_1:=x}\in\lift{A}}}$, then $\termbind{x_1}{p_1}{p_2}\in \lift A$,
\item \label{ppt:sat} $\lift{A}$ is always closed under anti-reduction.
\end{enumerate}
\end{property}
\newcommand{\reffact}[1]{Fact \ref{#1}}

Finally, recall that the components of the evidenced frame are defined as follows.

\[
\efprop \eqdef \{\erase{\progset}\mid P\subseteq\Prog ~\land~ \forall p\in\progset.\erase{p}\in\termvalue\}
\qquad\qquad\qquad\qquad
\efevd \eqdef \Lambda
\qquad\qquad\qquad\qquad
\phi_1 \xle e \phi_2\eqdef \forall p_1\in\phi_1.e\,p_1\in\lift{\phi_2}\]

% We say that a set $X$ of programs is closed under anti-reduction or saturated whenever for all $p$ and $p'$,
% if $p'\in X$ and $p \betared p'$, then $p\in X$.

% Given any set $X$ of values, we write
% $\sat = \{\erase{\satclosure{V}} \mid V\subseteq \termvalue^\}$.
% \textbf{Definitions}
% \begin{itemize}
% \item $\efprop\eqdef\efprop \eqdef \{\erase{\progset}\mid \forall p\in\progset.\erase{p}\in\termvalue\}$
% \item  $\efevd\eqdef \Lambda$ 
% \item $\phi_1 \xle e \phi_2\eqdef \forall p_1\in\phi_1.e\,p_1\in\lift{\phi_2}$ .
% \end{itemize}
% (where for any $V=\erase{X}\in\sat$, we define $\lift{V}\eqdef \erase{\lift{X}}$) \emnote{sth to check here?}

These components define an evidenced frame for any pure instance of $\effhol$, but more generally, it is valid for any instance that satisfies the previous property.\\

\noindent\textbf{Theorem \ref{thm:ef}}.
For any instance such that Property \ref{ppt} holds,
\emph{$(\efprop,\efevd,\cdot \xle \cdot \cdot)$ is an evidenced frame.}

\begin{proof}
We prove that we have all the necessary constructs.
% In each case, it is an easy exercise to check that the defined evidences are indeed the erasure of a (polymorphic) term in $\instance$. We give full details for the first components to illustrate the use of Fact \ref{ppt:mon}--(4), and only provide the definitions of the core components for conjunction and universal implication, the proof being analogous.
\begin{description}[leftmargin=*]
\item[Reflexivity] We let ~$\eid \eqdef \lambda x.\termret{x}$
and for any $\phi \in\efprop$, and any $\erase{p}=V\in\phi$, we get that
$\termapp{(\lambda x.\termret{x})}{V} \betared \termret{V}$.
Using \reffact{ppt:ret} of the Property~\ref{ppt}, we get that $\termret{V}\in\lift{\phi}$,
and we can thus conclude by anti-reduction that $\termapp{\eid}{V}\in \lift{\phi}$
% \emnote{could also be, even if less natural, for one type X which is irrelevant anyway}
% \begin{align*}
% \termapp{\eid}{p}\in\lift{\phi}
% & \Leftarrow \termret{p}\in\lift{\phi}& \text{by anti-reduction}\\
% & \Leftarrow p\in{\phi}& \text{by \reffact{ppt:ret}}
% \end{align*}

\item[Transitivity]
We define $\ecompos{e_1}{e_2}\eqdef \lambda x_1.\termbind{x_2}{\termapp{e_1}{x}}{\termapp{e_2}{x_2}}$.
%                         = \erase{ \termabs{x_1}{\tau_1}{\termbind{x_2}{\termapp{e_1}{x}}{\termapp{e_2}{x_2}}}}$.
Let then $\phi_1, \phi_2, \phi_3\in\efprop$ and $e_1, e_2\in \efevd$ be such that {$\phi_1 \xle{e_1} \phi_2$} and $\phi_2 \xle{e_2} \phi_3$.
We show that $\phi_1 \xle{\ecompos{e_1}{e_2}} \phi_3$.
Let $V_1\in\phi_1$, and let us prove that
$(\ecompos{e_1}{e_2})\,V_1 \in\lift{\phi_3}$.
By anti-reduction, it suffices to prove that
$\termbind{x_2}{\termapp{e_1}{V_1}}{\termapp{e_2}{x_2}}\in\lift{\phi_3}$.
Using \reffact{ppt:bind}, it suffices to prove that
$e_1\,V_1\in \lift{\scompbase{x_2}{\Prog}{e_2\,x_2\in\lift{\phi_3}}}$.
Now, since by assumption on $e_2$, we have $e_2\,x_2\in \lift{\phi_3}$  for any $x_2\in\phi_2$, using \reffact{ppt:mon} we get that
$\lift{\phi_2}\subset \lift{\scompbase{x_2}{\Prog}{e_2\,x_2\in\lift{\phi_3}}}$.
We conclude by using the hypothesis on $e_1$ to get that $e_1\,V_1 \in\lift{\phi_2}$.~\\

\item[Top] Any non-empty set $\psi\in\efprop$ with a value $V\in\psi$
(for instance, with $V\eqdef\lambda x.\termret{x}$)
allows us to define $\!\top \eqdef \psi$ and $\etrue \eqdef \lambda x.\termret{V}$,
it is then direct to see that for any $\phi\in\efprop$,
$\phi \xle \etrue \top$.\\

\item[Conjunction]
For any $\phi_1,\phi_2 \in \efprop$,
we define $\phi_1\wedge \phi_2 \eqdef \{\erase{\tpair{p_1}{p_2}} \mid \erase{p_i}\in\phi_i\}$.
As for evidences, for any $e_1,e_2\in E$, we define $\efst\eqdef\lambda x.\termret{\pi_1\,x}$,
 $\esnd\eqdef\lambda x.\termret{\pi_2\,x}$
 and $\epair{e_1}{e_2}\eqdef \lambda x.\termbind{x_1}{e_1\,x}{\termbind{x_2}{e_2\,x}{\termret{\tpair{x_1}{x_2}}}}$.

Then for any $\tpair{V_1}{V_2}\in\phi_1\wedge\phi_2$, to prove that
$\termapp{\efst}{\tpair{V_1}{V_2}}\in\lift{\phi_1}$, by anti-reduction it suffices to show that
$\termret{\pi_1\,\tpair{V_1}{V_2}}\in\lift{\phi_1}$.
By \reffact{ppt:ret}, this follows from the fact that
${\pi_1\,\tpair{V_1}{V_2}}\in{\phi_1}$, which holds by anti-reduction since
$\pi_1\,\tpair{V_1}{V_2}\betared V_1\in\phi_1$.
The proof for $\esnd$ is analogous.

Last, if $\phi \xle{e_1} \phi_1$ and  $\phi \xle{e_2} \phi_2$, let us prove
that $\phi \xle{\epair{e_1}{e_2}} \phi_1 \wedge \phi_2$.
We have for any $V\in\phi$:
\[
\epair{e_1}{e_2}\,V\betared \termbind{x_1}{e_1\,V}{\termbind{x_2}{e_2\,V}{\termret{\tpair{x_1}{x_2}}}}
\]
Using anti-reduction and \reffact{ppt:bind}, it suffices to show that%\emnote{this was giving an error so I commented out}%$\phi_i\in\Sat_{\tau_i}$}
\[e_1\,V\in
\lift{
  \scompbase{x_1}{\Prog}
  {\termbind{x_2}{e_2\,V}{\termret{\tpair{x_1}{x_2}}}
      \in\lift{\phi_1\wedge\phi_2}
  }
}
\]
By hypothesis, we know that $e_1\,V\in\lift{\phi_1}$, hence using \reffact{ppt:mon},
it is enough to prove that $\phi_1\subset \scompbase{x_1}{\Prog}
  {\termbind{x_2}{e_2\,V}{\termret{\tpair{x_1}{x_2}}}
      \in\lift{\phi_1\wedge\phi_2}
  }$.
This amounts to showing that $\termbind{x_2}{e_2\,V}{\termret{\tpair{x_1}{x_2}}}
      \in\lift{\phi_1\wedge\phi_2}$ for any $x_1\in\phi_1$, which, using \reffact{ppt:bind} again, follows from
\[
   e_2\,V\in \lift{  \scompbase{x_2}{\Prog}
  {\termret{\tpair{x_1}{x_2}} \in\lift{\phi_1\wedge\phi_2}
  }
  }
\]
Using that $e_2\,V\in\lift{\phi_2}$ and \reffact{ppt:mon} again, this follows from the fact that
for any $x_2\in\phi_2$, $\termret{\tpair{x_1}{x_2}} \in\lift{\phi_1\wedge\phi_2}$,
which holds from \reffact{ppt:ret} and the definition of $\phi_1\wedge\phi_2$.
\\

\item[Universal Implication]
We define the universal implication by
$\phi_1\imp \vec{\phi} \eqdef\{\lambda x.p \mid \forall V_1\in\phi_1.\forall \phi\in\vec{\phi}.p[x:=V_1] \in\lift{\phi}\}$\emnote{explain how we can be sure this is the erasure of some set of typed terms...}. As for the corresponding evidences, we define
$\elambda{e}\eqdef \lambda x_1.\termret{\lambda x_2.e\,\tpair{x_1}{x_2}}$ and
$\eeval\eqdef \lambda x.(\pi_1 x)\,(\pi_2\,x)$.
Proofs of the expected properties are then standard, and follow the line of the corresponding cases for the soundness proof.
To prove that $\eeval$ satisfies the expected properties, let 
$V_1\in\phi_1$, $\lambda x.p \in (\phi_1 \imp \vec{\phi})$ and $\phi\in\vec{\phi}$. 
We have 
\[\eeval\,\tpair{\lambda x.p}{V_1} = (\lambda x.(\pi_1 x)\,(\pi_2\,x))\,\tpair{\lambda x.p}{V_1}\betared \termapp{(\lambda x.p)}{V_1} \betared p[x:=V_1]\] 
By definition of $(\phi_1 \imp \vec{\phi})$, we know that $p[x:=V_1]\in\lift{\phi}$ and we can thus conclude that $(\phi_1 \imp \vec{\phi}) \wedge \phi_1 \xle{\eeval} \phi$ by anti-reduction~(\reffact{ppt:sat}).

As for $\elambda{-}$, let $\phi_1,\phi_2,\vec{\phi}$ and $e$ be such that 
$(\forall \phi \in \vec{\phi}.\; \phi_1 \wedge \phi_2 \xle{e} \phi)$. 
Let then $V_1\in\phi_1$. We need to show that 
$\elambda{e}\,V_1\in \lift{\phi_2 \imp \vec{\phi}} $.
We have:
\[\elambda{e}\,V_1 = (\lambda x_1.\termret{\lambda x_2.e\,\tpair{x_1}{x_2}})\,V_1 \betared
\termret{\lambda x_2.e\,\tpair{V_1}{x_2}}\]
hence using Facts \ref{ppt:ret} and \ref{ppt:sat} it suffices to show that 
$ \lambda x_2.e\,\tpair{V_1}{x_2}\in\phi_2 \imp \vec{\phi}$. 
This follows directly from the hypothesis on~$e$, since for any $V_2 \in \phi_2$ and any $\phi\in\vec{\phi}$,
we have $(\lambda x_2.e\,\tpair{V_1}{x_2})\,V_2 \betared e\,\tpair{V_1}{V_2}$. Since $\tpair{V_1}{V_2}\in\phi_1\land \phi_2$, the latter computation is indeed in $\lift{\phi}$.

\end{description}
\end{proof}

%\input{section/ef_proof2}

% \subsubsection*{SEMANTIC ATTEMPT}
% \input{section/semantics}

% \input{section/semModel}

% ============OLD STUFF==========

% \input{section/firstorder}
% % \input{section/higherorder}
% \section{OLD STUFF TO REMOVE}
% \input{section/OldpairExm}
% ~\newpage
% \section{Countable Choice via Memoization}
% \input{section/CC}
\end{appendices}
}

\end{document}